\documentclass[11pt]{article}%
\usepackage{amssymb}
\usepackage{graphicx}
\usepackage{amsmath}
\usepackage{lscape}
\usepackage{ulem}%
\usepackage{amsfonts}
\providecommand{\U}[1]{\protect\rule{.1in}{.1in}}
\newtheorem{theorem}{Theorem}
\newtheorem{assumption}{Assumption}

\newtheorem{lemma}{Lemma}

\newtheorem{proposition}{Proposition}

\newenvironment{proof}[1][Proof]{\textbf{#1.} }{\  \rule{0.5em}{0.5em}}
\renewcommand{\baselinestretch}{1.25}

\newcommand{\bei}{\begin{itemize}}
\newcommand{\ei}{\end{itemize}}

\begin{document}

\title{Dynamic Semiparametric Models for \linebreak Expected Shortfall (and
Value-at-Risk)\thanks{For helpful comments we\ thank Tim Bollerslev, Rob
Engle, Jia Li, Nour Meddahi, and seminar participants at the Bank of Japan,
Duke University, EPFL, Federal Reserve Bank of New\ York, Hitotsubashi
University, New York University, Toulouse School of Economics, the University
of Southern California, and the 2015 Oberwolfach Workshop on Quantitative Risk
Management where this project started. The first author would particularly
like to thank the finance department at NYU\ Stern, where much of his work on
this paper was completed. Contact address:\ Andrew Patton, Department of
Economics, Duke University, 213 Social Sciences Building, Box 90097, Durham NC
27708-0097. Email: \texttt{andrew.patton@duke.edu}. }}
\author{Andrew J. Patton\\Duke University
\and Johanna F. Ziegel\\University of Bern
\and Rui Chen\\Duke University{}}
\date{First version:\ 5 December 2015. This version: 11 July 2017.}
\maketitle

\begin{center}
\textbf{Abstract }
\end{center}

\bigskip

Expected Shortfall (ES) is the average return on a risky asset conditional on
the return being below some quantile of its distribution, namely its
Value-at-Risk (VaR). The Basel III Accord, which will be implemented in the
years leading up to 2019, places new attention on ES, but unlike VaR, there is
little existing work on modeling ES. We use recent results from statistical
decision theory to overcome the problem of \textquotedblleft
elicitability\textquotedblright\ for ES by \textit{jointly }modelling ES
and\ VaR, and propose new dynamic models for these risk measures. We provide
estimation and inference methods for the proposed models, and confirm via
simulation studies that the methods have good finite-sample properties. We
apply these models to daily returns on four international equity indices, and
find the proposed new ES-VaR models outperform forecasts based on GARCH or
rolling window models.

\bigskip

\vfill\textbf{Keywords: }Risk management, tails, crashes, forecasting,
generalized autoregressive score.

\textbf{J.E.L. codes:} G17, C22, G32, C58. \ \pagebreak%

\def\baselinestretch{1.65}\small\normalsize

\section{Introduction}

The financial crisis of 2007-08 and its aftermath led to numerous changes in
financial market regulation and banking supervision. One important change
appears in the Third Basel Accord (Basel Committee, 2010), where new emphasis
is placed on \textquotedblleft Expected Shortfall\textquotedblright\ (ES)\ as
a measure of risk, complementing, and in parts substituting, the more-familiar
Value-at-Risk (VaR) measure. Expected Shortfall is the expected return on an
asset conditional on the return being below a given quantile of its
distribution, namely its\ VaR. That is, if $Y_{t}$ is the return on some asset
over some horizon (e.g., one day or one week) with conditional (on information
set $\mathcal{F}_{t-1}$) distribution $F_{t}$, which we assume to be strictly
increasing with finite mean, the $\alpha$-level VaR and ES are:%
\begin{align}
\mathrm{ES}_{t}  & =\mathbb{E}\left[  Y_{t}|Y_{t}\leq\mathrm{VaR}%
_{t},\mathcal{F}_{t-1}\right] \label{eqnES}\\
\text{where \ \ }\mathrm{VaR}_{t}  & =F_{t}^{-1}\left(  \alpha\right)  \text{,
\ for }\alpha\in\left(  0,1\right) \label{eqnVaR}\\
\text{and \ \ }Y_{t}|\mathcal{F}_{t-1}  & \thicksim F_{t}%
\end{align}
\qquad As Basel III\ is implemented worldwide (implementation is expected to
occur in the period leading up to January 1$^{\text{st}}$, 2019), ES will
inevitably gain, and require, increasing attention from risk managers and
banking supervisors and regulators. The new \textquotedblleft market
discipline\textquotedblright\ aspects of Basel III mean that ES and VaR will
be regularly disclosed by banks, and so a knowledge of these measures will
also likely be of interest to these banks' investors and counter-parties.

There is, however, a paucity of empirical models for expected shortfall. The
large literature on volatility models (see Andersen \textit{et al.} (2006) for
a review) and VaR models (see Komunjer (2013) and McNeil et al. (2015)), have
provided many useful models for these measures of risk. However, while ES has
long been known to be a \textquotedblleft coherent\textquotedblright\ measure
of risk (Artzner, \textit{et al.} 1999), in contrast with VaR, the literature
contains relatively few models for ES; some exceptions are discussed below.
This dearth is perhaps in part because regulatory interest in this risk
measure is only recent, and perhaps also due to the fact that this measure is
not \textquotedblleft elicitable.\textquotedblright\ A risk measure (or
statistical functional more generally) is said to be \textquotedblleft
elicitable\textquotedblright\ if there exists a loss function such that the
measure is the solution to minimizing the expected loss. For example, the mean
is elicitable using the quadratic loss function, and VaR is elicitable using
the piecewise-linear or \textquotedblleft tick\textquotedblright\ loss
function. Having such a loss function is a stepping stone to building dynamic
models for these quantities. We use recent results from Fissler and
Ziegel\ (2016), who show that ES is \textit{jointly} \textit{elicitable}
with\ VaR, to build new dynamic models for ES and VaR.

This paper makes three main contributions. Firstly, we present some novel
dynamic models for ES and\ VaR, drawing on the GAS framework of\ Creal,
\textit{et al. }(2013), as well as successful models from the volatility
literature, see Andersen \textit{et al.} (2006). The models we propose are
semiparametric in that they impose parametric structures for the dynamics of
ES and VaR, but are completely agnostic about the conditional distribution of
returns (aside from regularity conditions required for estimation and
inference). The models proposed in this paper are related to the class of
\textquotedblleft CAViaR\textquotedblright\ models proposed by Engle and
Manganelli (2004a), in that we directly parameterize the measure(s) of risk
that are of interest, and avoid the need to specify a conditional distribution
for returns. The models we consider make estimation and prediction fast and
simple to implement. Our semiparametric approach eliminates the need to
specify and estimate a conditional density, thereby removing the possibility
that such a model is misspecified, though at a cost of a loss of efficiency
compared with a correctly specified density model.

Our second contribution is asymptotic theory for a general class of dynamic
semiparametric models for ES and\ VaR. This theory is an extension of results
for VaR presented in Weiss (1991) and Engle and\ Manganelli\ (2004a), and
draws on identification results in\ Fissler and\ Ziegel\ (2016) and results
for M-estimators in Newey and\ McFadden (1994). We present conditions under
which the estimated parameters of the VaR and\ ES models are consistent and
asymptotically normal, and we present a consistent estimator of the asymptotic
covariance matrix. We show via an extensive Monte Carlo study that the
asymptotic results provide reasonable approximations in realistic simulation
designs. In addition to being useful for the new models we propose, the
asymptotic theory we present provides a general framework for other
researchers to develop, estimate, and evaluate new models for VaR and ES.

Our third contribution is an extensive application of our new models and
estimation methods in an out-of-sample analysis of forecasts of ES and\ VaR
for four international equity indices over the period January 1990 to December
2016. We compare these new models with existing methods from the literature
across a range of tail probability values $\left(  \alpha\right)  $ used in
risk management. We use Diebold and Mariano (1995) tests to identify the
best-performing models for ES and VaR, and we present simple regression-based
methods, related to those of Engle and\ Manganelli (2004a) and Nolde
and\ Ziegel\ (2017), to \textquotedblleft backtest\textquotedblright\ the ES forecasts.

Some work on expected shortfall estimation and prediction has appeared in the
literature, overcoming the problem of elicitability in different ways: Engle
and\ Manganelli (2004b) discuss using extreme value theory, combined with
GARCH or\ CAViaR dynamics, to obtain forecasts of ES. Cai and\ Wang (2008)
propose estimating VaR and\ ES based on nonparametric conditional
distributions, while Taylor (2008) and Gsch\"{o}pf \textit{et al.} (2015)
estimate models for \textquotedblleft expectiles\textquotedblright\ (Newey
and\ Powell, 1987) and map these to ES. Zhu and Galbraith (2011) propose using
flexible parametric distributions for the standardized residuals from models
for the conditional mean and variance. Drawing on Fissler and\ Ziegel (2016),
we overcome the problem of elicitability more directly, and open up new
directions for ES modeling and prediction.

In recent independent work, Taylor (2017) proposes using the asymmetric
Laplace distribution to jointly estimate dynamic models for VaR and\ ES. He
shows the intriguing result that the negative log-likelihood of this
distribution corresponds to one of the loss functions presented in Fissler and
Ziegel\ (2016), and thus can be used to estimate and evaluate such models.
Unlike our paper, Taylor (2017) provides no asymptotic theory for his proposed
estimation method, nor any simulation studies of its reliability. However,
given the link he presents, the theoretical results we present below can be
used to justify \textit{ex post} the methods of his paper.

The remainder of the paper is structured as follows. In Section \ref{sMODELS}
we present new dynamic semiparametric models for ES and VaR and compare them
with the main existing models for ES and VaR. In Section \ref{sTHEORY} we
present asymptotic distribution theory for a generic dynamic semiparametric
model for ES and VaR, and in Section \ref{sSIMULATION} we study the
finite-sample properties of the asymptotic theory in some realistic Monte
Carlo designs. Section \ref{sAPPLICATION} we apply the new models to daily
data on four international equity indices, and compare these models both
in-sample and out-of-sample with existing models. Section \ref{sCONCLUSION}
concludes. Proofs and additional technical details are presented in the
appendix, and a supplemental web appendix contains detailed proofs and
additional analyses.

\section{\label{sMODELS}Dynamic models for ES and VaR}

In this section we propose some new dynamic models for expected shortfall (ES)
and Value-at-Risk (VaR). We do so by exploiting recent work in Fissler and
Ziegel\ (2016) which shows that these variables are elicitable
\textit{jointly}, despite the fact that ES was known to be not elicitable
separately, see Gneiting (2011a). The models we propose are based on the GAS
framework of Creal, \textit{et al.} (2013) and\ Harvey (2013), which we
briefly review in Section \ref{sGAS} below.

\subsection{\label{sFZ0}A consistent scoring rule for ES and VaR}

Fissler and Ziegel\ (2016) show that the following class of loss functions (or
\textquotedblleft scoring rules\textquotedblright),\ indexed by the functions
$G_{1}$ and $G_{2},$ is consistent for VaR and\ ES. That is, minimizing the
expected loss using any of these loss functions returns the true VaR and ES.
In the functions below, we use the notation $v$ and $e$ for VaR and ES.%
\begin{align}
L_{FZ}\left(  Y,v,e;\alpha,G_{1},G_{2}\right)   & =\left(  \mathbf{1}\left\{
Y\leq v\right\}  -\alpha\right)  \left(  G_{1}\left(  v\right)  -G_{1}\left(
Y\right)  +\frac{1}{\alpha}G_{2}\left(  e\right)  v\right) \label{eqnFZloss}\\
& -G_{2}\left(  e\right)  \left(  \frac{1}{\alpha}\mathbf{1}\left\{  Y\leq
v\right\}  Y-e\right)  -\mathcal{G}_{2}\left(  e\right) \nonumber
\end{align}
where $G_{1}$ is weakly increasing, $G_{2}$ is strictly increasing and
strictly positive, and $\mathcal{G}_{2}^{\prime}=G_{2}.$ We will refer to the
above class as \textquotedblleft FZ loss functions.\textquotedblright%
\footnote{Consistency of the FZ loss function for VaR and\ ES also requires
imposing that $e\leq v,$ which follows naturally from the definitions of ES
and\ VaR in equations (1) and (2). We discuss how we impose this restriction
empirically in Sections \ref{sSIMULATION} and \ref{sAPPLICATION}
below.}\ Minimizing any member of this class yields VaR\ and ES:%
\begin{equation}
\left(  \mathrm{VaR}_{t},\mathrm{ES}_{t}\right)  =\arg\min_{\left(
v,e\right)  }~\mathbb{E}_{t-1}\left[  L_{FZ}\left(  Y_{t},v,e;\alpha
,G_{1},G_{2}\right)  \right]
\end{equation}

Using the FZ loss function for estimation and forecast evaluation requires
choosing $G_{1}$ and $G_{2}.$ We choose these so that the loss function
generates loss differences (between competing forecasts) that are homogeneous
of degree zero. This property has been shown in volatility forecasting
applications to lead to higher power in Diebold-Mariano (1995) tests in Patton
and Sheppard (2009). Nolde and Ziegel (2017) show that there does \textit{not
}generally exist an FZ loss function that generates loss differences that are
homogeneous of degree zero. However, zero-degree homogeneity may be attained
by exploiting the fact that, for the values of $\alpha$ that are of interest
in risk management applications (namely, values ranging from around 0.01 to
0.10), we may assume that $\mathrm{ES}_{t}<0$ a.s. $\forall~t.$ The following
proposition shows that if we further impose that $\mathrm{VaR}_{t}<0~$a.s.
$\forall~t,$ then, up to irrelevant location and scale factors, there is only
\textit{one} FZ loss function that generates loss differences that are
homogeneous of degree zero.\footnote{If VaR can be positive, then there is one
free shape parameter in the class of zero-homogeneous FZ loss functions
($\varphi_{1}/\varphi_{2},$ in the notation of the proof of Proposition
\ref{propFZ0}). In that case, our use of the loss function in equation
(\ref{eqnFZ0}) can be interpreted as setting that shape parameter to zero.
This shape parameter does not affect the consistency of the loss function, as
it is a member of the FZ class, but it may affect the ranking of misspecified
models, see Patton (2016).} The fact that the $L_{FZ0}$\ loss function defined
below is unique has the added benefit that there are, of course, no remaining
shape or tuning parameters to be specified.

\begin{proposition}
\label{propFZ0}Define the FZ loss difference for two forecasts $\left(
v_{1t},e_{1t}\right)  $ and $\left(  v_{2t},e_{2t}\right)  $ as \linebreak%
$L_{FZ}\left(  Y_{t},v_{1t},e_{1t};\alpha,G_{1},G_{2}\right)  -L_{FZ}\left(
Y_{t},v_{2t},e_{2t};\alpha,G_{1},G_{2}\right)  .$ Under the assumption that
VaR and ES are both strictly negative, the loss differences generated by a FZ
loss function are homogeneous of degree zero iff $G_{1}(x)=0$ and
$G_{2}(x)=1/x.$ The resulting \textquotedblleft FZ0\textquotedblright\ loss
function is:%
\begin{equation}
L_{FZ0}\left(  Y,v,e;\alpha\right)  =-\frac{1}{\alpha e}\mathbf{1}\left\{
Y\leq v\right\}  \left(  v-Y\right)  +\frac{v}{e}+\log\left(  -e\right)
-1\label{eqnFZ0}%
\end{equation}

\end{proposition}

All proofs are presented in Appendix A. In Figure \ref{figFZ0loss} we plot
$L_{FZ0}$ when $Y=-1.$ In the left panel we fix $e=-2.06$ and vary $v,$ and in
the right panel we fix $v=-1.64$ and vary $e.$ (These values for $\left(
v,e\right)  $ are the $\alpha=0.05$ VaR and ES from a standard Normal
distribution.) As neither of these are the complete loss function, the minimum
is not zero in either panel. The left panel shows that the implied VaR loss
function resembles the \textquotedblleft tick\textquotedblright\ loss function
from quantile estimation, see Komunjer (2005) for example. In the right panel
we see that the implied ES loss function resembles the \textquotedblleft
QLIKE\textquotedblright\ loss function from volatility forecasting, see Patton
(2011) for example. In both panels, values of $\left(  v,e\right)  $ where
$v<e$ are presented with a dashed line, as by definition $\mathrm{ES}_{t}$ is
below $\mathrm{VaR}_{t},$ and so such values that would never be considered in
practice. In Figure \ref{figFZ0contour} we plot the contours of expected FZ0
loss for a standard Normal random variable. The minimum value, which is
attained when $\left(  v,e\right)  =\left(  -1.64,-2.06\right)  $, is marked
with a star, and we see that the \textquotedblleft iso-expected
loss\textquotedblright\ contours are convex. Fissler (2017) shows that
convexity of iso-expected loss contours holds more generally for the FZ0 loss
function under any distribution with finite first moments, unique $\alpha
$-quantiles, continuous densities, and negative ES.

\bigskip

[ INSERT FIGURES \ref{figFZ0loss} AND\ \ref{figFZ0contour} ABOUT HERE ]

\bigskip

With the FZ0 loss function in hand, it is then possible to consider
semiparametric dynamic models for ES and\ VaR:%
\begin{equation}
\left(  \mathrm{VaR}_{t},\mathrm{ES}_{t}\right)  =\left(  v\left(
\mathbf{Z}_{t-1};\mathbf{\theta}\right)  ,e\left(  \mathbf{Z}_{t-1}%
;\mathbf{\theta}\right)  \right)
\end{equation}
that is, where the true VaR and ES are some specified parametric functions of
elements of the information set, $\mathbf{Z}_{t-1}\in\mathcal{F}_{t-1}.$ The
parameters of this model are estimated via:%
\begin{equation}
\mathbf{\hat{\theta}}_{T}=\arg\min_{\mathbf{\theta}}~\frac{1}{T}%
{\displaystyle\sum\nolimits_{t=1}^{T}}
L_{FZ0}\left(  Y_{t},v\left(  \mathbf{Z}_{t-1};\mathbf{\theta}\right)
,e\left(  \mathbf{Z}_{t-1};\mathbf{\theta}\right)  ;\alpha\right)
\end{equation}
Such models impose a parametric structure on the dynamics of VaR and\ ES,
through their relationship with lagged information, but require no
assumptions, beyond regularity conditions, on the conditional distribution of
returns. In this sense, these models are semiparametric. Using theory for
M-estimators (see White (1994) and Newey and McFadden (1994) for example) we
establish in Section \ref{sTHEORY} below the asymptotic properties of such
estimators. Before doing so, we first consider some new dynamic specifications
for\ ES and\ VaR.

\subsection{\label{sGAS}A GAS model for ES and VaR}

One of the challenges in specifying a dynamic model for a risk measure, or any
other quantity of interest, is the mapping from lagged information to the
current value of the variable. Our first proposed specification for ES
and\ VaR draws on the work of Creal, \textit{et al.} (2013) and\ Harvey
(2013), who proposed a general class of models called \textquotedblleft
generalized autoregressive score\textquotedblright\ (GAS)\ models by the
former authors, and \textquotedblleft dynamic conditional
score\textquotedblright\ models by the latter author. In both cases the models
start from an assumption that the target variable has some parametric
conditional distribution, where the parameter (vector) of that distribution
follows a GARCH-like equation. The forcing variable in the model is the lagged
score of the log-likelihood, scaled by some positive definite matrix, a common
choice for which is the inverse Hessian. This specification nests many well
known models, including ARMA, GARCH (Bollerslev, 1986) and ACD (Engle and
Russell, 1998) models. See Koopman \textit{et al.} (2016) for an overview of
GAS\ and related models.

We adopt this modeling approach and apply it to our M-estimation problem. In
this application, the forcing variable is a function of the derivative and
Hessian of the $L_{FZ0}$ loss function rather than a log-likelihood. We will
consider the following GAS(1,1) model for ES and VaR:%
\begin{equation}
\left[
\begin{array}
[c]{c}%
v_{t+1}\\
e_{t+1}%
\end{array}
\right]  =\mathbf{w+B}\left[
\begin{array}
[c]{c}%
v_{t}\\
e_{t}%
\end{array}
\right]  +\mathbf{AH}_{t}^{-1}\mathbf{\nabla}_{t}%
\end{equation}
where $\mathbf{w}$ is a $\left(  2\times1\right)  $ vector and $\mathbf{B}$
and $\mathbf{A}$ are $\left(  2\times2\right)  $ matrices. The forcing
variable in this specification is comprised of two components, the first is
the score:\footnote{Note that the expression given for $\partial
L_{FZ0}/\partial v_{t}$ only holds for $Y_{t}\neq v_{t}.$ As we assume that
$Y_{t}$ is continuously distributed, this holds with probability one.}%
\begin{align}
\mathbf{\nabla}_{t}  & \equiv\left[
\begin{array}
[c]{c}%
\partial L_{FZ0}\left(  Y_{t},v_{t},e_{t};\alpha\right)  /\partial v_{t}\\
\partial L_{FZ0}\left(  Y_{t},v_{t},e_{t};\alpha\right)  /\partial e_{t}%
\end{array}
\right]  =\left[
\begin{array}
[c]{c}%
\frac{1}{\alpha v_{t}e_{t}}\lambda_{v,t}\\
\frac{-1}{\alpha e_{t}^{2}}\left(  \lambda_{v,t}+\alpha\lambda_{e,t}\right)
\end{array}
\right] \\
\text{where \ \ }\lambda_{v,t}  & \equiv-v_{t}\left(  \mathbf{1}\left\{
Y_{t}\leq v_{t}\right\}  -\alpha\right) \label{eqnLAMv}\\
\lambda_{e,t}  & \equiv\frac{1}{\alpha}\mathbf{1}\left\{  Y_{t}\leq
v_{t}\right\}  Y_{t}-e_{t}\label{eqnLAMe}%
\end{align}
The scaling matrix, $\mathbf{H}_{t},$ is related to the Hessian:%
\begin{equation}
\mathbf{I}_{t}\equiv\left[
\begin{array}
[c]{cc}%
\frac{\partial^{2}\mathbb{E}_{t-1}\left[  L_{FZ0}\left(  Y_{t},v_{t}%
,e_{t}\right)  \right]  }{\partial v_{t}^{2}} & \frac{\partial^{2}%
\mathbb{E}_{t-1}\left[  L_{FZ0}\left(  Y_{t},v_{t},e_{t}\right)  \right]
}{\partial v_{t}\partial e_{t}}\\
\bullet & \frac{\partial^{2}\mathbb{E}_{t-1}\left[  L_{FZ0}\left(  Y_{t}%
,v_{t},e_{t}\right)  \right]  }{\partial e_{t}^{2}}%
\end{array}
\right]  =\left[
\begin{array}
[c]{cc}%
-\frac{f_{t}\left(  v_{t}\right)  }{\alpha e_{t}} & 0\\
0 & \frac{1}{e_{t}^{2}}%
\end{array}
\right]
\end{equation}
The second equality above exploits the fact that $\partial^{2}\mathbb{E}%
_{t-1}\left[  L_{FZ0}\left(  Y_{t},v_{t},e_{t};\alpha\right)  \right]
/\partial v_{t}\partial e_{t}=0$ under the assumption that the dynamics for
VaR and\ ES are correctly specified. The first element of the matrix
$\mathbf{I}_{t}$ depends on the unknown conditional density of $Y_{t}.$ We
would like to avoid estimating this density, and we approximate the term
$f_{t}\left(  v_{t}\right)  $ as being proportional to $v_{t}^{-1}.$ This
approximation holds exactly if $Y_{t}$ is a zero-mean location-scale random
variable, $Y_{t}=\sigma_{t}\eta_{t}$, where\ $\eta_{t}\thicksim iid~F_{\eta
}\left(  0,1\right)  ,$ as in that case we have:%
\begin{equation}
f_{t}\left(  v_{t}\right)  =f_{t}\left(  \sigma_{t}v_{\alpha}\right)
=\frac{1}{\sigma_{t}}f_{\eta}\left(  v_{\alpha}\right)  \equiv k_{\alpha}%
\frac{1}{v_{t}}%
\end{equation}
where $k_{\alpha}\equiv v_{\alpha}f_{\eta}\left(  v_{\alpha}\right)  $ is a
constant with the same sign as $v_{t}$. We define $\mathbf{H}_{t}$ to equal
$\mathbf{I}_{t}$ with the first element replaced using the approximation in
the above equation.\footnote{Note that we do \textit{not }use the fact that
the scaling matrix is exactly the inverse Hessian (e.g., by invoking the
information matrix equality) in our empirical application or our theoretical
analysis. Also, note that if we considered a value of $\alpha$ for which
$v_{t}=0,$ then $v_{\alpha}=0$ and we cannot justify our approximation using
this approach. However, we focus on cases where $\alpha\ll1/2,$ and so we are
comfortable assuming $v_{t}\neq0,$ making $k_{\alpha}$ invertible.} The
forcing variable in our GAS model for VaR and\ ES then becomes:%
\begin{equation}
\mathbf{H}_{t}^{-1}\mathbf{\nabla}_{t}=\left[
\begin{array}
[c]{c}%
\frac{-1}{k_{\alpha}}\lambda_{v,t}\\
\frac{-1}{\alpha}\left(  \lambda_{v,t}+\alpha\lambda_{e,t}\right)
\end{array}
\right]
\end{equation}
Notice that the second term in the model is a linear combination of the two
elements of the forcing variable, and since the forcing variable is
premultiplied by a coefficient matrix, say $\mathbf{\tilde{A}},$ we can
equivalently use%
\begin{align}
\mathbf{\tilde{A}H}_{t}^{-1}\mathbf{\nabla}_{t}  & =\mathbf{A\lambda}_{t}\\
\text{where \ }\mathbf{\lambda}_{t}  & \equiv\left[  \lambda_{v,t}%
,\lambda_{e,t}\right]  ^{\prime}\nonumber
\end{align}
We choose to work with the $\mathbf{A\lambda}_{t}$ parameterization, as the
two elements of this forcing variable $\left(  \lambda_{v,t},\lambda
_{e,t}\right)  $ are not directly correlated, while the elements of
$\mathbf{H}_{t}^{-1}\mathbf{\nabla}_{t}$ are correlated due to the overlapping
term ($\lambda_{v,t}$) appearing in both elements. This aids the
interpretation of the results of the model without changing its fit.

To gain some intuition for how past returns affect current forecasts of ES and
VaR in this model, consider the \textquotedblleft news impact
curve\textquotedblright\ of this model, which presents $\left(  v_{t+1}%
,e_{t+1}\right)  $ as a function of $Y_{t}$ through its impact on
$\mathbf{\lambda}_{t}\equiv\left[  \lambda_{v,t},\lambda\,_{e,t}\right]
^{\prime},$ holding all other variables constant. Figure \ref{figNIC} shows
these two curves for $\alpha=0.05,$ using the estimated parameters for this
model when applied to daily returns on the S\&P 500 index (details are
presented in Section \ref{sAPPLICATION} below). We consider two values for the
\textquotedblleft current\textquotedblright\ value of $\left(  v,e\right)  $:
10\% above and below the long-run average for these variables. We see that for
values where~$Y_{t}>v_{t},$ the news impact curves are flat, reflecting the
fact that on those days the value of the realized return does not enter the
forcing variable. When $Y_{t}\leq v_{t},$ we see that ES and VaR react
linearly to $Y$ and this reaction is through the $\lambda_{e,t}$ forcing
variable; the reaction through the $\lambda_{v,t}$ forcing variable is a
simple step (down) in both of these risk measures.

\bigskip

[ INSERT FIGURE \ref{figNIC} ABOUT HERE ]

\subsection{\label{sFZ1F}A one-factor GAS\ model for ES and VaR}

The specification in Section \ref{sGAS} allows ES and VaR to evolve as two
separate, correlated, processes. In many risk forecasting applications, a
useful simpler model is one based on a structure with only one time-varying
risk measure, e.g. volatility. We will consider a one-factor model in this
section, and will name the model in Section \ref{sGAS} a \textquotedblleft
two-factor\textquotedblright\ GAS model.

Consider the following one-factor GAS model for ES and VaR, where both risk
measures are driven by a single variable, $\kappa_{t}$:%
\begin{align}
v_{t}  & =a\exp\left\{  \kappa_{t}\right\} \\
e_{t}  & =b\exp\left\{  \kappa_{t}\right\}  \text{, \ where }b<a<0\nonumber\\
\text{and \ }\kappa_{t}  & =\omega+\beta\kappa_{t-1}+\gamma H_{t-1}%
^{-1}s_{t-1}\nonumber
\end{align}
The forcing variable, $H_{t-1}^{-1}s_{t-1},$ in the evolution equation for
$\kappa_{t}$ is obtained from the FZ0 loss function, plugging in $\left(
a\exp\left\{  \kappa_{t}\right\}  ,b\exp\left\{  \kappa_{t}\right\}  \right)
$ for $\left(  v_{t},e_{t}\right)  $. \ Using details provided in Appendix
B.2, we find that the score and Hessian are:%
\begin{align}
s_{t}  & \equiv\frac{\partial L_{FZ0}\left(  Y_{t},a\exp\left\{  \kappa
_{t}\right\}  ,b\exp\left\{  \kappa_{t}\right\}  ;\alpha\right)  }%
{\partial\kappa}=-\frac{1}{e_{t}}\left(  \frac{1}{\alpha}\mathbf{1}\left\{
Y_{t}\leq v_{t}\right\}  Y_{t}-e_{t}\right) \\
\text{and \ }I_{t}  & \equiv\frac{\partial^{2}\mathbb{E}_{t-1}\left[
L_{FZ0}\left(  Y_{t},a\exp\left\{  \kappa_{t}\right\}  ,b\exp\left\{
\kappa_{t}\right\}  ;\alpha\right)  \right]  }{\partial\kappa_{t}^{2}}%
=\frac{\alpha-k_{\alpha}a_{\alpha}}{\alpha}%
\end{align}
where $k_{\alpha}$ is a negative constant and $a_{\alpha}$ lies between zero
and one. The Hessian, $I_{t}$, turns out to be a constant in this case, and
since we estimate a free coefficient on our forcing variable, we simply set
$H_{t}$ to one. Note that the VaR score, $\lambda_{v,t}=\partial L/\partial
v$, turns out to drop out from the forcing variable. Thus the one-factor GAS
model for ES and\ VaR becomes:%
\begin{equation}
\kappa_{t}=\omega+\beta\kappa_{t-1}+\gamma\frac{-1}{b\exp\left\{  \kappa
_{t-1}\right\}  }\left(  \frac{1}{\alpha}\mathbf{1}\left\{  Y_{t-1}\leq
a\exp\left\{  \kappa_{t-1}\right\}  \right\}  Y_{t-1}-b\exp\left\{
\kappa_{t-1}\right\}  \right) \label{eqnFZ1F}%
\end{equation}

Using the FZ loss function for estimation, we are unable to identify $\omega,$
as there exists $\left(  \tilde{\omega},\tilde{a},\tilde{b}\right)
\neq\left(  \omega,a,b\right)  $ such that both triplets yield identical
sequences of ES and VaR estimates, and thus identical values of the objective
function. We fix $\omega=0$ and forfeit identification of the level of the
series for $\kappa_{t}$, though we of course retain the ability to model and
forecast ES and VaR.\footnote{This one-factor model for ES and VaR can also be
obtained by considering a zero-mean volatility model for $Y_{t}$, with $iid$
standardized residuals, say denoted $\eta_{t}.$ In this case, $\kappa_{t}$ is
the log conditional standard deviation of $Y_{t}$, and $a=F_{\eta}^{-1}\left(
\alpha\right)  $ and $b=\mathbb{E}\left[  \eta|\eta\leq a\right]  .$ (We
exploit this interpretation when linking these models to GARCH models in
Section \ref{sGFZ} below.) The lack of identification of $\omega$ means that
we do not identify the level of log volatility.} Foreshadowing the empirical
results in Section \ref{sAPPLICATION}, we find that this one-factor GAS model
outperforms the two-factor GAS model in out-of-sample forecasts for most of
the asset return series that we study.

\subsection{Existing dynamic models for ES and VaR}

As noted in the introduction, there is a relative paucity of dynamic models
for ES and VaR, but there is not a complete absence of such models. The
simplest existing model is based on a simple rolling window estimate of these
quantities:
\begin{align}
\widehat{\mathrm{VaR}}_{t}  & =\widehat{\text{\textit{Quantile}}}\left\{
Y_{s}\right\}  _{s=t-m}^{t-1}\\
\widehat{\mathrm{ES}}_{t}  & =\frac{1}{\alpha m}%
{\displaystyle\sum\limits_{s=t-m}^{t-1}}
Y_{s}\mathbf{1}\left\{  Y_{s}\leq\widehat{\mathrm{VaR}}_{s}\right\} \nonumber
\end{align}
where $\widehat{\text{\textit{Quantile}}}\left\{  Y_{s}\right\}
_{s=t-m}^{t-1} $ denotes the sample quantile of $Y_{s}$ over the period
$s\in\left[  t-m,t-1\right]  .$ Common choices for the window size, $m,$
include 125, 250 and 500, corresponding to six months, one year and two years
of daily return observations respectively.

A more challenging competitor for the new ES and VaR models proposed in this
paper are those based on ARMA-GARCH dynamics for the conditional mean and
variance, accompanied by some assumption for the distribution of the
standardized residuals. These models all take the form:%
\begin{align}
Y_{t}  & =\mu_{t}+\sigma_{t}\eta_{t}\\
\eta_{t}  & \thicksim iid~F_{\eta}\left(  0,1\right) \nonumber
\end{align}
where $\mu_{t}$ and $\sigma_{t}^{2}$ are specified to follow some ARMA
and\ GARCH model, and $F_{\eta}\left(  0,1\right)  $ is some arbitrary,
strictly increasing, distribution with mean zero and variance one. What
remains is to specify a distribution for the standardized residual, $\eta_{t}%
$. Given a choice for $F_{\eta},$ VaR and\ ES forecasts are obtained as:%
\begin{align}
v_{t}  & =\mu_{t}+a\sigma_{t}\text{, \ \ where \ }a=F_{\eta}^{-1}\left(
\alpha\right) \\
e_{t}  & =\mu_{t}+b\sigma_{t}\text{, \ \ where \ }b=\mathbb{E}\left[  \eta
_{t}|\eta_{t}\leq a\right] \nonumber
\end{align}
Two parametric choices for $F_{\eta}$ are common in the literature:%
\begin{align}
\eta_{t}  & \thicksim iid~N\left(  0,1\right) \\
\eta_{t}  & \thicksim iid~Skew~t\left(  0,1,\nu,\lambda\right) \nonumber
\end{align}
There are various skew $t$ distributions used in the literature; in the
empirical analysis below we use that of Hansen (1994). A nonparametric
alternative is to estimate the distribution of $\eta_{t}$ using the empirical
distribution function (EDF), an approach that is also known as
\textquotedblleft filtered historical simulation,\textquotedblright\ and one
that is perhaps the best existing model for ES, see the survey by Engle and
Manganelli (2004b).\footnote{Some authors have also considered modeling the
tail of $F_{\eta}$ using extreme value theory, however for the relatively
non-extreme values of $\alpha$ we consider here, past work (e.g., Engle and
Manganelli (2004b), Nolde and Ziegel (2016) and Taylor (2017)) has found\ EVT
to perform no better than the EDF, and so we do not include it in our
analysis.} We consider all of these models in our empirical analysis
in\ Section \ref{sAPPLICATION}.

\subsection{\label{sGASGARCH}GARCH and ES/VaR estimation}

In this section we consider two extensions of the models presented above, in
an attempt to combine the success and parsimony of GARCH models with this
paper's focus on ES and\ VaR forecasting.

\subsubsection{\label{sGFZ}Estimating a GARCH model via FZ minimization}

If an ARMA-GARCH model, including the specification for the distribution of
standardized residuals, is correctly specified for the conditional
distribution of an asset return, then maximum likelihood is the most efficient
estimation method, and should naturally be adopted. If, on the other hand, we
consider an ARMA-GARCH model only as a useful approximation to the true
conditional distribution, then it is no longer clear that MLE is optimal.\ In
particular, if the application of the model is to ES and VaR forecasting, then
we might be able to improve the fitted ARMA-GARCH model by estimating the
parameters of that model via FZ loss minimization, as discussed in Section
\ref{sFZ0}. This estimation method is related to one discussed in Remark 1 of
Francq and Zako\"{\i}an (2015).

Consider the following model for asset returns:%
\begin{align}
Y_{t}  & =\kappa_{t}\eta_{t}\text{, \ \ }\eta_{t}\thicksim iid~F_{\eta}\left(
0,1\right) \\
\kappa_{t}^{2}  & =\omega+\beta\kappa_{t-1}^{2}+\gamma Y_{t-1}^{2}\nonumber
\end{align}
The variable $\kappa_{t}^{2}$ is the conditional variance and is assumed to
follow a GARCH(1,1) process. This model implies a structure analogous to the
one-factor GAS\ model presented in Section \ref{sFZ1F}, as we find:%
\begin{align}
v_{t}  & =a\cdot\kappa_{t}\text{, \ where }a=F_{\eta}^{-1}\left(
\alpha\right) \\
e_{t}  & =b\cdot\kappa_{t}\text{, \ where }b=\mathbb{E}\left[  \eta|\eta\leq
a\right] \nonumber
\end{align}
Some further results on VaR and ES in dynamic location-scale models are
presented in Appendix B.3. To apply this model to\ VaR and\ ES forecasting, we
also have to estimate the VaR and\ ES of the standardized residual, denoted
$\left(  a,b\right)  .$ Rather than estimating the parameters of this model
using (Q)MLE, we consider here estimating the via FZ loss minimization. As in
the one-factor GAS model, $\omega$ is unidentified and we set it to one, so
the parameter vector to be estimated is $\left(  \beta,\gamma,a,b\right)  $.
This estimation approach leads to a fitted GARCH model that is tailored to
provide the best-fitting ES and VaR forecasts, rather than the best-fitting
volatility forecasts.

\subsubsection{\label{sHYBRID}A hybrid GAS/GARCH model}

Finally, we consider a direct combination of the forcing variable suggested by
a GAS structure for a one-factor model of returns, described in equation
(\ref{eqnFZ1F}), with the successful GARCH model for volatility. We specify:%
\begin{align}
Y_{t}  & =\exp\left\{  \kappa_{t}\right\}  \eta_{t}\text{, \ \ }\eta
_{t}\thicksim iid~F_{\eta}\left(  0,1\right) \\
\kappa_{t}  & =\omega+\beta\kappa_{t-1}+\gamma\left(  -\frac{1}{e_{t-1}%
}\left(  \frac{1}{\alpha}\mathbf{1}\left\{  Y_{t-1}\leq v_{t-1}\right\}
Y_{t-1}-e_{t-1}\right)  \right)  +\delta\log\left\vert Y_{t-1}\right\vert
\nonumber
\end{align}
The variable~$\kappa_{t}$ is the log-volatility, identified up to scale. As
the latent variable in this model is log-volatility, we use the lagged log
absolute return rather than the lagged squared return, so that the units
remain in line for the evolution equation for $\kappa_{t}$. There are five
parameters in this model $\left(  \beta,\gamma,\delta,a,b\right)  ,$ and we
estimate them using FZ loss minimization.

\section{\label{sTHEORY}Estimation of dynamic models for ES and\ VaR}

This section presents asymptotic theory for the estimation of dynamic ES
and\ VaR models by minimizing FZ loss. Given a sample of observations $\left(
y_{1},\cdots,y_{T}\right)  $ and a constant $\alpha\in\left(  0,0.5\right)  $,
we are interested in estimating and forecasting the conditional $\alpha$
quantile (VaR) and corresponding expected shortfall of $Y_{t}$. Suppose
$Y_{t}$ is a real-valued random variable that has, conditional on information
set $\mathcal{F}_{t-1}$, distribution function $F_{t}\left(  \cdot
|\mathcal{F}_{t-1}\right)  \ $and corresponding density function $f_{t}\left(
\cdot|\mathcal{F}_{t-1}\right)  $. Let $v_{1}(\mathbf{\theta}^{0})$ and
$e_{1}(\mathbf{\theta}^{0})$ be some initial conditions for VaR and\ ES and
let $\mathcal{F}_{t-1}=\sigma\{Y_{t-1},\mathbf{X}_{t-1},\cdots,Y_{1}%
,\mathbf{X}_{1}\},$ where $\mathbf{X}_{t}$ is a vector of exogenous variables
or predetermined variables, be the information set available for forecasting
$Y_{t}$. The vector of unknown parameters to be estimated is $\mathbf{\theta
}^{0}\in\Theta\subset\mathbb{R}^{p}$.

The conditional VaR and ES of $Y_{t}$ at probability level $\alpha,$ that is
$\mathrm{VaR}_{\alpha}\left(  Y_{t}|\mathcal{F}_{t-1}\right)  $ and
$\mathrm{ES}_{\alpha}\left(  Y_{t}|\mathcal{F}_{t-1}\right)  $, are assumed to
follow some dynamic model:
\begin{equation}
\left[
\begin{array}
[c]{c}%
\mathrm{VaR}_{\alpha}\left(  Y_{t}|\mathcal{F}_{t-1}\right) \\
\mathrm{ES}_{\alpha}\left(  Y_{t}|\mathcal{F}_{t-1}\right)
\end{array}
\right]  =\left[
\begin{array}
[c]{c}%
v(Y_{t-1},\mathbf{X}_{t-1},\cdots,Y_{1},\mathbf{X}_{1};\mathbf{\theta}^{0})\\
e(Y_{t-1},\mathbf{X}_{t-1},\cdots,Y_{1},\mathbf{X}_{1};\mathbf{\theta}^{0})
\end{array}
\right]  \equiv\left[
\begin{array}
[c]{c}%
v_{t}(\mathbf{\theta}^{0})\\
e_{t}(\mathbf{\theta}^{0})
\end{array}
\right]  ,~t=1,\cdots,T,
\end{equation}
The unknown parameters are estimated as:%
\begin{align}
\mathbf{\hat{\theta}}_{T}  & \equiv\arg\min_{\mathbf{\theta}\in\Theta}%
~L_{T}(\mathbf{\theta})\\
\text{where \ }L_{T}(\mathbf{\theta})  & =\frac{1}{T}\sum_{t=1}^{T}%
L_{FZ0}\left(  Y_{t},v_{t}\left(  \mathbf{\theta}\right)  ,e_{t}\left(
\mathbf{\theta}\right)  ;\alpha\right) \nonumber
\end{align}
and the FZ loss function $L_{FZ0}$ is defined in equation (\ref{eqnFZ0}).
Below we provide conditions under which estimation of these parameters via FZ
loss minimization leads to a consistent and asymptotically normal estimator,
with standard errors that can be consistently estimated.

\begin{assumption}
(A) $L\left(  Y_{t},v_{t}\left(  \mathbf{\theta}\right)  ,e_{t}\left(
\mathbf{\theta}\right)  ;\alpha\right)  $ obeys the uniform law of large numbers.

(B)(i) $\Theta$ is a compact subset of $\mathbb{R}^{p}$ for $p<\infty.$
(ii)$\{Y_{t}\}_{t=1}^{\infty}$ is a strictly stationary process. Conditional
on all the past information $\mathcal{F}_{t-1}$, the distribution of $Y_{t}$
is $F_{t}\left(  \cdot|\mathcal{F}_{t-1}\right)  $ which, for all $t,$ belongs
to a class of distribution functions on $\mathbb{R}$ with finite first moments
and unique $\alpha$-quantiles. (iii) $\forall t$, both $v_{t}(\mathbf{\theta
})$ and $e_{t}(\mathbf{\theta})$ are $\mathcal{F}_{t-1}$-measurable and
continuous in $\mathbf{\theta}$. (iv) If $\Pr\left[  v_{t}(\mathbf{\theta
})=v_{t}(\mathbf{\theta}^{0})\cap e_{t}(\mathbf{\theta})=e_{t}(\mathbf{\theta
}^{0})\right]  =1~\forall~t$, then $\mathbf{\theta=\theta}^{0}.$
\end{assumption}

\begin{theorem}
[Consistency]\label{thmCONSISTENCY} Under Assumption 1, $\mathbf{\hat{\theta}%
}_{T}\overset{p}{\rightarrow}\mathbf{\theta}^{0}$ as $T\rightarrow\infty.$
\end{theorem}

The proof of Theorem \ref{thmCONSISTENCY}, provided in Appendix A, is
straightforward given Theorem 2.1 of Newey and\ McFadden (1994) and Corollary
5.5 of Fissler and Ziegel (2016). Note that a variety of uniform laws of large
numbers (our Assumption 1(A)) are available for the time series applications
we consider here, see Andrews (1987) and P\"{o}tscher and Prucha (1989) for
example. Zwingmann and Holzmann (2016) show that if the $\alpha$-quantile is
not unique (violating our Assumption 1(B)(iii)), then the convergence rate and
asymptotic distribution of $\left(  \hat{v}_{T},\hat{e}_{T}\right)  $ are
non-standard, even in a setting with $iid$ data. We do not consider such
problematic cases here.

We next turn to the asymptotic distribution of our parameter estimator. In the
assumptions below, $K$ denotes a finite constant that can change from line to
line, and we use $\left\Vert \mathbf{x}\right\Vert $ to denote the Euclidean
norm of a vector $\mathbf{x.}$

\begin{assumption}
(A) For all $t$, we have (i) $v_{t}(\mathbf{\theta})$ and $e_{t}%
(\mathbf{\theta})$ are twice continuously differentiable in $\mathbf{\theta}$,
(ii) $v_{t}(\mathbf{\theta}^{0})\leq0$.

(B) For all $t$, we have (i) Conditional on all the past information
$\mathcal{F}_{t-1}$, $Y_{t}$ has a continuous density $f_{t}\left(
\cdot|\mathcal{F}_{t-1}\right)  $ that satisfies $f_{t}(y|\mathcal{F}%
_{t-1})\leq K<\infty$ and $\left\vert f_{t}(y_{1}|\mathcal{F}_{t-1}%
)-f_{t}(y_{2}|\mathcal{F}_{t})\right\vert \leq K\left\vert y_{1}%
-y_{2}\right\vert $, (ii) $\mathbb{E}\left[  \left\vert Y_{t}\right\vert
^{4+\delta}\right]  \leq K<\infty$, for some $0<\delta<1$.

(C) There exists$~$a neighborhood of $\mathbf{\theta}^{0}$, $\mathcal{N}%
\left(  \mathbf{\theta}^{0}\right)  $, such that for all $t$ we have (i)
$|1/e_{t}(\theta)|\leq K<\infty$, $\forall~\theta\in\mathcal{N}\left(
\mathbf{\theta}^{0}\right)  ,$ (ii) there exist some (possibly stochastic)
$\mathcal{F}_{t-1}$-measurable functions $V(\mathcal{F}_{t-1})$,
$V_{1}(\mathcal{F}_{t-1})$, $H_{1}(\mathcal{F}_{t-1})$, $V_{2}(\mathcal{F}%
_{t-1})$, $H_{2}(\mathcal{F}_{t-1})$ which satisfy $\forall\mathbf{\theta}%
\in\mathcal{N}(\mathbf{\theta}^{0})$: $|v_{t}(\mathbf{\theta})|\leq
V(\mathcal{F}_{t-1})$, $\left\Vert \nabla v_{t}(\mathbf{\theta})\right\Vert
\leq V_{1}(\mathcal{F}_{t-1})$, $\left\Vert \nabla e_{t}(\mathbf{\theta
})\right\Vert \leq H_{1}(\mathcal{F}_{t-1})$, $\left\Vert \nabla^{2}%
v_{t}(\mathbf{\theta})\right\Vert \leq V_{2}(\mathcal{F}_{t-1})$, and
$\left\Vert \nabla^{2}e_{t}(\mathbf{\theta})\right\Vert \leq H_{2}%
(\mathcal{F}_{t-1})$.

(D) For some $0<\delta<1$ and for all $t$ we have (i) $\mathbb{E}\left[
V_{1}(\mathcal{F}_{t-1})^{3+\delta}\right]  $, $\mathbb{E}\left[
H_{1}(\mathcal{F}_{t-1})^{3+\delta}\right]  $, $\mathbb{E}\left[
V_{2}(\mathcal{F}_{t-1})^{\frac{3+\delta}{2}}\right]  $, $\mathbb{E}\left[
H_{2}(\mathcal{F}_{t-1})^{\frac{3+\delta}{2}}\right]  \leq K$, (ii)
$\mathbb{E}\left[  V(\mathcal{F}_{t-1})^{2+\delta}V_{1}(\mathcal{F}%
_{t-1})H_{1}(\mathcal{F}_{t-1})^{2+\delta}\right]  \leq K$,\linebreak(iii)
$\mathbb{E}\left[  H_{1}(\mathcal{F}_{t-1})^{1+\delta}H_{2}(\mathcal{F}%
_{t-1})\left\vert Y_{t}\right\vert ^{2+\delta}\right]  $, $\mathbb{E}\left[
H_{1}(\mathcal{F}_{t-1})^{3+\delta}\left\vert Y_{t}\right\vert ^{2+\delta
}\right]  \leq K.$

(E) The matrix $\mathbf{D}_{T}$ defined in Theorem \ref{thmASYMPNORM} has
eigenvalues bounded below by a positive constant for $T$ sufficiently large.

(F) The sequence $\{T^{-1/2}\sum_{t=1}^{T}g_{t}(\mathbf{\theta}^{0})\}$ obeys
the CLT, where
\begin{align}
g_{t}(\mathbf{\theta}) &  =\frac{\partial L\left(  Y_{t},v_{t}\left(
\mathbf{\theta}\right)  ,e_{t}\left(  \mathbf{\theta}\right)  ;\alpha\right)
}{\partial\mathbf{\theta}}\\
= &  \nabla v_{t}(\mathbf{\theta})^{\prime}\frac{1}{-e_{t}(\mathbf{\theta}%
)}\left(  \frac{1}{\alpha}\mathbf{1}\left\{  Y_{t}\leq v_{t}(\mathbf{\theta
})\right\}  -1\right) \\
+ &  \nabla e_{t}(\mathbf{\theta})^{\prime}\frac{1}{e_{t}(\mathbf{\theta}%
)^{2}}\left(  \frac{1}{\alpha}\mathbf{1}\left\{  Y_{t}\leq v_{t}%
(\mathbf{\theta})\right\}  (v_{t}(\mathbf{\theta})-Y_{t})-v_{t}(\mathbf{\theta
})+e_{t}(\mathbf{\theta})\right) \nonumber
\end{align}

(G) $\left\{  Y_{t}\right\}  $ is $\alpha$-mixing of size $-2q/\left(
q-2\right)  $ for some $q>2.$
\end{assumption}

Most of the above assumptions are standard. Assumption 2(A)(i) imposes that
the VaR is negative, but given our focus on the left-tail $\left(  \alpha
\leq0.5\right)  $ of asset returns, this is not likely a binding constraint.
Assumptions 2(B),(C) and (E) are similar to those in Engle and Manganelli
(2004a). Assumption 2(B)(ii) requires at least $4+\delta$ moments of returns
to exist, however 2(D) may actually increase the number of required moments,
depending on the VaR-ES model employed. For the familiar GARCH(1,1)\ process,
used in our simulation study, it can be shown that we only need to assume that
$4+\delta$ moments exist. Assumption 2(F) allows for some CLT for mixing data
to be invoked, and 2(G) is a standard assumption on the time series dependence
of the data.

\begin{theorem}
[Asymptotic Normality]\label{thmASYMPNORM} Under Assumptions 1 and 2, we have
\begin{equation}
\sqrt{T}\mathbf{A}_{T}^{-1/2}\mathbf{D}_{T}(\mathbf{\hat{\theta}}%
_{T}-\mathbf{\theta}^{0})\overset{d}{\rightarrow}N(0,I)\text{ as }%
T\rightarrow\infty
\end{equation}
where
\begin{align}
\mathbf{D}_{T} &  =\mathbb{E}\left[  T^{-1}\sum_{t=1}^{T}\frac{f_{t}\left(
v_{t}(\mathbf{\theta}^{0})|\mathcal{F}_{t-1}\right)  }{-e_{t}(\mathbf{\theta
}^{0})\alpha}\nabla v_{t}(\mathbf{\theta}^{0})^{\prime}\nabla v_{t}%
(\mathbf{\theta}^{0})+\frac{1}{e_{t}(\mathbf{\theta}^{0})^{2}}\nabla
e_{t}(\mathbf{\theta}^{0})^{\prime}\nabla e_{t}(\mathbf{\theta}^{0})\right] \\
\mathbf{A}_{T} &  =\mathbb{E}\left[  T^{-1}\sum_{t=1}^{T}g_{t}(\mathbf{\theta
}^{0})g_{t}(\mathbf{\theta}^{0})^{\prime}\right]
\end{align}
and $g_{t}$ is defined in Assumption 2(F).
\end{theorem}

An outline of the proof of this theorem is given in Appendix A, and the
detailed lemmas underlying it are provided in the supplemental appendix. The
proof of Theorem \ref{thmASYMPNORM} builds on Huber (1967), Weiss (1991)
and\ Engle and\ Manganelli (2004a), who focused on the estimation of quantiles.

Finally, we present a result for estimating the asymptotic covariance matrix
of $\mathbf{\hat{\theta}}_{T},$ thereby enabling the reporting of standard
errors and confidence intervals.

\begin{assumption}
(A) The deterministic positive sequence $c_{T}$ satisfies $c_{T}=o(1)$ and
$c_{T}^{-1}=o(T^{1/2})$.

(B)(i) $T^{-1}\sum_{t=1}^{T}g_{t}(\mathbf{\theta}^{0})g_{t}(\mathbf{\theta
}^{0})^{^{\prime}}-\mathbf{A}_{T}\overset{p}{\rightarrow}{\mathbf{0}}$, where
$\mathbf{A}_{T}$ is defined in Theorem 2.

(ii) $T^{-1}\sum_{t=1}^{T}\frac{1}{e_{t}(\theta^{0})^{2}}\nabla e_{t}%
(\mathbf{\theta}^{0})^{\prime}\nabla e_{t}(\mathbf{\theta}^{0})-\mathbb{E}%
[T^{-1}\sum_{t=1}^{T}\frac{1}{e_{t}(\mathbf{\theta}^{0})^{2}}\nabla
e_{t}(\mathbf{\theta}^{0})^{\prime}\nabla e_{t}(\mathbf{\theta}^{0}%
)]\overset{p}{\rightarrow}{\mathbf{0}}$.

(iii) $T^{-1}\sum_{t=1}^{T}\frac{f_{t}(v_{t}(\mathbf{\theta}^{0}%
)|\mathcal{F}_{t-1})}{-e_{t}(\mathbf{\theta}^{0})\alpha}\nabla v_{t}%
(\mathbf{\theta}^{0})^{\prime}\nabla v_{t}(\mathbf{\theta}^{0})-\mathbb{E}%
[T^{-1}\sum_{t=1}^{T}\frac{f_{t}(v_{t}(\mathbf{\theta}^{0})|\mathcal{F}%
_{t-1})}{-e_{t}(\mathbf{\theta}^{0})\alpha}\nabla v_{t}(\mathbf{\theta}%
^{0})^{\prime}\nabla v_{t}(\mathbf{\theta}^{0})]\overset{p}{\rightarrow
}{\mathbf{0}}$.
\end{assumption}

\begin{theorem}
\label{thmVCV}Under Assumptions 1-3, $\mathbf{\hat{A}}_{T}-\mathbf{A}%
_{T}\overset{p}{\rightarrow}\mathbf{0}$ and $\mathbf{\hat{D}}_{T}%
-\mathbf{D}_{T}\overset{p}{\rightarrow}\mathbf{0}$, where
\begin{align*}
\mathbf{\hat{A}}_{T}= &  T^{-1}\sum_{t=1}^{T}g_{t}(\mathbf{\hat{\theta}}%
_{T})g_{t}(\mathbf{\hat{\theta}}_{T})^{\prime}\\
\mathbf{\hat{D}}_{T}= &  T^{-1}\sum_{t=1}^{T}\left\{  \frac{1}{2c_{T}%
}\mathbf{1}\left\{  \left\vert y_{t}-v_{t}\left(  \mathbf{\hat{\theta}}%
_{T}\right)  \right\vert <c_{T}\right\}  \frac{\nabla^{\prime}v_{t}\left(
\mathbf{\hat{\theta}}_{T}\right)  \nabla v_{t}\left(  \mathbf{\hat{\theta}%
}_{T}\right)  }{-\alpha e_{t}\left(  \mathbf{\hat{\theta}}_{T}\right)  }%
+\frac{\nabla^{\prime}e_{t}\left(  \mathbf{\hat{\theta}}_{T}\right)  \nabla
e_{t}\left(  \mathbf{\hat{\theta}}_{T}\right)  }{e_{t}^{2}\left(
\mathbf{\hat{\theta}}_{T}\right)  }\right\}
\end{align*}

\end{theorem}

This result extends Theorem 3 in Engle and\ Manganelli\ (2004a) from dynamic
VaR models to dynamic joint models for VaR and ES. The key choice in
estimating the asymptotic covariance matrix is the bandwidth parameter in
Assumption 3(A). In our simulation study below we set this to $T^{-1/3}$ and
we find that this leads to satisfactory finite-sample properties.

The results here extend some very recent work in the literature: Dimitriadis
and Bayer (2017) consider VaR-ES regression, but focus on $iid$ data and
linear specifications.\footnote{Dimitriadis and Bayer (2017) also consider a
variety of FZ loss functions, in contrast with our focus on the FZ0 loss
function, and they consider both~$M$ and GMM (or $Z$, in their
notation)\ estimation, while we focus only on $M $ estimation.} Barendse
(2017) considers \textquotedblleft interquantile expectation
regression,\textquotedblright\ which nests VaR-ES regression as a special
case. He allows for time series data, but imposes that the models are linear.
Our framework allows for time series data and nonlinear models.

\section{\label{sSIMULATION}Simulation study}

In this section we investigate the finite-sample accuracy of the asymptotic
theory for dynamic ES and\ VaR models presented in the previous section. For
ease of comparison with existing studies of related models, such as volatility
and VaR models, we consider a GARCH(1,1) for the DGP, and estimate the
parameters by FZ-loss minimization. Specifically, the DGP is%
\begin{align}
Y_{t}  & =\sigma_{t}\eta_{t}\\
\sigma_{t}^{2}  & =\omega+\beta\sigma_{t-1}^{2}+\gamma Y_{t-1}^{2}\nonumber\\
\eta_{t}  & \thicksim iid~F_{\eta}\left(  0,1\right)
\end{align}
We set the parameters of this DGP to $\left(  \omega,\beta,\gamma\right)
=\left(  0.05,0.9,0.05\right)  .$ We consider two choices for the distribution
of $\eta_{t}$: a standard Normal, and the standardized skew $t$ distribution
of Hansen (1994), with degrees of freedom and skewness parameters in the
latter set to $\left(  5,-0.5\right)  .$ Under this DGP, the ES and VaR are
proportional to $\sigma_{t}$, \ with
\begin{equation}
\left(  \mathrm{VaR}_{t}^{\alpha},\mathrm{ES}_{t}^{\alpha}\right)  =\left(
a_{\alpha},b_{\alpha}\right)  \sigma_{t}%
\end{equation}
We make the dependence of the coefficients of proportionality $\left(
a_{\alpha},b_{\alpha}\right)  $ on $\alpha$ explicit here, as we consider a
variety of values of $\alpha$ in this simulation study: $\alpha\in\left\{
0.01,0.025,0.05,0.10,0.20\right\}  .$ Interest in VaR and\ ES from regulators
focuses on the smaller of these values of $\alpha,$ but we also consider the
larger values to better understand the properties of the asymptotic
approximations at various points in the tail of the distribution.

For a standard Normal distribution, with CDF and PDF denoted $\Phi$ and
$\phi,$ we have:%
\begin{align}
a_{\alpha}  & =\Phi^{-1}\left(  \alpha\right) \\
b_{\alpha}  & =-\phi\left(  \Phi^{-1}\left(  \alpha\right)  \right)
/\alpha\nonumber
\end{align}
For Hansen's skew $t$ distribution we can obtain $a_{\alpha}$ from the inverse
CDF, but no closed-form expression for $b_{\alpha}$ is available; we instead
use a simulation of 10 million $iid$ draws to estimate it. As noted above, FZ
loss minimization does not allow us to identify $\omega$ in the GARCH model,
and in our empirical work we set this parameter to 1. To facilitate
comparisons of the accuracy of estimates of $\left(  a_{\alpha},b_{\alpha
}\right)  $ in our simulation study we instead set $\omega$ at its true value.
This is done without loss of generality and merely eases the presentation of
the results. To match our empirical application, we replace the parameter
$a_{\alpha}$ with $c_{\alpha}=a_{\alpha}/b_{\alpha},$ and so our parameter
vector becomes $\left[  \beta,\gamma,b_{\alpha},c_{\alpha}\right]  .$

We consider two sample sizes, $T\in\left\{  2500,5000\right\}  $ corresponding
to 10 and 20 years of daily returns respectively. These large sample sizes
enable us to consider estimating models for quantiles as low as 1\%, which are
often used in risk management. We repeat all simulations 1000 times.

Table 1 presents results for the estimation of this model on standard Normal
innovations, and Table 2 presents corresponding results for skew $t$
innovations. The top row of each panel present the true parameter values, with
the latter two parameters changing across $\alpha.$ The second row presents
the median estimated parameter across simulations, and the third row presents
the average bias in the estimated parameter. Both of these measures indicate
that the parameter estimates are nicely centered on the true parameter values.
The penultimate row presents the cross-simulation standard deviations of the
estimated parameters, and we observe that these decrease with the sample size
and increase as we move further into the tails (i.e., as $\alpha$ decreases),
both as expected. Comparing the standard deviations across Tables 1 and 2, we
also note that they are higher for skew $t$ innovations than Normal
innovations, again as expected.

The last row in each panel presents the coverage probabilities for 95\%
confidence intervals for each parameter, constructed using the estimated
standard errors, with bandwidth parameter $c_{T}=\left\lfloor T^{-1/3}%
\right\rfloor $. For $\alpha\geq0.05$ we see that the coverage is reasonable,
ranging from around 0.88 to 0.96. For $\alpha=0.025$ or $\alpha=0.01$ the
coverage tends to be too low, particularly for the smaller sample size. Thus
some caution is required when interpreting the standard errors for the models
with the smallest values of $\alpha.$ In Table S1 of the Supplemental Appendix
we present results for (Q)MLE for the GARCH model corresponding to the results
in Tables 1 and 2, using the theory of Bollerslev and Wooldridge (1992). In
Tables S2 and S3 we present results for CAViaR estimation of this model, using
the \textquotedblleft tick\textquotedblright\ loss function and the theory of
Engle and\ Manganelli\ (2004a).\footnote{In (Q)MLE, the parameters to be
estimated are $\left[  \omega,\beta,\gamma\right]  .$ In \textquotedblleft
CAViaR\textquotedblright\ estimation, which is done by minimizing the
\textquotedblleft tick\textquotedblright\ loss function, the parameters to be
estimated are $\left[  \beta,\gamma,a_{\alpha}\right]  ,$ since in this case
the parameter $\omega$ is again unidentified. As for the study of FZ
estimation, we set $\omega$ to its true value to facilitate interpretation of
the results.} We find that (Q)MLE has better finite sample properties than\ FZ
minimization, but CAViaR estimation has slightly worse properties than\ FZ\ minimization.

\bigskip

[INSERT TABLES 1 AND 2 ABOUT HERE ]

\bigskip

In Table 3 we compare the efficiency of FZ estimation relative to (Q)MLE and
to CAViaR estimation, for the parameters that all three estimation methods
have in common, namely $\left[  \beta,\gamma\right]  .$ As expected, when the
innovations are standard Normal, FZ estimation is substantially less efficient
than MLE, however when the innovations are skew $t$ the loss in efficiency
drops and for some values of $\alpha$ FZ estimation is actually more efficient
than QMLE. This switch in the ranking of the competing estimators is
qualitatively in line with results in\ Francq and Zako\"{\i}an (2015). In
Panel B of Table 3, we see that FZ estimation is generally, though not
uniformly, more efficient than CAViaR estimation.

In many applications, interest is more focused on the forecasted values of VaR
and\ ES than the estimated parameters of the models. To study this, Table 4
presents results on the accuracy of the fitted VaR and\ ES estimates for the
three estimation methods: (Q)MLE, CAViaR and\ FZ estimation. To obtain
estimates of VaR and\ ES from the (Q)ML estimates, we follow common empirical
practice and compute the sample\ VaR and\ ES of the estimated standardized
residuals. In the first column of each panel we present the mean absolute
error (MAE) from (Q)MLE, and in the next two columns we present the
\textit{relative MAE} of CAViaR and FZ to (Q)MLE. Table 4 reveals that (Q)MLE
is the most accurate estimation method. Averaging across values of $\alpha,$
CAViaR is about 40\% worse for Normal innovations, and 24\% worse for skew $t$
innovations, while FZ fares somewhat better, being about 30\% worse for Normal
innovations and 16\% worse for skew $t$ innovations. The superior performance
of (Q)MLE is not surprising when the innovations are Normal, as that
corresponds to (full) maximum likelihood, which has maximal efficiency.
Weighing against the loss in FZ estimation efficiency is the
\textit{robustness} that FZ estimation offers relative to\ QML. For
applications even further from Normality, e.g. with time-varying skewness or
kurtosis, the loss in efficiency of QML is likely even greater.

\bigskip

[INSERT TABLES 3\ AND 4 ABOUT HERE ]

\bigskip

Overall, these simulation results show that the asymptotic results of the
previous section provide reasonable approximations in finite samples, with the
approximations improving for larger sample sizes and less extreme values of
$\alpha.$ Compared with MLE, estimation by FZ loss minimization is generally
less accurate, while it is generally more accurate than estimation using
the\ CAViaR approach of Engle and Manganelli\ (2004a). The latter
outperformance is likely attributable to the fact that FZ estimation draws on
information from two tail measures, VaR and\ ES, while CAViaR was designed to
only model\ VaR.

\section{\label{sAPPLICATION}Forecasting equity index ES and VaR}

We now apply the models discussed in Section \ref{sMODELS} to the forecasting
of ES and\ VaR for daily returns on four international equity indices. We
consider the S\&P 500 index, the Dow\ Jones Industrial Average, the NIKKEI 225
index of Japanese stocks, and the FTSE 100 index of UK stocks. Our sample
period is 1 January 1990 to 31 December 2016, yielding between 6,630 and 6,805
observations per series (the exact numbers vary due to differences in holidays
and market closures). In our out-of-sample analysis, we use the first ten
years for estimation, and reserve the remaining 17 years for evaluation and
model comparison.

Table 5 presents full-sample summary statistics on these four return series.
Average annualized returns range from -2.7\% for the\ NIKKEI to 7.2\% for the
DJIA, and annualized standard deviations range from 17.0\% to 24.7\%. All
return series exhibit mild negative skewness (around -0.15) and substantial
kurtosis (around 10). The lower two panels of Table 5 present the sample VaR
and\ ES for four choices of $\alpha.$

Table 6 presents results from standard time series models estimated on these
return series over the in-sample period (Jan 1990 to Dec 1999). In the first
panel we present the estimated parameters of the optimal ARMA$\left(
p,q\right)  $ models, where the choice of $\left(  p,q\right)  $ is made using
the BIC. The $R^{2}$ values from the optimal models never rises above 1\%,
consistent with the well-known lack of predictability of these series. The
second panel presents the parameters of the\ GARCH(1,1) model for conditional
variance, and the lower panel presents the estimated parameters the skew $t$
distribution applied to the standardized residuals. All of these parameters
are broadly in line with values obtained by other authors for these or similar series.

\bigskip

[ INSERT TABLES 5 AND 6 ABOUT HERE ]

\subsection{In-sample estimation}

We now present estimates of the parameters of the models presented in\ Section
\ref{sMODELS}, along with standard errors computed using the theory from
Section \ref{sTHEORY}.\footnote{Computational details on the estimation of
these models are given in\ Appendix C.} In the interests of space, we only
report the parameter estimates for the S\&P 500 index for $\alpha=0.05$. The
two-factor GAS model based on the FZ0 loss function is presented in the left
panel of Table 7. This model allows for separate dynamics in\ VaR and\ ES, and
we present the parameters for each of these risk measures in separate columns.
We observe that the persistence of these processes is high, with the estimated
$b$ parameters equal to 0.973 and 0.977, similar to the persistence found in
GARCH\ models (e.g., see Table 6). The model-implied average values of VaR
and\ ES are -2.001 and -2.556, similar to the sample values of these measures
reported in Table 5. We also observe that in neither equation is the
coefficient on $\lambda_{v}$ statistically significant: the $t$-statistics on
$a_{v}$ are both well below one. The coefficients on $\lambda_{e}$ are both
larger, and more significant (the $t$-statistics are 1.58 and 1.75),
indicating that the forcing variable from the ES part of the FZ0 loss function
is the more informative component. However, the overall imprecision of the
four coefficients on the forcing variables is suggestive that this model is over-parameterized.

The right panel of Table 7 shows three one-factor models for ES and\ VaR. The
first is the one-factor\ GAS model, which is nested in the two-factor model
presented in the left panel. We see a slight loss in fit (the average loss is
slightly greater) but the parameters of this model are estimated with greater
precision. The one-factor GAS model fits slightly better than the GARCH model
estimated via FZ loss minimization (reported in the penultimate
column).\footnote{Recall that in all of the one-factor models, the intercept
$\left(  \omega\right)  $ in the GAS\ equation is unidentified. We fix it at
zero for the GAS-1F and\ Hybrid models, and at one for the GARCH-FZ model.
This has no impact on the fit of these models for VaR and ES, but it means
that we cannot interpret the estimated $\left(  a,b\right)  $ parameters as
the VaR and\ ES of the standardized residuals, and we no longer expect the
estimated values to match the sample estimates in Table 5.} The
\textquotedblleft hybrid\textquotedblright\ model, augmenting the one-factor
GAS model with a GARCH-type forcing variable, fits better than the other
one-factor models, and also better than the larger two-factor GAS\ model, and
we observe that the coefficient on the GARCH forcing variable $\left(
\delta\right)  $ is significantly different from zero (with a $t$-statistic of 2.07).

\bigskip

[ INSERT TABLE 7 ABOUT HERE ]

\subsection{Out-of-sample forecasting}

We now turn to the out-of-sample (OOS) forecast performance of the models
discussed above, as well as some competitor models from the existing
literature. We will focus initially on the results for $\alpha=0.05,$ given
the focus on that percentile in the extant VaR literature. (Results for other
values of $\alpha$ are considered later, with details provided in the
supplemental appendix.) We will consider a total of ten models for forecasting
ES and\ VaR. Firstly, we consider three rolling window methods, using window
lengths of 125, 250 and 500 days. We next consider ARMA-GARCH models, with the
ARMA\ model orders selected using the BIC, and assuming that the distribution
of the innovations is standard Normal or skew $t,$ or estimating it
nonparametrically using the sample\ ES and VaR of the estimated standardized
residuals. Finally we consider four new semiparametric dynamic models for ES
and VaR: the two-factor GAS model presented in\ Section \ref{sGAS}, the
one-factor GAS model presented in Section \ref{sFZ1F}, a GARCH model estimated
using FZ loss minimization, and the \textquotedblleft hybrid\textquotedblright%
\ GAS/GARCH model presented in Section \ref{sGASGARCH}. We estimate these
models using the first ten years as our in-sample period, and retain those
parameter estimates throughout the OOS period.

In Figure \ref{figSPve1} below we plot the fitted 5\% ES and VaR for the S\&P
500 return series, using three models: the rolling window model using a window
of 125 days, the GARCH-EDF model, and the one-factor GAS model. This figure
covers both the in-sample and out-of-sample periods. The figure shows that the
average ES was estimated at around -2\%, rising as high as around -1\% in the
mid 90s and mid 00s, and falling to its most extreme values of around -10\%
during the financial crisis in late 2008. Thus, like volatility, ES fluctuates
substantially over time.

Figure \ref{figSPve2} zooms in on the last two years of our sample period, to
better reveal the differences in the estimates from these models. We observe
the usual step-like movements in the rolling window estimate of VaR and ES, as
the more extreme observations enter and leave the estimation window. Comparing
the GARCH and GAS estimates, we see how they differ in reacting to returns:
the GARCH estimates are driven by lagged squared returns, and thus move
stochastically each day.\ The GAS estimates, on the other hand, only use
information from returns when the VaR is violated, and on other days the
estimates revert deterministically to the long-run mean. This generates a
smoother time series of VaR and\ ES estimates. We investigate below which of
these estimates provides a better fit to the data.

\bigskip

[ INSERT FIGURES \ref{figSPve1} AND \ref{figSPve2} ABOUT HERE ]

\bigskip

The left panel of Table 8 presents the average OOS losses, using the FZ0 loss
function from equation (\ref{eqnFZ0}), for each of the ten models, for the
four equity return series. The lowest values in each column are highlighted in
bold, and the second-lowest are in italics. We observe that the one-factor
GAS\ model, labelled FZ1F, is the preferred model for the two US equity
indices, while the Hybrid model is the preferred model for the\ NIKKEI and
FTSE indices. The worst model is the rolling window with a window length of
500 days.

While average losses are useful for an initial look at OOS forecast
performance, they do not reveal whether the gains are statistically
significant. Table 9 presents Diebold-Mariano t-statistics on the loss
differences, for the S\&P 500 index. Corresponding tables for the other three
equity return series are presented in Table S4 of the supplemental appendix.
The tests are conducted as \textquotedblleft row model minus column
model\textquotedblright\ and so a positive number indicates that the column
model outperforms the row model. The column \textquotedblleft
FZ1F\textquotedblright\ corresponding to the one-factor GAS model contains all
positive entries, revealing that this model out-performed all competing
models. This outperformance is strongly significant for the comparisons to the
rolling window forecasts, as well as the GARCH model with Normal innovations.
The gains relative to the GARCH model with skew $t$ or nonparametric
innovations are not significant, with DM $t$-statistics of 1.48 and 1.16
respectively. Similar results are found for the best models for each of the
other three equity return series. Thus the worst models are easily separated
from the better models, but the best few models are generally not
significantly different.\footnote{Table S5 in the supplemental appendix
presents results analogous to\ Table 8, but with alpha=0.025, which is the
value for ES that is the focus of the Basel III accord. The rankings and
results are qualitatively similar to those for alpha=0.05 discussed here.}

\bigskip

[ INSERT TABLES 8 AND\ 9 ABOUT HERE ]

\bigskip

To complement the study of the relative performance of these models for ES
and\ VaR, we now consider goodness-of-fit tests for\ the OOS forecasts of\ VaR
and\ ES. Under correct specification of the model for VaR and\ ES, we know
that
\begin{equation}
\mathbb{E}_{t-1}\left[
\begin{array}
[c]{c}%
\partial L_{FZ0}\left(  Y_{t},v_{t},e_{t};\alpha\right)  /\partial v_{t}\\
\partial L_{FZ0}\left(  Y_{t},v_{t},e_{t};\alpha\right)  /\partial e_{t}%
\end{array}
\right]  =0
\end{equation}
and we note that this implies that $\mathbb{E}_{t-1}\left[  \lambda
_{v,t}\right]  =\mathbb{E}_{t-1}\left[  \lambda_{e,t}\right]  =0,$ where
$\left(  \lambda_{v,t},\lambda_{e,t}\right)  $ are defined in equations
(\ref{eqnLAMv})-(\ref{eqnLAMe}). Thus the variables $\lambda_{v,t}$ and
$\lambda_{e,t}$ can be considered as a form of \textquotedblleft generalized
residual\textquotedblright\ for this model. To mitigate the impact of serial
correlation in these measures (which comes through the persistence of $v_{t}$
and $e_{t}$) we use standardized versions of these residuals:%
\begin{align}
\lambda_{v,t}^{s}  & \equiv\frac{\lambda_{v,t}}{v_{t}}=\mathbf{1}\left\{
Y_{t}\leq v_{t}\right\}  -\alpha\\
\lambda_{e,t}^{s}  & \equiv\frac{\lambda_{e,t}}{e_{t}}=\frac{1}{\alpha
}\mathbf{1}\left\{  Y_{t}\leq v_{t}\right\}  \frac{Y_{t}}{e_{t}}-1\nonumber
\end{align}
These standardized generalized residuals are also conditionally mean zero
under correct specification, and we note that the standardized residual for
VaR is simply the demeaned \textquotedblleft hit\textquotedblright\ variable,
which is the focus of well-known tests from the VaR literature, see
Christoffersen (1998) and Engle and\ Manganelli (2004a). We adopt the
\textquotedblleft dynamic quantile (DQ)\textquotedblright\ testing approach of
Engle and\ Manganelli (2004a), which is based on simple regressions of these
generalized residuals on elements of the information set available at the time
the forecast was made. Consider, then the following \textquotedblleft
DQ\textquotedblright\ and \textquotedblleft DES\textquotedblright%
\ regressions:%
\begin{align}
\lambda_{v,t}^{s}  & =a_{0}+a_{1}\lambda_{v,t-1}^{s}+a_{2}v_{t}+u_{v,t}\\
\lambda_{e,t}^{s}  & =b_{0}+b_{1}\lambda_{e,t-1}^{s}+b_{2}e_{t}+u_{e,t}%
\nonumber
\end{align}
We test forecast optimality by testing that all terms ($\mathbf{a=}\left[
a_{0},a_{1},a_{2}\right]  ^{\prime}$\ and $\mathbf{b=}\left[  b_{0}%
,b_{1},b_{2}\right]  ^{\prime}$) in these regressions are zero, against the
usual two-sided alternative. Similar \textquotedblleft conditional
calibration\textquotedblright\ tests are presented in Nolde and\ Ziegel
(2017). One could also consider a joint test of both of the above null
hypotheses, however we will focus on these separately so that we can determine
which variable is well/poorly specified.

The right two panels of Table 8 present the $p$-values from the tests of the
goodness-of-fit of the VaR and\ ES forecasts. Entries greater than 0.10
(indicating no evidence against optimality at the 0.10 level) are in bold, and
entries between 0.05 and 0.10 are in italics. For the S\&P 500 index and the
DJIA, we see that only one model passes both the VaR and\ ES tests: the
one-factor GAS model. For the NIKKEI we see that all of the dynamic models
pass these two tests, while all three of the rolling window models fail. For
the FTSE index, on the other hand, we see that all ten models considered here
fail both the goodness-of-fit tests. The outcomes for the NIKKEI and the\ FTSE
each, in different ways, present good examples of the problem highlighted
in\ Nolde and Ziegel (2017), that many different models may pass a
goodness-of-fit test, or all models may fail, which makes discussing their
relative performance difficult. To do so, one can look at Diebold-Mariano
tests of differences in average loss, as we do in Table 9.

Finally, in Table 10 we look at the performance of these models across four
values of $\alpha,$ to see whether the best-performing models change with how
deep in the tails we are. We find that this is indeed the case: for
$\alpha=0.01,$ the best-performing model across the four return series is the
GARCH model estimated by FZ loss minimization, followed by the GARCH model
with nonparametric residuals. These two models are also the (equal) best two
models for $\alpha=0.025$. \ For $\alpha=0.05$ and $\alpha=0.10$ the two best
models are the one-factor GAS model and the Hybrid model. These rankings are
perhaps related to the fact that the forcing variable in the GAS model depends
on observing a violation of the VaR, and for very small values of $\alpha$
these violations occur only infrequently. In contrast, the GARCH\ model uses
the information from the squared residual, and so information from the data
moves the risk measures whether a VaR violation was observed or not. When
$\alpha$ is not so small, the forcing variable suggested by the GAS model
applied to the FZ loss function starts to out-perform.

\bigskip

[ INSERT TABLE 10 ABOUT HERE ]

\section{\label{sCONCLUSION}Conclusion}

With the implementation of the Third Basel Accord in the next few years, risk
managers and regulators will place greater focus on expected shortfall\ (ES)
as a measure of risk, complementing and partly substituting previous emphasis
on\ Value-at-Risk (VaR). We draw on recent results from statistical decision
theory (Fissler and\ Ziegel, 2016) to propose new dynamic models for ES and
VaR. The models proposed are semiparametric, in that they impose parametric
structures for the dynamics of ES and VaR, but are agnostic about the
conditional distribution of returns. We also present asymptotic distribution
theory for the estimation of these models, and we verify that the theory
provides a good approximation in finite samples. We apply the new models and
methods to daily returns on four international equity indices, over the period
1990 to 2016, and find the proposed new ES-VaR models outperform forecasts
based on GARCH or rolling window models.

The asymptotic theory presented in this paper facilitates considering a large
number of extensions of the models presented here. Our models all focus on a
single value for the tail probability $\left(  \alpha\right)  ,$ and extending
these to consider multiple values simultaneously could prove fruitful. For
example, one could consider the values 0.01, 0.025 and 0.05, to capture
various points in the left tail, or one could consider 0.05 and 0.95 to
capture both the left and right tails simultaneously. Another natural
extension is to make use of exogenous information in the model; the models
proposed here are all univariate, and one might expect that information from
options markets, high frequency data, or news announcements to also help
predict VaR and\ ES. We leave these interesting extensions to future research.

\bigskip

\bigskip

\pagebreak

{\LARGE Appendix A:\ Proofs}

\bigskip

\begin{proof}
[Proof of Proposition \ref{propFZ0}]Theorem C.3 of Nolde and\ Ziegel (2017)
shows that under the assumption that ES is strictly negative, the loss
differences generated by a FZ loss function are homogeneous of degree zero iff
$G_{1}(x)=\varphi_{1}\mathbf{1}\left\{  x\geq0\right\}  $ and $G_{2}%
(x)=-\varphi_{2}/x$ with $\varphi_{1}\geq0$ and $\varphi_{2}>0$. Denote the
resulting loss function as $L_{FZ0}^{\ast}\left(  Y,v,e;\alpha,\varphi
_{1},\varphi_{2}\right)  ,$ and notice that:%
\begin{align*}
L_{FZ0}^{\ast}\left(  Y,v,e;\alpha,\varphi_{1},\varphi_{2}\right)   &
=\varphi_{1}\left(  \mathbf{1}\left\{  Y\leq v\right\}  -\alpha\right)
\left(  \mathbf{1}\left\{  v\geq0\right\}  -\mathbf{1}\left\{  Y\geq0\right\}
\right) \\
& +\varphi_{2}\left\{  -\left(  \mathbf{1}\left\{  Y\leq v\right\}
-\alpha\right)  \frac{1}{\alpha}\frac{v}{e}+\frac{1}{e}\left(  \frac{1}%
{\alpha}\mathbf{1}\left\{  Y\leq v\right\}  Y-e\right)  +\log\left(
-e\right)  \right\} \\
& =\varphi_{1}\left(  \mathbf{1}\left\{  Y\leq v\right\}  -\alpha\right)
\left(  \mathbf{1}\left\{  v\geq0\right\}  -\mathbf{1}\left\{  Y\geq0\right\}
\right)  +\varphi_{2}L_{FZ0}\left(  y,v,e;\alpha\right) \\
& =\varphi_{2}L_{FZ0}\left(  Y,v,e;\alpha\right)  +\varphi_{1}\alpha
\mathbf{1}\left\{  Y\geq0\right\}  +\varphi_{1}\left(  1-\alpha-\mathbf{1}%
\left\{  Y\geq0\right\}  \right)  \mathbf{1}\left\{  v\geq0\right\}
\end{align*}
Under the assumption that $v<0,$ the third term vanishes. The second term is
purely a function of $Y$ and so can be disregarded; we can set $\varphi_{1}=0$
without loss of generality. The first term is affected by a scaling parameter
$\varphi_{2}>0,$ and we can set $\varphi_{2}=1$ without loss of generality.
Thus we obtain the $L_{FZ0}$ given in equation (\ref{eqnFZ0}). If $v$ can be
positive, then setting $\varphi_{1}=0$ is interpretable as fixing this shape
parameter value at a particular value.
\end{proof}

\bigskip

\begin{proof}
[Proof of Theorem \ref{thmCONSISTENCY}]The proof is based on Theorem 2.1 of
Newey and McFadden (1994). We only need to show that $E[L_{T}(\cdot)]$ is
uniquely minimized at $\mathbf{\theta}^{0}$, because the other assumptions of
Newey and McFadden's theorem are clearly satisfied. By Corollary (5.5) of
Fissler and Ziegel (2016), given Assumption 1(B)(iii) and the fact that our
choice of the objective function $L_{FZ0}$ satisfies the condition as in
Corollary (5.5) of Fissler and Ziegel (2016), we know that $\mathbb{E}\left[
L\left(  Y_{t},v_{t}\left(  \mathbf{\theta}\right)  ,e_{t}\left(
\mathbf{\theta}\right)  ;\alpha\right)  |\mathcal{F}_{t-1}\right]  $ is
uniquely minimized at $\left(  \mathrm{VaR}_{\alpha}(Y_{t}|\mathcal{F}%
_{t-1}),\mathrm{ES}_{\alpha}(Y_{t}|\mathcal{F}_{t-1})\right)  ,$ which equals
$\left(  v_{t}(\mathbf{\theta}^{0}),e_{t}(\mathbf{\theta}^{0})\right)  $ under
correct specification. Combining this assumption and Assumption 1(B)(iv), we
know that $\mathbf{\theta}^{0}$ is a unique minimizer of $\mathbb{E}%
[L_{T}(\cdot)]$, completing the proof.
\end{proof}

\bigskip

\begin{proof}
[Outline of proof of Theorem \ref{thmASYMPNORM}]We consider the population
objective function $\lambda_{T}(\mathbf{\theta})=T^{-1}\sum_{t=1}%
^{T}\mathbb{E}\left[  g_{t}(\mathbf{\theta})\right]  ,$ and take a mean-value
expansion of $\lambda_{T}(\mathbf{\hat{\theta}})$ around $\mathbf{\theta}%
^{0}.$ We show in Lemma \ref{lemmaLAM} that:%
\begin{align*}
\sqrt{T}(\mathbf{\hat{\theta}-\theta}^{0})  & =-\Lambda_{T}^{-1}%
(\mathbf{\theta}^{0})\frac{1}{\sqrt{T}}\sum_{t=1}^{T}g_{t}(\mathbf{\theta}%
^{0})+o_{p}(1)\\
\text{where \ \ }\Lambda_{T}(\mathbf{\theta}^{\ast})  & =T^{-1}\sum_{t=1}%
^{T}\left.  \frac{\partial\mathbb{E}\left[  g_{t}(\mathbf{\theta})\right]
}{\partial\mathbf{\theta}}\right\vert _{\mathbf{\theta=\theta}^{\ast}}%
\end{align*}
In the supplemental appendix we prove Lemma \ref{lemmaLAM} by building on and
extending Weiss (1991), who extends Huber (1967) to non-\textit{iid} data. We
draw on Weiss' Lemma A.1, and we verify that all five assumptions (N1-N5 in
his notation) for that lemma are satisfied: N1, N2 and N5 are obviously
satisfied given our Assumptions 1-2, and we show in Lemmas \ref{lemmaN3i} -
\ref{lemmaN4} that assumptions N3 and N4 are satistfied. Assumption 2(F)
allows a CLT to be applied: the asymptotic covariance matrix is $\mathbf{A}%
_{T}=\mathbb{E}\left[  T^{-1}\sum_{t=1}^{T}g_{t}(\mathbf{\theta}^{0}%
)g_{t}(\mathbf{\theta}^{0})^{\prime}\right]  ,$ and we denote $\Lambda
_{T}(\mathbf{\theta}^{0})$ as $\mathbf{D}_{T},$ leading to the stated result.
\end{proof}

\bigskip

\begin{proof}
[Proof of Theorem \ref{thmVCV}]Given Assumption 3B(i) and the result in
Theorem 1, the proof that $\mathbf{\hat{A}}_{T}-\mathbf{A}_{T}\overset
{p}{\rightarrow}\mathbf{0}$ is standard and omitted. Next, define
\[
\mathbf{\tilde{D}}_{T}=T^{-1}\sum_{t=1}^{T}\{(2c_{T})^{-1}\mathbf{1}%
\{|y_{t}-v_{t}(\mathbf{\theta}^{0})|<c_{T}\}\frac{1}{-e_{t}(\mathbf{\theta
}^{0})\alpha}\nabla v_{t}(\mathbf{\theta}^{0})^{\prime}\nabla v_{t}%
(\mathbf{\theta}^{0})+\frac{1}{e_{t}(\mathbf{\theta}^{0})^{2}}\nabla
e_{t}(\mathbf{\theta}^{0})^{\prime}\nabla e_{t}(\mathbf{\theta}^{0})\}
\]
To prove the result we will show that $\mathbf{\hat{D}}_{T}-\mathbf{\tilde{D}%
}_{T}=o_{p}(1)$ and $\mathbf{\tilde{D}}_{T}-\mathbf{D}_{T}=o_{p}(1)$.
\ Firstly, consider%
\begin{align*}
\Vert\mathbf{\hat{D}}_{T}-\mathbf{\tilde{D}}_{T}\Vert &  \leq\left\Vert
(2Tc_{T})^{-1}\right. \\
\times\sum_{t=1}^{T}\{ &  (\mathbf{1}\{|y_{t}-v_{t}(\mathbf{\hat{\theta}}%
_{T})|<c_{T}\}-\mathbf{1}\{|y_{t}-v_{t}(\mathbf{\theta}^{0})|<c_{T}\})\frac
{1}{-e_{t}(\mathbf{\hat{\theta}}_{T})\alpha}\nabla v_{t}(\mathbf{\hat{\theta}%
}_{T})^{\prime}\nabla v_{t}(\mathbf{\hat{\theta}}_{T})\\
+ &  \mathbf{1}\left\{  |y_{t}-v_{t}(\mathbf{\theta}^{0})|<c_{T}\right\}
\frac{1}{-e_{t}(\mathbf{\hat{\theta}}_{T})\alpha}\left(  \nabla v_{t}%
(\mathbf{\hat{\theta}}_{T})-\nabla v_{t}(\mathbf{\theta}^{0})\right)
^{\prime}\nabla v_{t}(\mathbf{\hat{\theta}}_{T})\\
+ &  \mathbf{1}\{|y_{t}-v_{t}(\mathbf{\theta}^{0})|<c_{T}\}\left(  \frac
{1}{-\alpha e_{t}(\mathbf{\hat{\theta}}_{T})}-\frac{1}{-\alpha e_{t}%
(\mathbf{\theta}^{0})}\right)  \nabla v_{t}\mathbf{(\theta}^{0})^{\prime
}\nabla v_{t}(\mathbf{\hat{\theta}}_{T})\\
+ &  \mathbf{1}\{|y_{t}-v_{t}(\mathbf{\theta}^{0})|<c_{T}\}\frac{1}{-\alpha
e_{t}(\mathbf{\theta}^{0})}\nabla v_{t}(\mathbf{\theta}^{0})^{\prime}(\nabla
v_{t}(\mathbf{\hat{\theta}}_{T})-\nabla v_{t}(\mathbf{\theta}^{0}))\\
&  \left.  +\frac{c_{T}-\hat{c}_{T}}{c_{T}}\mathbf{1}\{|y_{t}-v_{t}%
(\mathbf{\theta}^{0})|<c_{T}\}\frac{1}{-e_{t}(\mathbf{\theta}^{0})\alpha
}\nabla v_{t}(\mathbf{\theta}^{0})^{\prime}\nabla v_{t}(\mathbf{\theta}%
^{0})\}\right\Vert \\
&  +T^{-1}\sum_{t=1}^{T}\left\Vert \frac{1}{e_{t}(\mathbf{\hat{\theta}}%
_{T})^{2}}\nabla e_{t}(\mathbf{\hat{\theta}}_{T})^{\prime}\nabla
e_{t}(\mathbf{\hat{\theta}}_{T})-\frac{1}{e_{t}\mathbf{(\theta}^{0})^{2}%
}\nabla e_{t}(\mathbf{\theta}^{0})^{\prime}\nabla e_{t}(\mathbf{\theta}%
^{0})\right\Vert
\end{align*}
The last line above was shown to be $o_{p}(1)$ in the proof of Theorem 2. The
difficult quantity in the first term (over the first six lines above) is the
indicator, and following the same steps as in Engle and Manganelli (2004a),
that term is also $o_{p}(1).$ Next, consider $\mathbf{\tilde{D}}%
_{T}\mathbf{-D}_{T}$:%
\begin{align*}
\mathbf{\tilde{D}}_{T}\mathbf{-D}_{T}  & =\frac{1}{2Tc_{T}}\sum_{t=1}%
^{T}\left(  \mathbf{1}\left\{  \left\vert Y_{t}-v_{t}(\mathbf{\theta}%
^{0})\right\vert <c_{T}\right\}  -\mathbb{E}\left[  \mathbf{1}\left\{
\left\vert Y_{t}-v_{t}\left(  \mathbf{\theta}^{0}\right)  \right\vert
<c_{T}\right\}  |\mathcal{F}_{t-1}\right]  \right) \\
& \times\frac{\nabla^{\prime}v_{t}(\mathbf{\theta}^{0})\nabla v_{t}%
(\mathbf{\theta}^{0})}{-e_{t}(\mathbf{\theta}^{0})\alpha}\\
& +\frac{1}{T}\sum_{t=1}^{T}\left\{  \frac{1}{2c_{T}}\mathbb{E}[\mathbf{1}%
\{|Y_{t}-v_{t}(\mathbf{\theta}^{0})|<c_{T}\}|\mathcal{F}_{t-1}]\frac{1}%
{-e_{t}(\mathbf{\theta}^{0})\alpha}\nabla^{\prime}v_{t}(\mathbf{\theta}%
^{0})\nabla v_{t}(\mathbf{\theta}^{0})\right. \\
& \left.  -\mathbb{E}\left[  \frac{f_{t}(v_{t}(\mathbf{\theta}^{0}))}%
{-e_{t}(\mathbf{\theta}^{0})\alpha}\nabla^{\prime}v_{t}(\mathbf{\theta}%
^{0})\nabla v_{t}(\mathbf{\theta}^{0})\right]  \right\}
\end{align*}
Following Engle and Manganelli (2004a), assumptions 1-3 are sufficient to
show\textbf{\ }$\mathbf{\tilde{D}}_{T}-\mathbf{D}_{T}=o_{p}(1)$ and the result follows.
\end{proof}

\bigskip

{\LARGE Appendix B:\ Derivations}

\bigskip

{\Large Appendix B.1: Generic calculations for the FZ0 loss function}

The FZ0 loss function is:%

\begin{equation}
L_{FZ0}\left(  Y,v,e;\alpha\right)  =-\frac{1}{\alpha e}\mathbf{1}\left\{
Y\leq v\right\}  \left(  v-Y\right)  +\frac{v}{e}+\log\left(  -e\right)  -1
\end{equation}
Note that this is \textit{not} homogeneous, as for any $k>0,~L_{FZ0}\left(
kY,kv,ke;\alpha\right)  =L_{FZ0}\left(  Y,v,e;\alpha\right)  +\log\left(
k\right)  $, but this loss function generates loss \textit{differences }that
are homogenous of degree zero, as the additive additional term above drops out.

We will frequently use the first derivatives of this loss function, and the
second derivatives of the expected loss for an absolutely continuous random
variable with density $f$ and CDF $F$. These are (for $v\neq y$):%
\begin{align}
\nabla_{v}  & \equiv\frac{\partial L_{FZ0}\left(  Y,v,e;\alpha\right)
}{\partial v}=-\frac{1}{\alpha e}\left(  \mathbf{1}\left\{  Y\leq v\right\}
-\alpha\right)  \equiv\frac{1}{\alpha ve}\lambda_{v}\\
\nabla_{e}  & \equiv\frac{\partial L_{FZ0}\left(  Y,v,e;\alpha\right)
}{\partial e}\\
& =\frac{1}{\alpha e^{2}}\mathbf{1}\left\{  Y\leq v\right\}  \left(
v-Y\right)  -\frac{v}{e^{2}}+\frac{1}{e}\nonumber\\
& =\frac{v}{\alpha e^{2}}\left(  \mathbf{1}\left\{  Y\leq v\right\}
-\alpha\right)  -\frac{1}{e^{2}}\left(  \frac{1}{\alpha}\mathbf{1}\left\{
Y\leq v\right\}  Y-e\right) \nonumber\\
& \equiv\frac{-1}{\alpha e^{2}}\left(  \lambda_{v}+\alpha\lambda_{e}\right)
\nonumber
\end{align}
where%
\begin{align}
\lambda_{v}  & \equiv-v\left(  \mathbf{1}\left\{  Y\leq v\right\}
-\alpha\right) \\
\lambda_{e}  & \equiv\frac{1}{\alpha}\mathbf{1}\left\{  Y\leq v\right\}  Y-e
\end{align}
and%
\begin{align}
\frac{\partial^{2}\mathbb{E}\left[  L_{FZ0}\left(  Y,v,e;\alpha\right)
\right]  }{\partial v^{2}}  & =-\frac{1}{\alpha e}f\left(  v\right) \\
\frac{\partial^{2}\mathbb{E}\left[  L_{FZ0}\left(  Y,v,e;\alpha\right)
\right]  }{\partial v\partial e}  & =\frac{1}{\alpha e^{2}}\left(  F\left(
v\right)  -\alpha\right) \\
& =0\text{, at the true value of }\left(  v,e\right) \nonumber\\
\frac{\partial^{2}\mathbb{E}\left[  L_{FZ0}\left(  Y,v,e;\alpha\right)
\right]  }{\partial e^{2}}  & =\frac{1}{e^{2}}-\frac{2}{\alpha e^{3}}\left\{
\left(  F\left(  v\right)  -\alpha\right)  v-\left(  \mathbb{E}\left[
\mathbf{1}\left\{  Y\leq v\right\}  Y\right]  -\alpha e\right)  \right\} \\
& =\frac{1}{e^{2}}\text{, at the true value of }\left(  v,e\right) \nonumber
\end{align}

\bigskip

{\Large Appendix B.2: Derivations for the one-factor GAS\ model for ES and
VaR}

Here we present the calculations to compute $s_{t}$ and $I_{t}$ for this
model. Below we use:%
\begin{align}
\frac{\partial v}{\partial\kappa}  & =\frac{\partial^{2}v}{\partial\kappa^{2}%
}=a\exp\left\{  \kappa\right\}  =v\\
\frac{\partial e}{\partial\kappa}  & =\frac{\partial^{2}e}{\partial\kappa^{2}%
}=b\exp\left\{  \kappa\right\}  =e
\end{align}
And so we find (for $v_{t}\neq Y_{t}$)%
\begin{align}
s_{t}  & \equiv\frac{\partial L_{FZ0}\left(  Y_{t},v_{t},e_{t};\alpha\right)
}{\partial\kappa_{t}}\\
& =\frac{\partial L_{FZ0}\left(  Y_{t},v_{t},e_{t};\alpha\right)  }{\partial
v_{t}}\frac{\partial v_{t}}{\partial\kappa_{t}}+\frac{\partial L_{FZ0}\left(
Y_{t},v_{t},e_{t};\alpha\right)  }{\partial e_{t}}\frac{\partial e_{t}%
}{\partial\kappa_{t}}\nonumber\\
& =\left\{  -\frac{1}{\alpha e_{t}}\left(  \mathbf{1}\left\{  Y_{t}\leq
v_{t}\right\}  -\alpha\right)  \right\}  v_{t}\nonumber\\
& +\left\{  -\frac{1}{e_{t}^{2}}\left(  \frac{1}{\alpha}\mathbf{1}\left\{
Y_{t}\leq v_{t}\right\}  Y_{t}-e_{t}\right)  +\frac{v_{t}}{e_{t}^{2}}\frac
{1}{\alpha}\left(  \mathbf{1}\left\{  Y_{t}\leq v_{t}\right\}  -\alpha\right)
\right\}  e_{t}\nonumber\\
& =-\frac{1}{e_{t}}\left(  \frac{1}{\alpha}\mathbf{1}\left\{  Y_{t}\leq
v_{t}\right\}  Y_{t}-e_{t}\right) \\
& \equiv-\lambda_{et}/e_{t}%
\end{align}
Thus, the $\lambda_{vt}$ term drops out of $s_{t}$ and we are left with
$-\lambda_{et}/e_{t}.$

Next we calculate $I_{t}:$%
\begin{align}
I_{t}  & \equiv\frac{\partial^{2}\mathbb{E}_{t-1}\left[  L_{FZ0}\left(
Y_{t},v_{t},e_{t};\alpha\right)  \right]  }{\partial\kappa_{t}^{2}}\\
& =\frac{\partial^{2}\mathbb{E}_{t-1}\left[  L_{FZ0}\left(  Y_{t},v_{t}%
,e_{t};\alpha\right)  \right]  }{\partial v_{t}^{2}}\left(  \frac{\partial
v_{t}}{\partial\kappa_{t}}\right)  ^{2}+\frac{\partial^{2}\mathbb{E}%
_{t-1}\left[  L_{FZ0}\left(  Y_{t},v_{t},e_{t};\alpha\right)  \right]
}{\partial v_{t}\partial e_{t}}\frac{\partial v_{t}}{\partial\kappa_{t}%
}\nonumber\\
& +\frac{\partial^{2}\mathbb{E}_{t-1}\left[  L_{FZ0}\left(  Y_{t},v_{t}%
,e_{t};\alpha\right)  \right]  }{\partial e_{t}^{2}}\left(  \frac{\partial
e_{t}}{\partial\kappa_{t}}\right)  ^{2}+\frac{\partial^{2}\mathbb{E}%
_{t-1}\left[  L_{FZ0}\left(  Y_{t},v_{t},e_{t};\alpha\right)  \right]
}{\partial v_{t}\partial e_{t}}\frac{\partial e_{t}}{\partial\kappa_{t}%
}\nonumber\\
& +\frac{\partial\mathbb{E}_{t-1}\left[  L_{FZ0}\left(  Y_{t},v_{t}%
,e_{t};\alpha\right)  \right]  }{\partial v_{t}}\frac{\partial^{2}v_{t}%
}{\partial\kappa_{t}^{2}}+\frac{\partial\mathbb{E}_{t-1}\left[  L_{FZ0}\left(
Y_{t},v_{t},e_{t};\alpha\right)  \right]  }{\partial e_{t}}\frac{\partial
^{2}e_{t}}{\partial\kappa_{t}^{2}}\nonumber
\end{align}
But note that under correct specification,
\begin{equation}
\frac{\partial^{2}\mathbb{E}_{t-1}\left[  L\left(  Y_{t},v_{t},e_{t}%
;\alpha\right)  \right]  }{\partial v_{t}\partial e_{t}}=\frac{\partial
\mathbb{E}_{t-1}\left[  L\left(  Y_{t},v_{t},e_{t};\alpha\right)  \right]
}{\partial v_{t}}=\frac{\partial\mathbb{E}_{t-1}\left[  L\left(  Y_{t}%
,v_{t},e_{t};\alpha\right)  \right]  }{\partial e_{t}}=0
\end{equation}
and so the Hessian simplifies to:%
\begin{align}
I_{t}  & =\frac{\partial^{2}\mathbb{E}_{t-1}\left[  L_{FZ0}\left(  Y_{t}%
,v_{t},e_{t};\alpha\right)  \right]  }{\partial v_{t}^{2}}\left(
\frac{\partial v_{t}}{\partial\kappa_{t}}\right)  ^{2}+\frac{\partial
^{2}\mathbb{E}_{t-1}\left[  L_{FZ0}\left(  Y_{t},v_{t},e_{t};\alpha\right)
\right]  }{\partial e_{t}^{2}}\left(  \frac{\partial e_{t}}{\partial\kappa
_{t}}\right)  ^{2}\\
& =-\frac{1}{\alpha e_{t}}f_{t}\left(  v_{t}\right)  v_{t}^{2}+1\\
& =\frac{\alpha-k_{\alpha}a_{\alpha}}{\alpha}\text{, \ since }f_{t}\left(
v_{t}\right)  =\frac{k_{\alpha}}{v_{t}}\text{ and }\frac{v_{t}}{e_{t}%
}=a_{\alpha}\text{, for this DGP.}%
\end{align}
Thus although the Hessian could\textit{\ }vary with time, as it is a
derivative of the conditional expected loss, in this specification it
simplifies to a constant. \bigskip

{\Large Appendix B.3: ES and VaR in location-scale models}

Dynamic location-scale models are widely used for asset returns and in this
section we consider what such a specification implies for the dynamics of ES
and VaR. Consider the following:%
\begin{equation}
Y_{t}=\mu_{t}+\sigma_{t}\eta_{t}\text{, \ \ }\eta_{t}\thicksim iid~F_{\eta
}\left(  0,1\right) \label{eqnIID}%
\end{equation}
where, for example, $\mu_{t}$ is some ARMA\ model and $\sigma_{t}^{2}$ is some
GARCH model. For asset returns that follow equation (\ref{eqnIID}) we have:%
\begin{align}
v_{t}  & =\mu_{t}+a\sigma_{t}\text{, \ \ where \ }a=F_{\eta}^{-1}\left(
\alpha\right) \\
e_{t}  & =\mu_{t}+b\sigma_{t}\text{, \ \ where \ }b=\mathbb{E}\left[  \eta
_{t}|\eta_{t}\leq a\right] \nonumber
\end{align}
and we we can recover $\left(  \mu_{t},\sigma_{t}\right)  $ from $\left(
v_{t},e_{t}\right)  $:%
\begin{equation}
\left[
\begin{array}
[c]{c}%
\mu_{t}\\
\sigma_{t}%
\end{array}
\right]  =\frac{1}{b-a}\left[
\begin{array}
[c]{cc}%
b & -a\\
-1 & 1
\end{array}
\right]  \left[
\begin{array}
[c]{c}%
v_{t}\\
e_{t}%
\end{array}
\right]
\end{equation}
Thus under the conditional location-scale assumption, we can back out the
conditional mean and variance from the VaR and\ ES. Next note that if $\mu
_{t}=0~\forall$~$t,$ then $v_{t}=c\cdot e_{t}$, \ where $c=a/b\in\left(
0,1\right)  $. Daily asset returns often have means that are close to zero,
and so this restriction is one that may be plausible in the data. A related,
though less plausible, restriction is that $\sigma_{t}=\bar{\sigma}~\forall
$~$t,$ and in that case we have the simplification that $v_{t}=d+e_{t}$, where
$d=\left(  a-b\right)  \bar{\sigma}>0.\bigskip$

{\LARGE Appendix C:\ Estimation using the FZ0 loss function}

\bigskip

The FZ0 loss function, equation (\ref{eqnFZ0}), involves the indicator
function $\mathbf{1}\left\{  Y_{t}\leq v_{t}\right\}  $ and so necessitates
the use of a numerical search algorithm that does not rely on
differentiability of the objective function; we use the function
\texttt{fminsearch} in Matlab. However, in preliminary simulation analyses we
found that this algorithm was sensitive to the starting values used in the
search. To overcome this, we initially consider a \textquotedblleft
smoothed\textquotedblright\ version of the FZ0 loss function, where we replace
the indicator variable with a Logistic function:%
\begin{align}
\tilde{L}_{FZ0}\left(  Y,v,e;\alpha,\tau\right)   & =-\frac{1}{\alpha e}%
\Gamma\left(  Y_{t},v_{t};\tau\right)  \left(  v-Y\right)  +\frac{v}{e}%
+\log\left(  -e\right)  -1\\
\text{where \ \ }\Gamma\left(  Y_{t},v_{t};\tau\right)   & \equiv\frac
{1}{1+\exp\left\{  \tau\left(  Y_{t}-v_{t}\right)  \right\}  }\text{, \ for
}\tau>0
\end{align}
where $\tau$ is the smoothing parameter, and the smoothing function $\Gamma$
converges to the indicator function as $\tau\rightarrow\infty.$ In GAS\ models
that involve an indicator function in the forcing variable, we alter the
forcing variable in the same way, to ensure that the objective function as a
function of $\mathbf{\theta}$ is differentiable. In these cases the loss
function \textit{and} the model itself are slightly altered through this smoothing.

In our empirical implementation, we obtain \textquotedblleft
smart\textquotedblright\ starting values by first estimating the model using
the \textquotedblleft smoothed FZ0\textquotedblright\ loss function with
$\tau=5.$ This choice of $\tau$ gives some smoothing for values of $Y_{t}$
that are roughly within $\pm1$ of $v_{t}.$ Call the resulting parameter
estimate $\mathbf{\tilde{\theta}}_{T}^{\left(  5\right)  }.$ Since this
objective function is differentiable, we can use more familiar gradient-based
numerical search algorithms, such as \texttt{fminunc}\ or \texttt{fmincon} in
Matlab, which are often less sensitive to starting values. We then re-estimate
the model, using $\mathbf{\tilde{\theta}}_{T}^{\left(  5\right)  }$ as the
starting value, setting $\tau=20$ and obtain $\mathbf{\tilde{\theta}}%
_{T}^{\left(  20\right)  }.$ This value of $\tau$ smoothes values of $Y_{t}$
within roughly $\pm0.25$ of $v_{t},$ and so this objective function is closer
to the true objective function. Finally, we use $\mathbf{\tilde{\theta}}%
_{T}^{\left(  20\right)  }$ as the starting value in the optimization of the
actual FZ0 objective function, with no artificial smoothing, using the
function \texttt{fminsearch}, and obtain $\mathbf{\hat{\theta}}_{T}$. We found
that this approach largely eliminated the sensitivity to starting values.

\bigskip\bigskip\bigskip%

\def\baselinestretch{1.0}\small\normalsize

\bigskip

\pagebreak

\begin{center}
\textbf{Table 1: Simulation results for Normal innovations}

\bigskip%

\begin{tabular}
[c]{rccccrcccc}\hline
& \multicolumn{4}{c}{$T=2500$} &  & \multicolumn{4}{c}{$T=5000$}%
\\\cline{2-5}\cline{7-10}
& $\beta$ & $\gamma$ & $b_{\alpha}$ & $c_{\alpha}$ &  & $\beta$ & $\gamma$ &
$b_{\alpha}$ & $c_{\alpha}$\\\cline{2-5}\cline{7-10}%
\multicolumn{1}{l}{} &  &  &  &  &  &  &  & \multicolumn{1}{l}{} & \\
\multicolumn{1}{l}{} & \multicolumn{9}{c}{$\alpha=0.01$}\\
True & 0.900 & 0.050 & -2.665 & 0.873 &  & 0.900 & 0.050 & -2.665 & 0.873\\
Median & 0.901 & 0.049 & -2.615 & 0.882 &  & 0.899 & 0.049 & -2.671 & 0.877\\
Avg bias & -0.017 & 0.015 & -0.108 & 0.008 &  & -0.011 & 0.006 & -0.089 &
0.004\\
St dev & 0.077 & 0.076 & 1.095 & 0.022 &  & 0.049 & 0.033 & 0.805 & 0.015\\
Coverage & 0.868 & 0.827 & 0.875 & 0.919 &  & 0.884 & 0.876 & 0.888 &
0.937\\\hline
&  &  &  &  &  &  &  &  & \\
\multicolumn{1}{l}{} & \multicolumn{9}{c}{$\alpha=0.025$}\\
True & 0.900 & 0.050 & -2.338 & 0.838 &  & 0.900 & 0.050 & -2.338 & 0.838\\
Median & 0.899 & 0.047 & -2.329 & 0.842 &  & 0.897 & 0.048 & -2.392 & 0.841\\
Avg bias & -0.017 & 0.007 & -0.137 & 0.004 &  & -0.011 & 0.002 & -0.111 &
0.002\\
St dev & 0.066 & 0.044 & 0.852 & 0.017 &  & 0.050 & 0.024 & 0.656 & 0.012\\
Coverage & 0.898 & 0.870 & 0.911 & 0.931 &  & 0.912 & 0.888 & 0.925 &
0.923\\\hline
&  &  &  &  &  &  &  &  & \\
\multicolumn{1}{l}{} & \multicolumn{9}{c}{$\alpha=0.05$}\\
True & 0.900 & 0.050 & -2.063 & 0.797 &  & 0.900 & 0.050 & -2.063 & 0.797\\
Median & 0.901 & 0.048 & -2.051 & 0.800 &  & 0.899 & 0.049 & -2.094 & 0.799\\
Avg bias & -0.013 & 0.005 & -0.097 & 0.002 &  & -0.008 & 0.002 & -0.081 &
0.001\\
St dev & 0.062 & 0.046 & 0.707 & 0.015 &  & 0.041 & 0.021 & 0.511 & 0.010\\
Coverage & 0.913 & 0.874 & 0.916 & 0.947 &  & 0.923 & 0.907 & 0.927 &
0.948\\\hline
&  &  &  &  &  &  &  &  & \\
\multicolumn{1}{l}{} & \multicolumn{9}{c}{$\alpha=0.10$}\\
True & 0.900 & 0.050 & -1.755 & 0.730 &  & 0.900 & 0.050 & -1.755 & 0.730\\
Median & 0.900 & 0.048 & -1.769 & 0.730 &  & 0.898 & 0.048 & -1.778 & 0.730\\
Avg bias & -0.015 & 0.006 & -0.103 & 0.000 &  & -0.009 & 0.001 & -0.072 &
0.000\\
St dev & 0.065 & 0.052 & 0.623 & 0.013 &  & 0.040 & 0.020 & 0.435 & 0.009\\
Coverage & 0.917 & 0.883 & 0.925 & 0.954 &  & 0.922 & 0.902 & 0.934 &
0.960\\\hline
&  &  &  &  &  &  &  &  & \\
\multicolumn{1}{l}{} & \multicolumn{9}{c}{$\alpha=0.20$}\\
True & 0.900 & 0.050 & -1.400 & 0.601 &  & 0.900 & 0.050 & -1.400 & 0.601\\
Median & 0.898 & 0.048 & -1.391 & 0.602 &  & 0.899 & 0.048 & -1.417 & 0.602\\
Avg bias & -0.017 & 0.008 & -0.091 & 0.000 &  & -0.010 & 0.002 & -0.064 &
0.000\\
St dev & 0.078 & 0.072 & 0.547 & 0.014 &  & 0.044 & 0.022 & 0.374 & 0.010\\
Coverage & 0.925 & 0.881 & 0.934 & 0.948 &  & 0.941 & 0.923 & 0.945 &
0.954\\\hline
\end{tabular}

\end{center}

\textit{Notes:} This table presents results from 1000 replications of the
estimation of VaR and\ ES from a GARCH(1,1)\ DGP with standard Normal
innovations. Details are described in\ Section \ref{sSIMULATION}. The top row
of each panel presents the true values of the parameters. The second, third,
and fourth rows present the median estimated parameters, the average bias, and
the standard deviation (across simulations) of the estimated parameters. The
last row of each panel presents the coverage rates for 95\% confidence
intervals constructed using estimated standard errors.

\bigskip

\begin{center}
\pagebreak

\textbf{Table 2: Simulation results for skew }$\mathbf{t}$%
\textbf{\ innovations}

\bigskip%

\begin{tabular}
[c]{rccccrcccc}\hline
& \multicolumn{4}{c}{$T=2500$} &  & \multicolumn{4}{c}{$T=5000$}%
\\\cline{2-5}\cline{7-10}
& $\beta$ & $\gamma$ & $b_{\alpha}$ & $c_{\alpha}$ &  & $\beta$ & $\gamma$ &
$b_{\alpha}$ & $c_{\alpha}$\\\cline{2-5}\cline{7-10}%
\multicolumn{1}{l}{} &  &  &  &  &  &  &  & \multicolumn{1}{l}{} & \\
\multicolumn{1}{l}{} & \multicolumn{9}{c}{$\alpha=0.01$}\\
True & 0.900 & 0.050 & -4.506 & 0.730 &  & 0.900 & 0.050 & -4.506 & 0.730\\
Median & 0.893 & 0.049 & -4.376 & 0.750 &  & 0.895 & 0.048 & -4.562 & 0.741\\
Avg bias & -0.047 & 0.038 & -0.399 & 0.018 &  & -0.028 & 0.014 & -0.340 &
0.009\\
St dev & 0.150 & 0.134 & 2.687 & 0.048 &  & 0.094 & 0.065 & 1.983 & 0.034\\
Coverage & 0.797 & 0.797 & 0.809 & 0.894 &  & 0.837 & 0.853 & 0.839 &
0.936\\\hline
&  &  &  &  &  &  &  &  & \\
\multicolumn{1}{l}{} & \multicolumn{9}{c}{$\alpha=0.025$}\\
True & 0.900 & 0.050 & -3.465 & 0.695 &  & 0.900 & 0.050 & -3.465 & 0.695\\
Median & 0.895 & 0.047 & -3.448 & 0.705 &  & 0.896 & 0.048 & -3.520 & 0.701\\
Avg bias & -0.028 & 0.014 & -0.254 & 0.008 &  & -0.017 & 0.005 & -0.198 &
0.004\\
St dev & 0.101 & 0.069 & 1.591 & 0.034 &  & 0.068 & 0.033 & 1.192 & 0.023\\
Coverage & 0.855 & 0.835 & 0.877 & 0.921 &  & 0.874 & 0.893 & 0.887 &
0.939\\\hline
&  &  &  &  &  &  &  &  & \\
\multicolumn{1}{l}{} & \multicolumn{9}{c}{$\alpha=0.05$}\\
True & 0.900 & 0.050 & -2.767 & 0.651 &  & 0.900 & 0.050 & -2.767 & 0.651\\
Median & 0.896 & 0.048 & -2.760 & 0.656 &  & 0.898 & 0.048 & -2.795 & 0.654\\
Avg bias & -0.021 & 0.007 & -0.187 & 0.005 &  & -0.011 & 0.003 & -0.114 &
0.003\\
St dev & 0.081 & 0.049 & 1.085 & 0.025 &  & 0.053 & 0.025 & 0.782 & 0.017\\
Coverage & 0.906 & 0.883 & 0.921 & 0.937 &  & 0.916 & 0.904 & 0.922 &
0.951\\\hline
&  &  &  &  &  &  &  &  & \\
\multicolumn{1}{l}{} & \multicolumn{9}{c}{$\alpha=0.10$}\\
True & 0.900 & 0.050 & -2.122 & 0.577 &  & 0.900 & 0.050 & -2.122 & 0.577\\
Median & 0.897 & 0.048 & -2.121 & 0.579 &  & 0.898 & 0.048 & -2.140 & 0.578\\
Avg bias & -0.017 & 0.006 & -0.125 & 0.003 &  & -0.008 & 0.002 & -0.069 &
0.002\\
St dev & 0.066 & 0.045 & 0.745 & 0.020 &  & 0.040 & 0.022 & 0.510 & 0.014\\
Coverage & 0.931 & 0.900 & 0.937 & 0.949 &  & 0.926 & 0.925 & 0.927 &
0.947\\\hline
&  &  &  &  &  &  &  &  & \\
\multicolumn{1}{l}{} & \multicolumn{9}{c}{$\alpha=0.20$}\\
True & 0.900 & 0.050 & -1.514 & 0.431 &  & 0.900 & 0.050 & -1.514 & 0.431\\
Median & 0.899 & 0.050 & -1.485 & 0.432 &  & 0.899 & 0.049 & -1.503 & 0.432\\
Avg bias & -0.019 & 0.006 & -0.089 & 0.001 &  & -0.008 & 0.002 & -0.049 &
0.001\\
St dev & 0.089 & 0.047 & 0.618 & 0.018 &  & 0.042 & 0.022 & 0.380 & 0.012\\
Coverage & 0.916 & 0.888 & 0.922 & 0.938 &  & 0.929 & 0.916 & 0.940 &
0.944\\\hline
\end{tabular}

\end{center}

\textit{Notes:} This table presents results from 1000 replications of the
estimation of VaR and\ ES from a GARCH(1,1)\ DGP with skew $t$ innovations.
Details are described in\ Section \ref{sSIMULATION}. The top row of each panel
presents the true values of the parameters. The second, third, and fourth rows
present the median estimated parameters, the average bias, and the standard
deviation (across simulations) of the estimated parameters. The last row of
each panel presents the coverage rates for 95\% confidence intervals
constructed using estimated standard errors.

\bigskip

\pagebreak

\begin{center}
\textbf{Table 3: Sampling variation of FZ estimation }

\textbf{relative to (Q)MLE and CAViaR}

\bigskip%

\begin{tabular}
[c]{rccccccccccc}\hline
& \multicolumn{5}{c}{\textit{Normal innovations}} &  &
\multicolumn{5}{c}{\textit{Skew t innovations}}\\\cline{2-6}\cline{8-12}
& \multicolumn{2}{c}{$T=2500$} &  & \multicolumn{2}{c}{$T=5000$} &  &
\multicolumn{2}{c}{$T=2500$} &  & \multicolumn{2}{c}{$T=5000$}\\\cline{2-3}%
\cline{5-6}\cline{8-9}\cline{11-12}%
\multicolumn{1}{l}{$\alpha$} & $\beta$ & $\gamma$ &  & $\beta$ & $\gamma$ &  &
$\beta$ & $\gamma$ &  & $\beta$ & $\gamma$\\\hline
\multicolumn{1}{l}{} & \multicolumn{1}{r}{} & \multicolumn{1}{r}{} &
\multicolumn{1}{r}{} & \multicolumn{1}{r}{} & \multicolumn{1}{r}{} &
\multicolumn{1}{r}{} & \multicolumn{1}{r}{} & \multicolumn{1}{r}{} &
\multicolumn{1}{r}{} & \multicolumn{1}{r}{} & \multicolumn{1}{r}{}\\
\multicolumn{5}{l}{\textbf{Panel A:\ FZ/(Q)ML}} & \multicolumn{1}{r}{} &
\multicolumn{1}{r}{} & \multicolumn{1}{r}{} & \multicolumn{1}{r}{} &
\multicolumn{1}{r}{} & \multicolumn{1}{r}{} & \multicolumn{1}{r}{}\\
\multicolumn{1}{l}{0.01} & 1.209 & 5.940 &  & 1.701 & 3.731 &  & 1.577 &
4.830 &  & 2.533 & 3.723\\
\multicolumn{1}{l}{0.025} & 1.034 & 3.394 &  & 1.764 & 2.694 &  & 1.055 &
2.485 &  & 1.853 & 1.905\\
\multicolumn{1}{l}{0.05} & 0.980 & 3.576 &  & 1.431 & 2.377 &  & 0.850 &
1.784 &  & 1.426 & 1.458\\
\multicolumn{1}{l}{0.10} & 1.021 & 4.074 &  & 1.406 & 2.302 &  & 0.698 &
1.627 &  & 1.095 & 1.250\\
\multicolumn{1}{l}{0.20} & 1.224 & 5.558 &  & 1.543 & 2.497 &  & 0.939 &
1.710 &  & 1.145 & 1.242\\\hline
\multicolumn{1}{l}{} & \multicolumn{1}{r}{} & \multicolumn{1}{r}{} &
\multicolumn{1}{r}{} & \multicolumn{1}{r}{} & \multicolumn{1}{r}{} &
\multicolumn{1}{r}{} & \multicolumn{1}{r}{} & \multicolumn{1}{r}{} &
\multicolumn{1}{r}{} & \multicolumn{1}{r}{} & \multicolumn{1}{r}{}\\
\multicolumn{5}{l}{\textbf{Panel B:\ FZ/CAViaR}} & \multicolumn{1}{r}{} &
\multicolumn{1}{r}{} & \multicolumn{1}{r}{} & \multicolumn{1}{r}{} &
\multicolumn{1}{r}{} & \multicolumn{1}{r}{} & \multicolumn{1}{r}{}\\
\multicolumn{1}{l}{0.01} & 0.982 & 1.162 &  & 0.951 & 0.975 &  & 1.062 &
1.384 &  & 0.912 & 1.465\\
\multicolumn{1}{l}{0.025} & 0.965 & 1.139 &  & 0.971 & 1.042 &  & 0.976 &
1.030 &  & 0.974 & 0.997\\
\multicolumn{1}{l}{0.05} & 0.925 & 1.238 &  & 0.910 & 0.930 &  & 0.885 &
0.819 &  & 0.920 & 0.903\\
\multicolumn{1}{l}{0.10} & 0.940 & 1.283 &  & 0.847 & 0.827 &  & 0.831 &
0.903 &  & 0.816 & 0.819\\
\multicolumn{1}{l}{0.20} & 0.855 & 0.671 &  & 0.703 & 0.510 &  & 0.736 &
0.437 &  & 0.503 & 0.515\\\hline
\end{tabular}

\end{center}

\bigskip

\textit{Notes:} This table presents the ratio of cross-simulation standard
deviations of parameter estimates obtained by FZ loss minimization and (Q)MLE
(Panel A), and CAViaR (Panel B). We consider only the parameters that are
common to these three estimation methods, namely the GARCH(1,1) parameters
$\beta$ and $\gamma.$ Ratios greater than one indicate the FZ estimator is
more variable than the alternative estimation method; ratios less than one
indicate the opposite.

\bigskip

\pagebreak%

\begin{landscape}%

\begin{center}
\textbf{Table 4: Mean absolute errors for VaR and ES estimates}

\bigskip%

\begin{tabular}
[c]{rccccccccccccccc}\hline
& \multicolumn{7}{c}{\textit{Normal innovations}} &  &
\multicolumn{7}{c}{\textit{Skew t innovations}}\\\cline{2-8}\cline{10-16}
& \multicolumn{3}{c}{\textbf{VaR}} &  & \multicolumn{3}{c}{\textbf{ES}} &  &
\multicolumn{3}{c}{\textbf{VaR}} &  & \multicolumn{3}{c}{\textbf{ES}%
}\\\cline{2-4}\cline{6-8}\cline{10-12}\cline{14-16}%
\multicolumn{1}{l}{} & \textit{MAE} & \multicolumn{2}{c}{\textit{MAE ratio}} &
& \textit{MAE} & \multicolumn{2}{c}{\textit{MAE ratio}} &  & \textit{MAE} &
\multicolumn{2}{c}{\textit{MAE ratio}} &  & \textit{MAE} &
\multicolumn{2}{c}{\textit{MAE ratio}}\\\cline{2-4}\cline{6-8}\cline{10-12}%
\cline{14-16}%
\multicolumn{1}{l}{$\alpha$} & MLE & CAViaR & FZ &  & MLE & CAViaR & FZ &  &
QMLE & CAViaR & FZ &  & QMLE & CAViaR & FZ\\\hline
\multicolumn{1}{l}{} &  &  &  &  &  &  &  &  &  &  &  &  &  &  & \\
\multicolumn{3}{l}{\textbf{Panel A: }$T=2500$} & \multicolumn{1}{r}{} &
\multicolumn{1}{r}{} & \multicolumn{1}{r}{} & \multicolumn{1}{r}{} &
\multicolumn{1}{r}{} & \multicolumn{1}{r}{} & \multicolumn{1}{r}{} &
\multicolumn{1}{r}{} & \multicolumn{1}{r}{} & \multicolumn{1}{r}{} &
\multicolumn{1}{r}{} & \multicolumn{1}{r}{} & \multicolumn{1}{r}{}\\
\multicolumn{1}{l}{0.01} & 0.069 & 1.368 & 1.369 &  & 0.084 & 1.487 & 1.345 &
& 0.196 & 1.327 & 1.381 &  & 0.342 & 1.249 & 1.252\\
\multicolumn{1}{l}{0.025} & 0.055 & 1.305 & 1.288 &  & 0.064 & 1.341 & 1.290 &
& 0.120 & 1.228 & 1.244 &  & 0.205 & 1.166 & 1.166\\
\multicolumn{1}{l}{0.05} & 0.043 & 1.302 & 1.271 &  & 0.051 & 1.332 & 1.289 &
& 0.084 & 1.193 & 1.166 &  & 0.141 & 1.154 & 1.129\\
\multicolumn{1}{l}{0.10} & 0.034 & 1.322 & 1.253 &  & 0.042 & 1.394 & 1.302 &
& 0.056 & 1.168 & 1.089 &  & 0.098 & 1.160 & 1.083\\
\multicolumn{1}{l}{0.20} & 0.026 & 1.443 & 1.257 &  & 0.033 & 1.652 & 1.377 &
& 0.034 & 1.301 & 1.087 &  & 0.066 & 1.404 & 1.121\\\hline
& \multicolumn{1}{r}{} & \multicolumn{1}{r}{} & \multicolumn{1}{r}{} &
\multicolumn{1}{r}{} & \multicolumn{1}{r}{} & \multicolumn{1}{r}{} &
\multicolumn{1}{r}{} & \multicolumn{1}{r}{} & \multicolumn{1}{r}{} &
\multicolumn{1}{r}{} & \multicolumn{1}{r}{} & \multicolumn{1}{r}{} &
\multicolumn{1}{r}{} & \multicolumn{1}{r}{} & \multicolumn{1}{r}{}\\
\multicolumn{3}{l}{\textbf{Panel B: }$T=5000$} & \multicolumn{1}{r}{} &
\multicolumn{1}{r}{} & \multicolumn{1}{r}{} & \multicolumn{1}{r}{} &
\multicolumn{1}{r}{} & \multicolumn{1}{r}{} & \multicolumn{1}{r}{} &
\multicolumn{1}{r}{} & \multicolumn{1}{r}{} & \multicolumn{1}{r}{} &
\multicolumn{1}{r}{} & \multicolumn{1}{r}{} & \multicolumn{1}{r}{}\\
\multicolumn{1}{l}{0.01} & 0.049 & 1.404 & 1.387 &  & 0.060 & 1.443 & 1.344 &
& 0.138 & 1.369 & 1.375 &  & 0.245 & 1.256 & 1.248\\
\multicolumn{1}{l}{0.025} & 0.038 & 1.306 & 1.291 &  & 0.044 & 1.348 & 1.313 &
& 0.087 & 1.245 & 1.234 &  & 0.145 & 1.197 & 1.185\\
\multicolumn{1}{l}{0.05} & 0.031 & 1.314 & 1.264 &  & 0.036 & 1.350 & 1.290 &
& 0.061 & 1.184 & 1.143 &  & 0.101 & 1.164 & 1.119\\
\multicolumn{1}{l}{0.10} & 0.024 & 1.365 & 1.265 &  & 0.029 & 1.449 & 1.320 &
& 0.041 & 1.155 & 1.067 &  & 0.071 & 1.158 & 1.069\\
\multicolumn{1}{l}{0.20} & 0.018 & 1.458 & 1.241 &  & 0.023 & 1.706 & 1.377 &
& 0.024 & 1.316 & 1.066 &  & 0.048 & 1.409 & 1.089\\\hline
\end{tabular}

\end{center}

\bigskip

\textit{Notes:} This table presents results on the accuracy of the fitted VaR
and\ ES estimates for the three estimation methods: (Q)MLE, CAViaR and\ FZ
estimation. In the first column of each panel we present the mean absolute
error (MAE) from (Q)MLE, computed across all dates in a given sample and all
1000 simulation replications. The next two columns present the
\textit{relative }MAE of CAViaR and FZ to (Q)MLE. Values greater than one
indicate (Q)MLE is more accurate (has lower MAE); values less than one
indicate the opposite.

\bigskip%

\end{landscape}%
\pagebreak

\begin{center}
\textbf{Table 5: Summary statistics}

\bigskip%

\begin{tabular}
[c]{rcccc}\hline
& \textbf{S\&P 500} & \textbf{DJIA} & \textbf{NIKKEI} & \textbf{FTSE}%
\\\cline{2-5}%
\multicolumn{1}{l}{Mean (Annualized)} & 6.776 & 7.238 & -2.682 & 3.987\\
\multicolumn{1}{l}{Std dev (Annualized)} & 17.879 & 17.042 & 24.667 & 17.730\\
\multicolumn{1}{l}{Skewness} & -0.244 & -0.163 & -0.114 & -0.126\\
\multicolumn{1}{l}{Kurtosis} & 11.673 & 11.116 & 8.580 & 8.912\\\hline
\multicolumn{1}{l}{VaR-0.01} & -3.128 & -3.034 & -4.110 & -3.098\\
\multicolumn{1}{l}{VaR-0.025} & -2.324 & -2.188 & -3.151 & -2.346\\
\multicolumn{1}{l}{VaR-0.05} & -1.731 & -1.640 & -2.451 & -1.709\\
\multicolumn{1}{l}{VaR-0.10} & -1.183 & -1.126 & -1.780 & -1.193\\\hline
\multicolumn{1}{l}{ES-0.01} & -4.528 & -4.280 & -5.783 & -4.230\\
\multicolumn{1}{l}{ES-0.025} & -3.405 & -3.215 & -4.449 & -3.295\\
\multicolumn{1}{l}{ES-0.05} & -2.697 & -2.553 & -3.603 & -2.643\\
\multicolumn{1}{l}{ES-0.10} & -2.065 & -1.955 & -2.850 & -2.031\\\hline
\end{tabular}

\end{center}

\bigskip

\textit{Notes:} This table presents summary statistics on the four daily
equity return series studied in Section \ref{sAPPLICATION}, over the full
sample period from January 1990 to December 2016. The first two rows report
the annualized mean and standard deviation of these returns in percent. The
second panel presents sample Value-at-Risk for four choices of $\alpha,$ and
the third panel presents corresponding sample Expected Shortfall estimates.

\bigskip

\bigskip

\begin{center}
\textbf{Table 6: ARMA, GARCH, and Skew t results}

\bigskip%

\begin{tabular}
[c]{rcccc}\hline
& \textbf{SP500} & \textbf{DJIA} & \textbf{NIKKEI} & \textbf{FTSE}\\
$\phi_{0}$ & 0.0269 & 0.0287 & -0.0106 & 0.0158\\
$\phi_{1}$ & 0.6482 & -0.0486 & -- & -0.0098\\
$\phi_{2}$ & -- & -0.0407 & -- & -0.0438\\
$\phi_{3}$ & -- & -- & -- & -0.0585\\
$\phi_{4}$ & -- & -- & -- & 0.0375\\
$\phi_{5}$ & -- & -- & -- & -0.0501\\
$\theta_{1}$ & -0.7048 & -- & -- & --\\
$R^{2}$ & 0.0056 & 0.0039 & 0.0000 & 0.0093\\\hline
$\omega$ & 0.0140 & 0.0165 & 0.0657 & 0.0162\\
$\beta$ & 0.9053 & 0.8970 & 0.8629 & 0.8932\\
$\alpha$ & 0.0824 & 0.0875 & 0.1125 & 0.0935\\\hline
$\nu$ & 6.9336 & 7.0616 & 7.8055 & 11.8001\\
$\lambda$ & -0.1146 & -0.0997 & -0.0659 & -0.1018\\\hline
\end{tabular}

\end{center}

\bigskip

\textit{Notes:} This table presents parameter estimates for the four daily
equity return series studied in Section \ref{sAPPLICATION}, over the in-sample
period from January 1990 to December 1999. The first panel presents the
optimal\ ARMA model according to the BIC, along with the $R^{2}$ of that
model. The second panel presents the estimated GARCH(1,1) parameters, and the
third panel presents the estimated parameters of the skewed $t$ distribution
applied to the estimated standardized residuals.

\bigskip

\pagebreak

\bigskip

\begin{center}
\textbf{Table 7: Estimated paramters of GAS models for VaR and\ ES}

\bigskip%

\begin{tabular}
[c]{lllllllll}\hline
&  &  &  &  &  &  &  & \\
& \multicolumn{2}{c}{\textit{GAS-2F}} & \multicolumn{1}{c}{} &  &  &
\multicolumn{1}{c}{\textit{GAS-1F}} & \multicolumn{1}{c}{\textit{GARCH-FZ}} &
\multicolumn{1}{c}{\textit{Hybrid}}\\\cline{2-3}\cline{7-9}
& \multicolumn{1}{c}{\textbf{VaR}} & \multicolumn{1}{c}{\textbf{ES}} &  &  &
&  &  & \\\cline{2-3}
& \multicolumn{1}{c}{} & \multicolumn{1}{c}{} & \multicolumn{1}{c}{} &  &
\multicolumn{1}{c}{\ \ \ \ \ \ \ \ \ \ \ \ \ \ } & \multicolumn{1}{c}{} &
\multicolumn{1}{c}{} & \multicolumn{1}{c}{}\\
\multicolumn{1}{c}{$w$} & \multicolumn{1}{c}{-0.046} &
\multicolumn{1}{c}{-0.069} & \multicolumn{1}{c}{} &  &
\multicolumn{1}{c}{$\beta$} & \multicolumn{1}{c}{0.990} &
\multicolumn{1}{c}{0.908} & \multicolumn{1}{c}{0.968}\\
\multicolumn{1}{c}{{\small (s.e.)}} & \multicolumn{1}{c}{{\small (0.010)}} &
\multicolumn{1}{c}{{\small (0.019)}} & \multicolumn{1}{c}{} &  &
\multicolumn{1}{c}{{\small (s.e.)}} & \multicolumn{1}{c}{{\small (0.004)}} &
\multicolumn{1}{c}{{\small (0.072)}} & \multicolumn{1}{c}{{\small (0.015)}}\\
\multicolumn{1}{c}{$b$} & \multicolumn{1}{c}{0.977} &
\multicolumn{1}{c}{0.973} & \multicolumn{1}{c}{} &  &
\multicolumn{1}{c}{$\gamma$} & \multicolumn{1}{c}{-0.010} &
\multicolumn{1}{c}{0.030} & \multicolumn{1}{c}{-0.011}\\
\multicolumn{1}{c}{{\small (s.e.)}} & \multicolumn{1}{c}{{\small (0.005)}} &
\multicolumn{1}{c}{{\small (0.007)}} &  &  &
\multicolumn{1}{c}{{\small (s.e.)}} & \multicolumn{1}{c}{{\small (0.002)}} &
\multicolumn{1}{c}{{\small (0.010)}} & \multicolumn{1}{c}{{\small (0.002)}}\\
\multicolumn{1}{c}{$a_{v}$} & \multicolumn{1}{c}{0.001} &
\multicolumn{1}{c}{0.001} &  &  & \multicolumn{1}{c}{$\delta$} &
\multicolumn{1}{c}{--} & \multicolumn{1}{c}{--} & \multicolumn{1}{c}{0.018}\\
\multicolumn{1}{c}{{\small (s.e.)}} & \multicolumn{1}{c}{{\small (0.092)}} &
\multicolumn{1}{c}{{\small (0.164)}} &  &  &
\multicolumn{1}{c}{{\small (s.e.)}} & \multicolumn{1}{c}{} &
\multicolumn{1}{c}{} & \multicolumn{1}{c}{{\small (0.009)}}\\
\multicolumn{1}{c}{$a_{e}$} & \multicolumn{1}{c}{0.007} &
\multicolumn{1}{c}{0.011} &  &  & \multicolumn{1}{c}{$a$} &
\multicolumn{1}{c}{-1.490} & \multicolumn{1}{c}{-2.659} &
\multicolumn{1}{c}{-2.443}\\
\multicolumn{1}{c}{{\small (s.e.)}} & \multicolumn{1}{c}{{\small (0.004)}} &
\multicolumn{1}{c}{{\small (0.007)}} &  &  &
\multicolumn{1}{c}{{\small (s.e.)}} & \multicolumn{1}{c}{{\small (0.346)}} &
\multicolumn{1}{c}{{\small (0.492)}} & \multicolumn{1}{c}{{\small (0.473)}}\\
& \multicolumn{1}{c}{} & \multicolumn{1}{c}{} &  &  & \multicolumn{1}{c}{$b$}
& \multicolumn{1}{c}{-2.089} & \multicolumn{1}{c}{-3.761} &
\multicolumn{1}{c}{-3.389}\\
& \multicolumn{1}{c}{} & \multicolumn{1}{c}{} &  &  &
\multicolumn{1}{c}{{\small (s.e.)}} & \multicolumn{1}{c}{{\small (0.487)}} &
\multicolumn{1}{c}{{\small (0.747)}} & \multicolumn{1}{c}{{\small (0.664)}}\\
&  &  &  &  &  &  &  & \\\hline
Avg loss & \multicolumn{2}{c}{0.747} & \multicolumn{1}{c}{} &  &  &
\multicolumn{1}{c}{0.750} & \multicolumn{1}{c}{0.762} &
\multicolumn{1}{c}{0.745}\\\hline
\end{tabular}

\end{center}

\bigskip

\textit{Notes:} This table presents parameter estimates and standard errors
for four GAS\ models of VaR and ES for the S\&P 500 index over the in-sample
period from January 1990 to December 1999. The left panel presents the results
for the two-factor GAS model in Section \ref{sGAS}. The right panel presents
the results for the three one-factor models: a one-factor GAS model
(from\ Section \ref{sFZ1F}), and a GARCH model estimated by\ FZ loss
minimization, and \textquotedblleft hybrid\textquotedblright\ one-factor GAS
model that includes a additional GARCH-type forcing variable (both from
Section \ref{sGASGARCH}). The bottom row of this table presents the average
(in-sample) losses from each of these four models.

\bigskip

\pagebreak\pagebreak%

\begin{landscape}%

\begin{center}
\bigskip

\textbf{Table 8: Out-of-sample average losses and goodness-of-fit tests
(alpha=0.05)}

\bigskip%

\begin{tabular}
[c]{rccccllllllllll}\hline
&  &  &  &  &  &  &  &  &  &  &  &  &  & \\
& \multicolumn{4}{c}{\textit{Average loss}} &  &
\multicolumn{4}{c}{\textit{GoF p-values: VaR}} &  &
\multicolumn{4}{c}{\textit{GoF p-values:\ ES}}\\\cline{2-5}\cline{7-10}%
\cline{12-15}
& \textbf{S\&P} & \textbf{DJIA} & \textbf{NIK} & \textbf{FTSE} &  &
\multicolumn{1}{c}{\textbf{S\&P}} & \multicolumn{1}{c}{\textbf{DJIA}} &
\multicolumn{1}{c}{\textbf{NIK}} & \multicolumn{1}{c}{\textbf{FTSE}} &  &
\multicolumn{1}{c}{\textbf{S\&P}} & \multicolumn{1}{c}{\textbf{DJIA}} &
\multicolumn{1}{c}{\textbf{NIK}} & \multicolumn{1}{c}{\textbf{FTSE}%
}\\\cline{2-15}
&  &  &  &  &  &  &  &  &  &  &  &  &  & \\
\multicolumn{1}{l}{RW-125} & 0.914 & 0.864 & 1.290 & 0.959 &  &
\multicolumn{1}{c}{0.021} & \multicolumn{1}{c}{0.013} &
\multicolumn{1}{c}{0.000} & \multicolumn{1}{c}{0.000} &  &
\multicolumn{1}{c}{0.029} & \multicolumn{1}{c}{0.018} &
\multicolumn{1}{c}{0.006} & \multicolumn{1}{c}{0.000}\\
\multicolumn{1}{l}{RW-250} & 0.959 & 0.909 & 1.294 & 1.002 &  &
\multicolumn{1}{c}{0.001} & \multicolumn{1}{c}{0.001} &
\multicolumn{1}{c}{0.007} & \multicolumn{1}{c}{0.000} &  &
\multicolumn{1}{c}{0.043} & \multicolumn{1}{c}{0.014} &
\multicolumn{1}{c}{0.018} & \multicolumn{1}{c}{0.002}\\
\multicolumn{1}{l}{RW-500} & 1.023 & 0.976 & 1.318 & 1.056 &  &
\multicolumn{1}{c}{0.001} & \multicolumn{1}{c}{0.001} &
\multicolumn{1}{c}{0.000} & \multicolumn{1}{c}{0.000} &  &
\multicolumn{1}{c}{0.012} & \multicolumn{1}{c}{0.011} &
\multicolumn{1}{c}{0.001} & \multicolumn{1}{c}{0.000}\\
\multicolumn{1}{l}{GCH-N} & 0.876 & 0.808 & 1.170 & 0.871 &  &
\multicolumn{1}{c}{0.031} & \multicolumn{1}{c}{\textbf{0.139}} &
\multicolumn{1}{c}{\textbf{0.532}} & \multicolumn{1}{c}{0.000} &  &
\multicolumn{1}{c}{0.001} & \multicolumn{1}{c}{0.006} &
\multicolumn{1}{c}{\textbf{0.187}} & \multicolumn{1}{c}{0.000}\\
\multicolumn{1}{l}{GCH-Skt} & 0.866 & 0.796 & 1.168 & \textit{0.863} &  &
\multicolumn{1}{c}{0.003} & \multicolumn{1}{c}{\textit{0.085}} &
\multicolumn{1}{c}{\textbf{0.114}} & \multicolumn{1}{c}{0.000} &  &
\multicolumn{1}{c}{0.003} & \multicolumn{1}{c}{\textit{0.085}} &
\multicolumn{1}{c}{\textbf{0.282}} & \multicolumn{1}{c}{0.000}\\
\multicolumn{1}{l}{GCH-EDF} & 0.862 & \textit{0.796} & \textit{1.166} &
0.867 &  & \multicolumn{1}{c}{0.003} & \multicolumn{1}{c}{0.029} &
\multicolumn{1}{c}{\textbf{0.583}} & \multicolumn{1}{c}{0.000} &  &
\multicolumn{1}{c}{0.014} & \multicolumn{1}{c}{\textit{0.098}} &
\multicolumn{1}{c}{\textbf{0.527}} & \multicolumn{1}{c}{0.000}\\
\multicolumn{1}{l}{FZ-2F} & \textit{0.856} & 0.798 & 1.206 & 1.098 &  &
\multicolumn{1}{c}{0.000} & \multicolumn{1}{c}{0.000} &
\multicolumn{1}{c}{\textbf{0.258}} & \multicolumn{1}{c}{0.000} &  &
\multicolumn{1}{c}{\textit{0.061}} & \multicolumn{1}{c}{\textbf{0.195}} &
\multicolumn{1}{c}{\textbf{0.247}} & \multicolumn{1}{c}{0.000}\\
\multicolumn{1}{l}{FZ-1F} & \textbf{0.853} & \textbf{0.784} & 1.191 & 0.867 &
& \multicolumn{1}{c}{\textbf{0.242}} & \multicolumn{1}{c}{\textbf{0.248}} &
\multicolumn{1}{c}{\textbf{0.317}} & \multicolumn{1}{c}{0.019} &  &
\multicolumn{1}{c}{\textbf{0.313}} & \multicolumn{1}{c}{\textbf{0.130}} &
\multicolumn{1}{c}{\textbf{0.612}} & \multicolumn{1}{c}{0.003}\\
\multicolumn{1}{l}{GCH-FZ} & 0.862 & 0.797 & 1.167 & 0.866 &  &
\multicolumn{1}{c}{0.005} & \multicolumn{1}{c}{0.001} &
\multicolumn{1}{c}{\textbf{0.331}} & \multicolumn{1}{c}{0.000} &  &
\multicolumn{1}{c}{0.018} & \multicolumn{1}{c}{0.011} &
\multicolumn{1}{c}{\textbf{0.389}} & \multicolumn{1}{c}{0.000}\\
\multicolumn{1}{l}{Hybrid} & 0.869 & 0.797 & \textbf{1.165} & \textbf{0.862} &
& \multicolumn{1}{c}{0.001} & \multicolumn{1}{c}{\textit{0.069}} &
\multicolumn{1}{c}{\textbf{0.326}} & \multicolumn{1}{c}{0.000} &  &
\multicolumn{1}{c}{0.010} & \multicolumn{1}{c}{\textbf{0.159}} &
\multicolumn{1}{c}{\textbf{0.518}} & \multicolumn{1}{c}{0.000}\\\hline
\end{tabular}

\end{center}

\bigskip

\textit{Notes:} The left panel of this table presents the average losses,
using the FZ0 loss function, for four daily equity return series, over the
out-of-sample period from January 2000 to December 2016, for ten different
forecasting models. The lowest average loss in each column is highlighted in
bold, the second-lowest is highlighted in italics. The first three rows
correspond to rolling window forecasts, the next three rows correspond
to\ GARCH forecasts based on different models for the standardized residuals,
and the last four rows correspond to models introduced in Section
\ref{sMODELS}. The middle and right panels of this table present $p$-values
from goodness-of-fit tests of the VaR and\ ES forecasts respectively. Values
that are greater than 0.10 (indicating no evidence against optimality at the
0.10 level) are in bold, and values between 0.05 and 0.10 are in italics.

\bigskip%

\end{landscape}%
\pagebreak

\begin{center}
\bigskip

\textbf{Table 9: Diebold-Mariano t-statistics on average out-of-sample loss
differences}

\textbf{alpha=0.05, S\&P 500 returns}%

\begin{tabular}
[c]{rcccccccccc}\hline
& RW125 & RW250 & RW500 & G-N & G-Skt & G-EDF & FZ-2F & FZ-1F & G-FZ &
Hybrid\\
RW125 &  & -2.580 & -4.260 & 2.109 & 2.693 & 2.900 & 2.978 & 3.978 & 3.020 &
2.967\\
RW250 & 2.580 &  & -4.015 & 3.098 & 3.549 & 3.730 & 3.799 & 4.701 & 3.921 &
4.110\\
RW500 & 4.260 & 4.015 &  & 4.401 & 4.783 & 4.937 & 5.168 & 5.893 & 5.125 &
5.450\\\hline
G-N & -2.109 & -3.098 & -4.401 &  & 3.670 & 3.068 & 1.553 & 2.248 & 2.818 &
0.685\\
G-Skt & -2.693 & -3.549 & -4.783 & -3.670 &  & 2.103 & 0.889 & 1.475 & 1.232 &
-0.403\\
G-EDF & -2.900 & -3.730 & -4.937 & -3.068 & -2.103 &  & 0.599 & 1.157 &
0.024 & -0.769\\\hline
FZ-2F & -2.978 & -3.799 & -5.168 & -1.553 & -0.889 & -0.599 &  & 0.582 &
-0.555 & -0.580\\
FZ-1F & -3.912 & -4.423 & -5.483 & -1.986 & -1.421 & -1.198 & -0.582 &  &
-1.266 & -1.978\\
G-FZ & -3.020 & -3.921 & -5.125 & -2.818 & -1.324 & -0.024 & 0.555 & 1.266 &
& -0.914\\
Hybrid & -3.276 & -4.137 & -5.272 & -1.492 & -0.419 & 0.045 & 0.580 & 1.978 &
0.914 & \\\hline
\end{tabular}

\end{center}

\bigskip

\textit{Notes:} This table presents $t$-statistics from Diebold-Mariano tests
comparing the average losses, using the FZ0 loss function, over the
out-of-sample period from January 2000 to December 2016, for ten different
forecasting models. A positive value indicates that the row model has higher
average loss than the column model. Values greater than 1.96 in absolute value
indicate that the average loss difference is significantly different from zero
at the 95\% confidence level. Values along the main diagonal are all
identically zero and are omitted for interpretability. The first three rows
correspond to rolling window forecasts, the next three rows correspond
to\ GARCH forecasts based on different models for the standardized residuals,
and the last four rows correspond to models introduced in Section
\ref{sMODELS}.

\bigskip

\pagebreak\bigskip

\begin{center}
\textbf{Table 10: Out-of-sample performance rankings for various alpha}%

\begin{tabular}
[c]{rrrrrrrrrrrr}\hline
&  &  &  &  &  &  &  &  &  &  & \\
& \multicolumn{5}{c}{$\alpha=0.01$} &  & \multicolumn{5}{c}{$\alpha=0.025$%
}\\\cline{2-6}\cline{8-12}%
\multicolumn{1}{l}{} & \multicolumn{1}{c}{\textbf{S\&P}} &
\multicolumn{1}{c}{\textbf{DJIA}} & \multicolumn{1}{c}{\textbf{NIK}} &
\multicolumn{1}{c}{\textbf{FTSE}} & \multicolumn{1}{c}{\textbf{Avg}} &
\multicolumn{1}{c}{} & \multicolumn{1}{c}{\textbf{S\&P}} &
\multicolumn{1}{c}{\textbf{DJIA}} & \multicolumn{1}{c}{\textbf{NIK}} &
\multicolumn{1}{c}{\textbf{FTSE}} & \multicolumn{1}{c}{\textbf{Avg}}\\
\multicolumn{1}{l}{RW-125} & \multicolumn{1}{c}{7} & \multicolumn{1}{c}{8} &
\multicolumn{1}{c}{10} & \multicolumn{1}{c}{7} & \multicolumn{1}{c}{8} &
\multicolumn{1}{c}{} & \multicolumn{1}{c}{8} & \multicolumn{1}{c}{8} &
\multicolumn{1}{c}{8} & \multicolumn{1}{c}{7} & \multicolumn{1}{c}{7.75}\\
\multicolumn{1}{l}{RW-250} & \multicolumn{1}{c}{8} & \multicolumn{1}{c}{9} &
\multicolumn{1}{c}{8} & \multicolumn{1}{c}{8} & \multicolumn{1}{c}{8.25} &
\multicolumn{1}{c}{} & \multicolumn{1}{c}{9} & \multicolumn{1}{c}{9} &
\multicolumn{1}{c}{7} & \multicolumn{1}{c}{8} & \multicolumn{1}{c}{8.25}\\
\multicolumn{1}{l}{RW-500} & \multicolumn{1}{c}{10} & \multicolumn{1}{c}{10} &
\multicolumn{1}{c}{9} & \multicolumn{1}{c}{9} & \multicolumn{1}{c}{9.5} &
\multicolumn{1}{c}{} & \multicolumn{1}{c}{10} & \multicolumn{1}{c}{10} &
\multicolumn{1}{c}{9} & \multicolumn{1}{c}{9} & \multicolumn{1}{c}{9.5}%
\\\hline
\multicolumn{1}{l}{G-N} & \multicolumn{1}{c}{6} & \multicolumn{1}{c}{6} &
\multicolumn{1}{c}{5} & \multicolumn{1}{c}{4} & \multicolumn{1}{c}{5.25} &
\multicolumn{1}{c}{} & \multicolumn{1}{c}{7} & \multicolumn{1}{c}{6} &
\multicolumn{1}{c}{4} & \multicolumn{1}{c}{3} & \multicolumn{1}{c}{5}\\
\multicolumn{1}{l}{G-Skt} & \multicolumn{1}{c}{5} & \multicolumn{1}{c}{3} &
\multicolumn{1}{c}{2} & \multicolumn{1}{c}{2} & \multicolumn{1}{c}{3} &
\multicolumn{1}{c}{} & \multicolumn{1}{c}{5} & \multicolumn{1}{c}{3} &
\multicolumn{1}{c}{1} & \multicolumn{1}{c}{1} & \multicolumn{1}{c}{2.5}\\
\multicolumn{1}{l}{G-EDF} & \multicolumn{1}{c}{4} & \multicolumn{1}{c}{2} &
\multicolumn{1}{c}{3} & \multicolumn{1}{c}{1} & \multicolumn{1}{c}{2.5} &
\multicolumn{1}{c}{} & \multicolumn{1}{c}{2} & \multicolumn{1}{c}{2} &
\multicolumn{1}{c}{3} & \multicolumn{1}{c}{2} & \multicolumn{1}{c}{2.25}%
\\\hline
\multicolumn{1}{l}{FZ-2F} & \multicolumn{1}{c}{1} & \multicolumn{1}{c}{4} &
\multicolumn{1}{c}{7} & \multicolumn{1}{c}{10} & \multicolumn{1}{c}{5.5} &
\multicolumn{1}{c}{} & \multicolumn{1}{c}{4} & \multicolumn{1}{c}{5} &
\multicolumn{1}{c}{10} & \multicolumn{1}{c}{10} & \multicolumn{1}{c}{7.25}\\
\multicolumn{1}{l}{FZ-1F} & \multicolumn{1}{c}{9} & \multicolumn{1}{c}{7} &
\multicolumn{1}{c}{6} & \multicolumn{1}{c}{6} & \multicolumn{1}{c}{7} &
\multicolumn{1}{c}{} & \multicolumn{1}{c}{3} & \multicolumn{1}{c}{4} &
\multicolumn{1}{c}{6} & \multicolumn{1}{c}{4} & \multicolumn{1}{c}{4.25}\\
\multicolumn{1}{l}{G-FZ} & \multicolumn{1}{c}{3} & \multicolumn{1}{c}{1} &
\multicolumn{1}{c}{1} & \multicolumn{1}{c}{3} & \multicolumn{1}{c}{2} &
\multicolumn{1}{c}{} & \multicolumn{1}{c}{1} & \multicolumn{1}{c}{1} &
\multicolumn{1}{c}{2} & \multicolumn{1}{c}{5} & \multicolumn{1}{c}{2.25}\\
\multicolumn{1}{l}{Hybrid} & \multicolumn{1}{c}{2} & \multicolumn{1}{c}{5} &
\multicolumn{1}{c}{4} & \multicolumn{1}{c}{5} & \multicolumn{1}{c}{4} &
\multicolumn{1}{c}{} & \multicolumn{1}{c}{6} & \multicolumn{1}{c}{7} &
\multicolumn{1}{c}{5} & \multicolumn{1}{c}{6} & \multicolumn{1}{c}{6}\\\hline
\multicolumn{1}{l}{} & \multicolumn{1}{c}{} & \multicolumn{1}{c}{} &
\multicolumn{1}{c}{} & \multicolumn{1}{c}{} & \multicolumn{1}{c}{} &
\multicolumn{1}{c}{} & \multicolumn{1}{c}{} & \multicolumn{1}{c}{} &
\multicolumn{1}{c}{} & \multicolumn{1}{c}{} & \multicolumn{1}{c}{}\\
\multicolumn{1}{l}{} & \multicolumn{5}{c}{$\alpha=0.05$} &
\multicolumn{1}{c}{} & \multicolumn{5}{c}{$\alpha=0.10$}\\\cline{2-6}%
\cline{8-12}%
\multicolumn{1}{l}{} & \multicolumn{1}{c}{\textbf{S\&P}} &
\multicolumn{1}{c}{\textbf{DJIA}} & \multicolumn{1}{c}{\textbf{NIK}} &
\multicolumn{1}{c}{\textbf{FTSE}} & \multicolumn{1}{c}{\textbf{Avg}} &
\multicolumn{1}{c}{} & \multicolumn{1}{c}{\textbf{S\&P}} &
\multicolumn{1}{c}{\textbf{DJIA}} & \multicolumn{1}{c}{\textbf{NIK}} &
\multicolumn{1}{c}{\textbf{FTSE}} & \multicolumn{1}{c}{\textbf{Avg}}\\
\multicolumn{1}{l}{RW-125} & \multicolumn{1}{c}{8} & \multicolumn{1}{c}{8} &
\multicolumn{1}{c}{8} & \multicolumn{1}{c}{7} & \multicolumn{1}{c}{7.75} &
\multicolumn{1}{c}{} & \multicolumn{1}{c}{8} & \multicolumn{1}{c}{8} &
\multicolumn{1}{c}{8} & \multicolumn{1}{c}{8} & \multicolumn{1}{c}{8}\\
\multicolumn{1}{l}{RW-250} & \multicolumn{1}{c}{9} & \multicolumn{1}{c}{9} &
\multicolumn{1}{c}{9} & \multicolumn{1}{c}{8} & \multicolumn{1}{c}{8.75} &
\multicolumn{1}{c}{} & \multicolumn{1}{c}{9} & \multicolumn{1}{c}{9} &
\multicolumn{1}{c}{9} & \multicolumn{1}{c}{9} & \multicolumn{1}{c}{9}\\
\multicolumn{1}{l}{RW-500} & \multicolumn{1}{c}{10} & \multicolumn{1}{c}{10} &
\multicolumn{1}{c}{10} & \multicolumn{1}{c}{9} & \multicolumn{1}{c}{9.75} &
\multicolumn{1}{c}{} & \multicolumn{1}{c}{10} & \multicolumn{1}{c}{10} &
\multicolumn{1}{c}{10} & \multicolumn{1}{c}{10} & \multicolumn{1}{c}{10}%
\\\hline
\multicolumn{1}{l}{G-N} & \multicolumn{1}{c}{7} & \multicolumn{1}{c}{7} &
\multicolumn{1}{c}{5} & \multicolumn{1}{c}{6} & \multicolumn{1}{c}{6.25} &
\multicolumn{1}{c}{} & \multicolumn{1}{c}{3} & \multicolumn{1}{c}{2} &
\multicolumn{1}{c}{5} & \multicolumn{1}{c}{5} & \multicolumn{1}{c}{3.75}\\
\multicolumn{1}{l}{G-Skt} & \multicolumn{1}{c}{5} & \multicolumn{1}{c}{3} &
\multicolumn{1}{c}{4} & \multicolumn{1}{c}{2} & \multicolumn{1}{c}{3.5} &
\multicolumn{1}{c}{} & \multicolumn{1}{c}{7} & \multicolumn{1}{c}{4} &
\multicolumn{1}{c}{4} & \multicolumn{1}{c}{4} & \multicolumn{1}{c}{4.75}\\
\multicolumn{1}{l}{G-EDF} & \multicolumn{1}{c}{4} & \multicolumn{1}{c}{2} &
\multicolumn{1}{c}{2} & \multicolumn{1}{c}{5} & \multicolumn{1}{c}{3.25} &
\multicolumn{1}{c}{} & \multicolumn{1}{c}{4} & \multicolumn{1}{c}{3} &
\multicolumn{1}{c}{3} & \multicolumn{1}{c}{3} & \multicolumn{1}{c}{3.25}%
\\\hline
\multicolumn{1}{l}{FZ-2F} & \multicolumn{1}{c}{2} & \multicolumn{1}{c}{6} &
\multicolumn{1}{c}{7} & \multicolumn{1}{c}{10} & \multicolumn{1}{c}{6.25} &
\multicolumn{1}{c}{} & \multicolumn{1}{c}{2} & \multicolumn{1}{c}{6} &
\multicolumn{1}{c}{7} & \multicolumn{1}{c}{7} & \multicolumn{1}{c}{5.5}\\
\multicolumn{1}{l}{FZ-1F} & \multicolumn{1}{c}{1} & \multicolumn{1}{c}{1} &
\multicolumn{1}{c}{6} & \multicolumn{1}{c}{4} & \multicolumn{1}{c}{3} &
\multicolumn{1}{c}{} & \multicolumn{1}{c}{1} & \multicolumn{1}{c}{7} &
\multicolumn{1}{c}{2} & \multicolumn{1}{c}{2} & \multicolumn{1}{c}{3}\\
\multicolumn{1}{l}{G-FZ} & \multicolumn{1}{c}{3} & \multicolumn{1}{c}{5} &
\multicolumn{1}{c}{3} & \multicolumn{1}{c}{3} & \multicolumn{1}{c}{3.5} &
\multicolumn{1}{c}{} & \multicolumn{1}{c}{6} & \multicolumn{1}{c}{5} &
\multicolumn{1}{c}{6} & \multicolumn{1}{c}{6} & \multicolumn{1}{c}{5.75}\\
\multicolumn{1}{l}{Hybrid} & \multicolumn{1}{c}{6} & \multicolumn{1}{c}{4} &
\multicolumn{1}{c}{1} & \multicolumn{1}{c}{1} & \multicolumn{1}{c}{3} &
\multicolumn{1}{c}{} & \multicolumn{1}{c}{5} & \multicolumn{1}{c}{1} &
\multicolumn{1}{c}{1} & \multicolumn{1}{c}{1} & \multicolumn{1}{c}{2}\\\hline
\end{tabular}

\end{center}

\bigskip

\textit{Notes:} This table presents the rankings (with the best performing
model ranked 1 and the worst ranked 10) based on average losses using the FZ0
loss function, for four daily equity return series, over the out-of-sample
period from January 2000 to December 2016, for ten different forecasting
models. The first three rows in each panel correspond to rolling window
forecasts, the next three rows correspond to\ GARCH forecasts based on
different models for the standardized residuals, and the last four rows
correspond to models introduced in Section \ref{sMODELS}. The last column in
each panel represents the average rank across the four equity return series.

\pagebreak%

\begin{figure}
[ph]
\begin{center}
\includegraphics[
height=2.8905in,
width=6.5778in
]%
{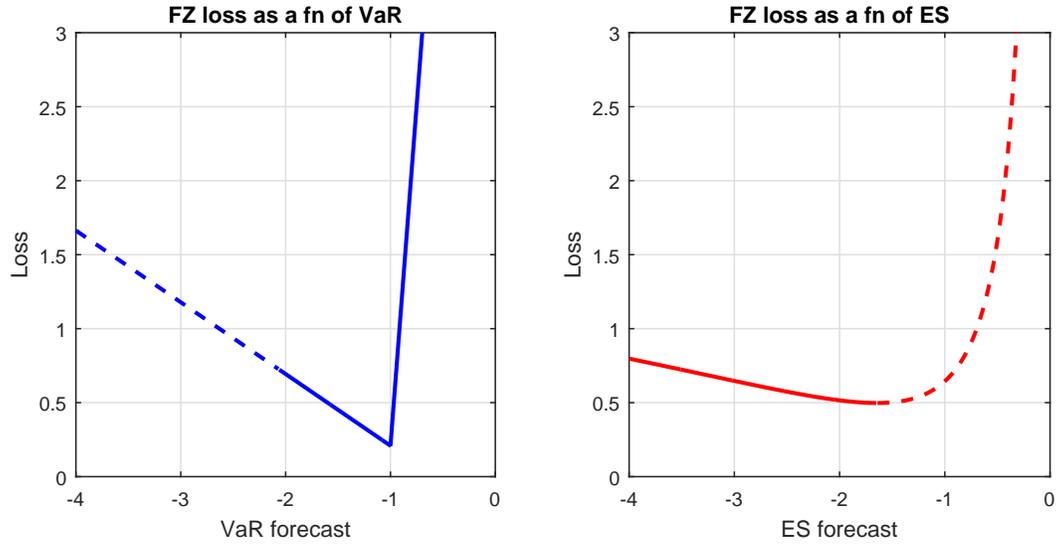}%
\caption{\label{figFZ0loss}\textit{This figure plots the FZ0 loss function
when }$Y=-1$\textit{\ and }$\alpha=0.05.$\textit{\ In the left panel we fix
}$e=-2.06$\textit{\ and vary }$v,$\textit{\ in the right panel we fix
}$v=-1.64$\textit{\ and vary }$e.$\textit{\ Values where }$v<e$\textit{\ are
indicated with a dashed line.}}%
\end{center}
\end{figure}

%

\begin{figure}
[ptb]
\begin{center}
\includegraphics[
height=3.8269in,
width=4.2473in
]%
{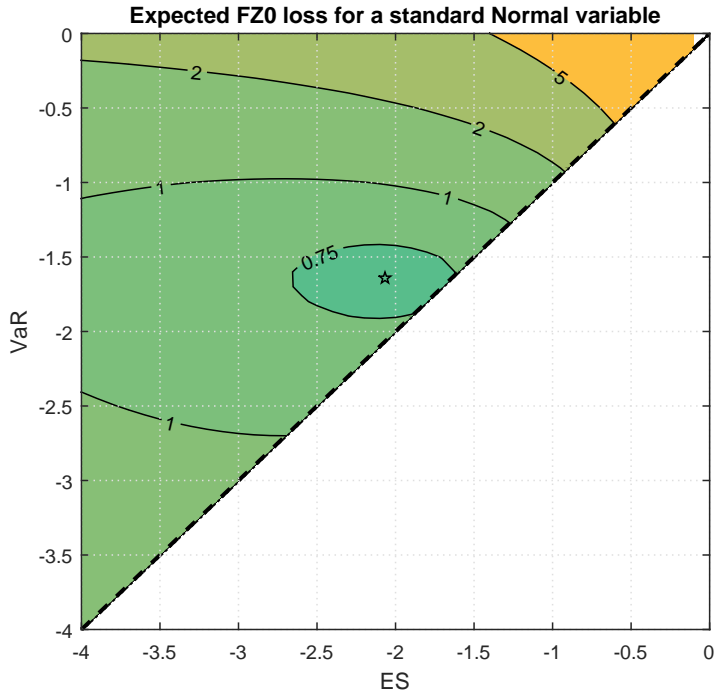}%
\caption{\label{figFZ0contour}\textit{Contours of expected FZ0 loss when the
target variable is standard Normal. Only values where ES$<$VaR$<$0 are
considered. The optimal value is marked with a star.}}%
\end{center}
\end{figure}
%

\begin{figure}
[ph]
\begin{center}
\includegraphics[
height=2.881in,
width=5.6597in
]%
{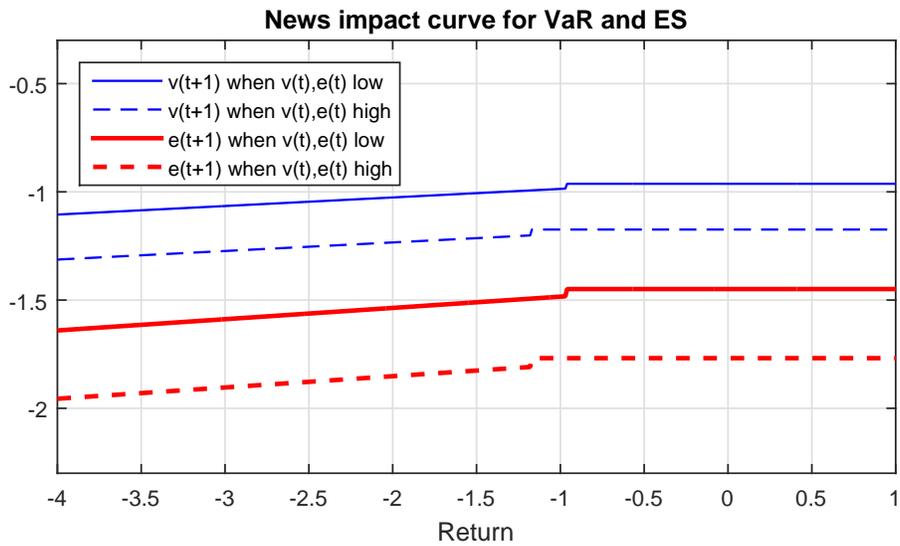}%
\caption{\label{figNIC}\textit{This figure shows the values of VaR and\ ES as
a function of the lagged return, when the lagged values of VaR and\ ES are
either low (10\% below average) or high (10\% above average). The function is
based on the estimated parameters for daily S\&P 500 returns.}}%
\end{center}
\end{figure}

\bigskip%

\begin{figure}
[ptb]
\begin{center}
\includegraphics[
height=7.072in,
width=5.0181in
]%
{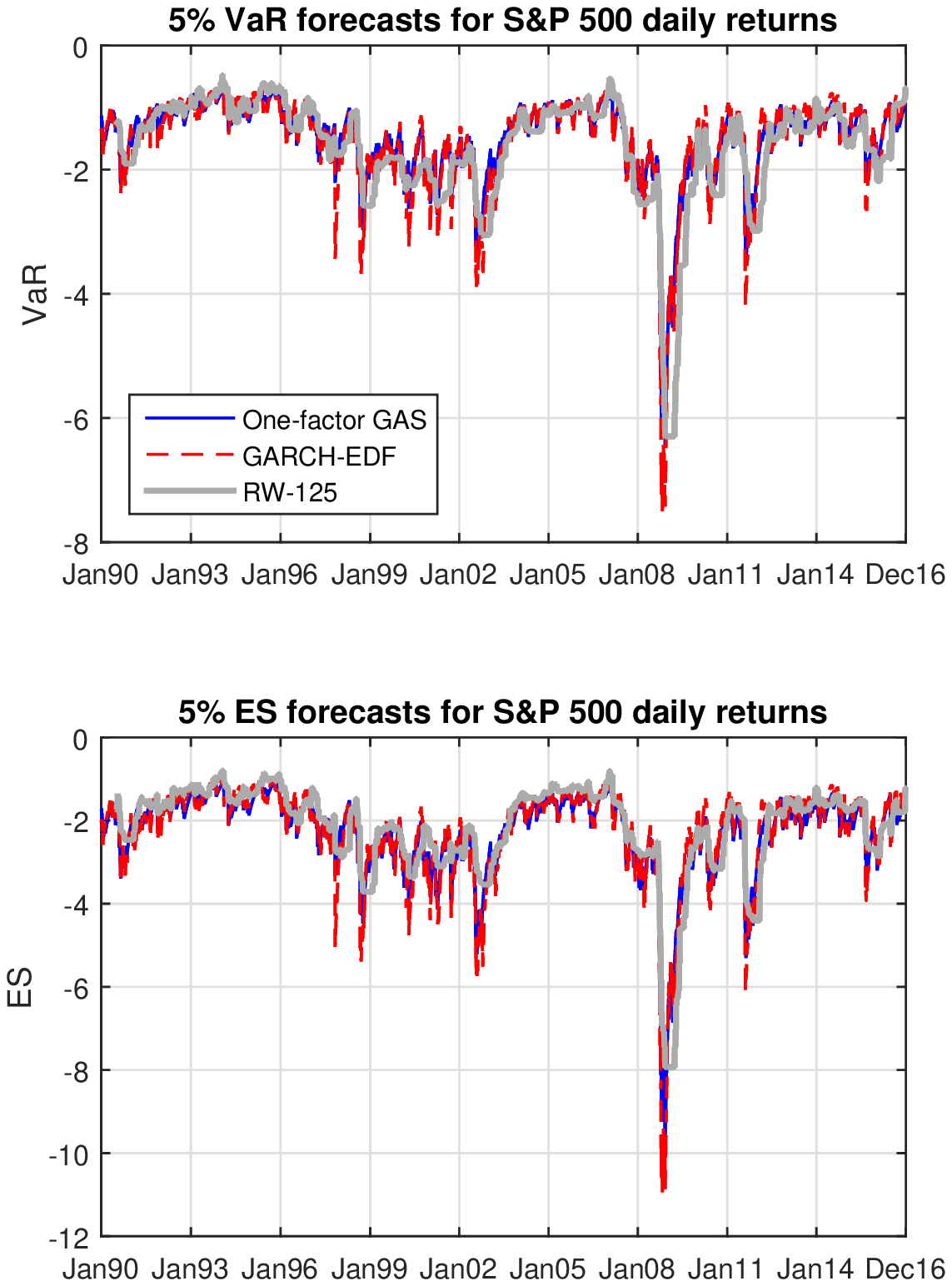}%
\caption{\label{figSPve1}\textit{This figure plots the estimated 5\%
Value-at-Risk (VaR) and Expected Shortfall (ES) for daily returns on the S\&P
500 index, over the period January 1990 to December 2016.\ The estimates are
based on a one-factor GAS model, a GARCH model, and a rolling window using 125
observations.}}%
\end{center}
\end{figure}
%

\begin{figure}
[ptb]
\begin{center}
\includegraphics[
height=7.0668in,
width=5.019in
]%
{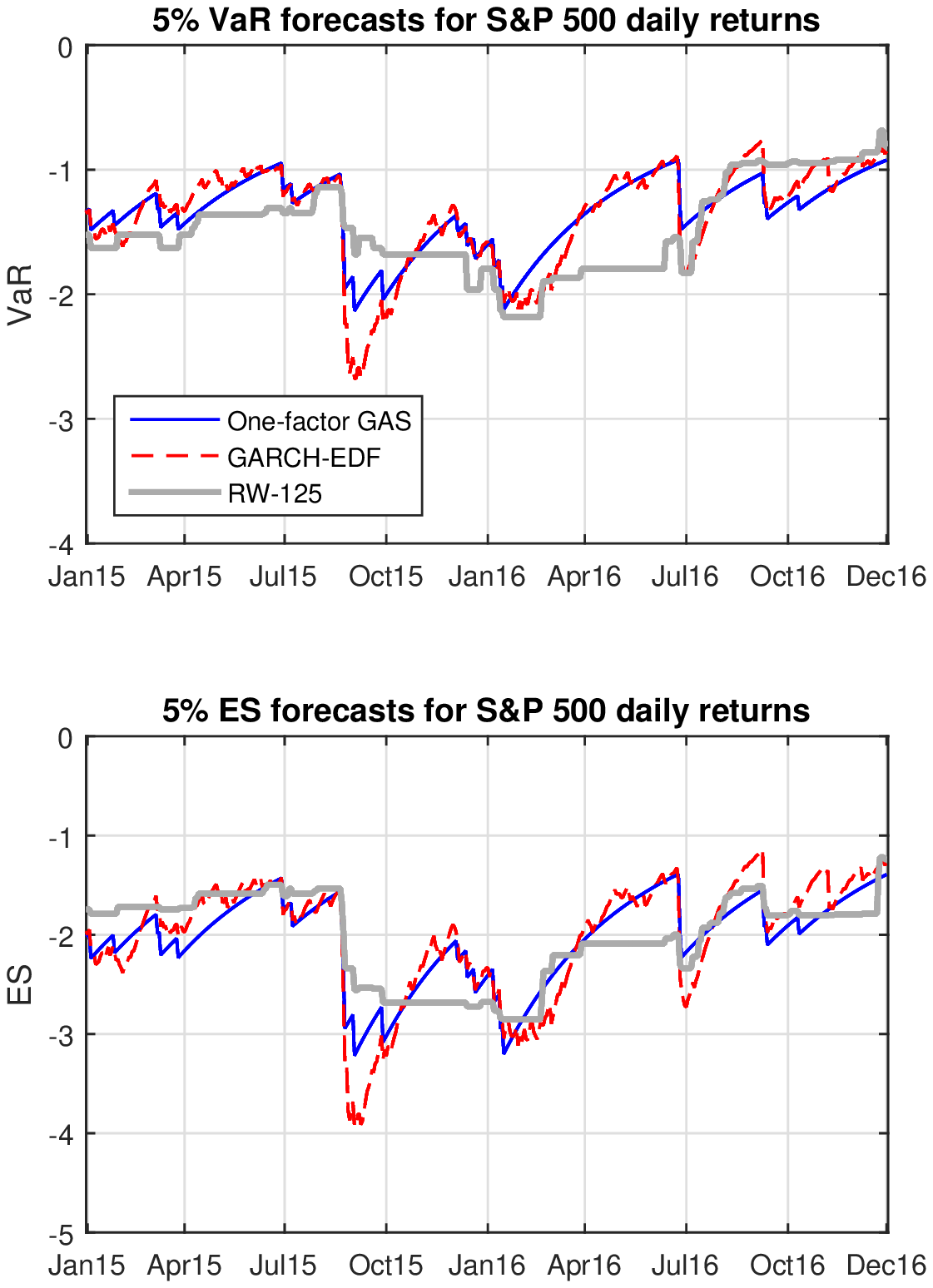}%
\caption{\label{figSPve2}\textit{This figure plots the estimated 5\%
Value-at-Risk (VaR) and Expected Shortfall (ES) for daily returns on the S\&P
500 index, over the period January 2015 to December 2016.\ The estimates are
based on a one-factor GAS model, a GARCH model, and a rolling window using 125
observations.}}%
\end{center}
\end{figure}

\bigskip

\pagebreak%

\def\baselinestretch{1.5}\small\normalsize
\setcounter{page}{1}\setcounter{equation}{0}

\begin{center}
{\Large Supplemental Appendix to:}

\bigskip

{\LARGE Dynamic Semiparametric Models for }\linebreak{\LARGE Expected
Shortfall (and Value-at-Risk)}

\bigskip\bigskip%

\begin{tabular}
[c]{ccc}%
{\Large Andrew J.\ Patton } & {\Large Johanna F. Ziegel} & {\Large Rui Chen
}\\
Duke University & University of Bern & \ \ Duke University \
\end{tabular}
{\Large \ }

\bigskip

11 July 2017
\end{center}

\bigskip

This appendix contains lemmas that provide further details on the proof of
Theorem 2 presented in the main paper, as well as additional tables of analysis.

\bigskip

{\LARGE Appendix SA.1:\ Detailed proofs}

Throughout this appendix, we suppress the subscript on $\mathbf{\hat{\theta}%
}_{T}$ for simplicity of presentation, and we denote the conditional
distribution and density functions as $F_{t}$ and $f_{t}$ rather than
$F_{t}\left(  \cdot|\mathcal{F}_{t-1}\right)  $ and $f_{t}\left(
\cdot|\mathcal{F}_{t-1}\right)  .$

\bigskip

In Lemmas \ref{lemmaLAM} and \ref{lemmaN3i} below, we will refer to the
expected score, defined as:%
\begin{align}
\lambda_{T}(\mathbf{\theta})=T^{-1}\sum_{t=1}^{T} &  \mathbb{E}\left[
g_{t}(\mathbf{\theta})\right] \\
=T^{-1}\sum_{t=1}^{T} &  \mathbb{E}\left[  \frac{1}{-e_{t}(\mathbf{\theta}%
)}\left(  \frac{F_{t}\left(  v_{t}(\mathbf{\theta})\right)  }{\alpha
}-1\right)  \nabla v_{t}(\mathbf{\theta})^{\prime}+\right. \nonumber\\
&  \left.  \frac{1}{e_{t}(\mathbf{\theta})^{2}}\left(  \frac{F_{t}\left(
v_{t}(\mathbf{\theta})\right)  }{\alpha}v_{t}(\mathbf{\theta})-\frac{{1}%
}{\alpha}\mathbb{E}_{t-1}[Y_{t}|1\left\{  Y_{t}\leq v_{t}(\mathbf{\theta
})\right\}  ]-v_{t}(\mathbf{\theta})+e_{t}(\mathbf{\theta})\right)  \nabla
e_{t}(\mathbf{\theta})^{\prime}\right] \nonumber
\end{align}

\bigskip

\begin{lemma}
\label{lemmaLAM}Let
\begin{equation}
\Lambda(\mathbf{\theta}^{\ast})=T^{-1}\sum_{t=1}^{T}\left.  \frac
{\partial\mathbb{E}\left[  g_{t}(\mathbf{\theta})\right]  }{\partial
\mathbf{\theta}}\right\vert _{\mathbf{\theta}=\mathbf{\theta}^{\ast}}%
\end{equation}
Then under Assumptions 1-2,%
\begin{equation}
\sqrt{T}(\mathbf{\mathbf{\hat{\theta}}}-\mathbf{\theta}^{0})=\left(
\Lambda^{-1}(\mathbf{\theta}^{0})+o_{p}(1)\right)  \left(  -\frac{1}{\sqrt{T}%
}\sum_{t=1}^{T}g_{t}(\mathbf{\theta}^{0})+o_{p}(1)\right)
\end{equation}

\end{lemma}

\bigskip

\begin{proof}
[Proof of Lemma \ref{lemmaLAM}]Consider a mean-value expansion of $\lambda
_{T}(\mathbf{\hat{\theta}})$ around $\mathbf{\theta}^{0}$:
\begin{align}
\lambda_{T}(\mathbf{\hat{\theta}})= &  \lambda_{T}(\mathbf{\theta}^{0}%
)+T^{-1}\sum_{t=1}^{T}\left.  \frac{\partial\mathbb{E}\left[  g_{t}%
(\mathbf{\theta})\right]  }{\partial\mathbf{\theta}}\right\vert
_{\mathbf{\theta}=\mathbf{\theta}^{\ast}}(\mathbf{\hat{\theta}}-\mathbf{\theta
}^{0})\\
&  =\Lambda(\mathbf{\theta}^{\ast})(\mathbf{\hat{\theta}}-\mathbf{\theta}^{0})
\end{align}
where $\mathbf{\theta}^{\ast}$ lies between $\mathbf{\hat{\theta}}$ and
$\mathbf{\theta}^{0}$, and noting that $\lambda_{T}(\mathbf{\theta}^{0})=0$
and the definition of $\Lambda(\mathbf{\theta}^{\ast})$ given in the statement
of the lemma. Proving the claim involves two results: (I) $\Lambda
^{-1}(\mathbf{\theta}^{\ast})=\Lambda^{-1}(\mathbf{\theta}^{0})+o_{p}(1), $
and (II) $\sqrt{T}\lambda_{T}(\mathbf{\hat{\theta}})=-\frac{1}{\sqrt{T}}%
\sum_{t=1}^{T}g_{t}(\mathbf{\theta}^{0})+o_{p}(1).$ Part (I) is easy to
verify: Since $v_{t}(\mathbf{\theta})$ and $e_{t}(\mathbf{\theta})$ are twice
continuously differentiable, and $e_{t}(\mathbf{\theta}^{0})<0$ ,
$\Lambda(\mathbf{\theta})$ is continuous in $\mathbf{\theta}$ and
$\Lambda(\mathbf{\theta})$ is non-singular in a neighborhood of
$\mathbf{\theta}^{0}$. Then by the continuous mapping theorem, $\mathbf{\theta
}^{\ast}\overset{p}{\rightarrow}\mathbf{\theta}^{0}\Rightarrow\Lambda
(\mathbf{\theta}^{\ast})^{-1}\overset{p}{\rightarrow}\Lambda^{-1}%
(\mathbf{\theta}^{0})$. Establishing (II) builds on Theorem 3 of Huber (1967)
and Lemma A.1 of Weiss (1991), which extends Huber's conclusion to the case of
non-\textit{iid} dependent random variables. We are going to verify the
conditions of Weiss's Lemma A.1. Since the other conditions are easily
checked, we only need to show that $T^{-1/2}\sum_{t=1}^{T}g_{t}(\mathbf{\hat
{\theta}})=o_{p}(1)$, which we show in Lemma \ref{lemmaGop1}, and that his
assumptions N3 and N4 hold, which we show in Lemmas \ref{lemmaN3i}%
-\ref{lemmaN4}.
\end{proof}

\bigskip

\begin{lemma}
\label{lemmaGop1}Under Assumptions 1-2, $T^{-1/2}\sum_{t=1}^{T}g_{t}%
(\mathbf{\hat{\theta}})=o_{p}(1).$
\end{lemma}

\bigskip

\begin{proof}
[Proof of Lemma \ref{lemmaGop1}]Let $\{e_{j}\}_{j=1}^{p}$ be the standard
basis of $\mathbb{R}^{p}$ and define\newline%
\begin{equation}
L_{T}^{j}(a)=T^{-1/2}\sum_{t=1}^{T}L_{FZ0}\left(  Y_{t},v_{t}(\mathbf{\hat
{\theta}}+ae_{j}),e_{t}(\mathbf{\hat{\theta}}+ae_{j});\alpha\right)
\end{equation}
where a is a scalar. Following Ruppert and Carroll's (1980) approach, let
$G_{T}^{j}(a)$ (a scalar) be the right derivative of $L_{T}^{j}(a)$, that is
\begin{align}
G_{T}^{j}(a) &  =T^{-1/2}\sum_{t=1}^{T}\left(  \frac{\nabla_{j}v_{t}%
(\mathbf{\hat{\theta}}+ae_{j})}{-e_{t}(\mathbf{\hat{\theta}}+ae_{j})}(\frac
{1}{\alpha}\mathbf{1}\left\{  Y_{t}\leq v_{t}(\mathbf{\hat{\theta}}%
+ae_{j})\right\}  -1)+\right. \\
&  \left.  \frac{\nabla_{j}e_{t}(\mathbf{\hat{\theta}}+ae_{j})}{e_{t}%
(\mathbf{\hat{\theta}}+ae_{j})^{2}}\left(  \frac{1}{\alpha}\mathbf{1}\left\{
Y_{t}\leq v_{t}(\mathbf{\hat{\theta}}+ae_{j})\right\}  (v_{t}(\mathbf{\hat
{\theta}}+ae_{j})-Y_{t})-v_{t}(\mathbf{\hat{\theta}}+ae_{j})+e_{t}%
(\mathbf{\hat{\theta}}+ae_{j})\right)  \right) \nonumber
\end{align}
$G_{T}^{j}(0)=\lim\limits_{\xi_{1}\rightarrow0+}G_{T}^{j}(\xi_{1})$ is the
right partial derivative of $L_{T}(\mathbf{\theta})$ at $\mathbf{\hat{\theta}%
}$ in the direction $\theta_{j}$, while $\lim\limits_{\xi_{2}\rightarrow
0+}G_{T}^{j}(-\xi_{2})$ is the left partial derivative of $L_{T}%
(\mathbf{\theta})$ at $\mathbf{\hat{\theta}}$ in the direction $\theta_{j}$.
Although $L_{T}(\mathbf{\theta})$ is not differentiable, due to the presence
of the indicator function, its left and right derivatives do exist, and
because $L_{T}(\mathbf{\theta})$ achieves its minimum at $\mathbf{\hat{\theta
}}$, its left derivative must be non-positive and its right derivative must be
non-negative. Thus,
\begin{align}
|G_{T}^{j}(0)| &  \leq\lim\limits_{\xi_{1}\rightarrow0+}G_{T}^{j}(\xi
_{1})-\lim\limits_{\xi_{2}\rightarrow0+}G_{T}^{j}(-\xi_{2})\nonumber\\
&  =T^{-1/2}\sum_{t=1}^{T}\left(  \frac{|\nabla_{j}v_{t}(\mathbf{\hat{\theta}%
})|}{-e_{t}(\mathbf{\hat{\theta}})}\frac{1}{\alpha}\mathbf{1}\left\{
Y_{t}=v_{t}(\hat{\theta})\right\}  +\frac{|\nabla_{j}e_{t}(\mathbf{\hat
{\theta}})|}{e_{t}(\mathbf{\hat{\theta}})^{2}}\frac{1}{\alpha}\left(
v_{t}(\mathbf{\hat{\theta}})-Y_{t}\right)  \mathbf{1}\left\{  Y_{t}%
=v_{t}(\mathbf{\hat{\theta}})\right\}  \right) \\
&  =T^{-1/2}\sum_{t=1}^{T}\frac{|\nabla_{j}v_{t}(\mathbf{\hat{\theta}}%
)|}{-e_{t}(\hat{\theta})}\frac{1}{\alpha}\mathbf{1}\left\{  Y_{t}%
=v_{t}(\mathbf{\hat{\theta}})\right\} \nonumber
\end{align}
The second term in the penultimate line vanishes as $\mathbf{1}\{Y_{t}%
=v_{t}(\mathbf{\hat{\theta}})\}(v_{t}(\mathbf{\hat{\theta}})-y_{t})$ is always zero.

By Assumption 2(C), for all $t$, $|\nabla_{j}v_{t}(\mathbf{\hat{\theta}}%
)|\leq\lVert\nabla v_{t}(\mathbf{\hat{\theta}})\rVert\leq V_{1}(\mathcal{F}%
_{t-1})$, $\left\vert 1/e_{t}(\mathbf{\hat{\theta}})\right\vert \leq H$, thus:%
\begin{equation}
|G_{T}^{j}(0)|\leq\frac{H}{\alpha}\left[  T^{-1/2}\max\limits_{1\leq t\leq
T}V_{1}(\mathcal{F}_{t-1})\right]  \left[  \sum_{t=1}^{T}\mathbf{1}\left\{
Y_{t}=v_{t}(\mathbf{\hat{\theta}})\right\}  \right]
\end{equation}
$H$ is finite by Assumption 2(C), and for all $\epsilon>0,$%
\begin{equation}
\Pr\left[  T^{-1/2}\max\limits_{1\leq t\leq T}V_{1}(\mathcal{F}_{t-1}%
)>\epsilon\right]  \leq\sum\limits_{t=1}^{T}\Pr\left[  V_{1}(\mathcal{F}%
_{t-1})>\epsilon T^{1/2}\right]  \leq\sum\limits_{t=1}^{T}\frac{\mathbb{E}%
[V_{1}(\mathcal{F}_{t-1})^{3}]}{\epsilon^{3}T^{3/2}}\rightarrow0
\end{equation}
with the latter inequality following from Markov's inequality. Since
$\mathbb{E}[V_{1}(\mathcal{F}_{t-1})^{3}]$ is finite by assumption 2(D), we
then have that $T^{-1/2}\max\limits_{1\leq t\leq T}V_{1}(\mathcal{F}%
_{t-1})=o_{p}\left(  1\right)  .$ By Assumption 2(B) we have $\sum_{t=1}%
^{T}\mathbf{1}\left\{  y_{t}=v_{t}(\mathbf{\hat{\theta}})\right\}  =0$ a.s.
\ We therefore have $G_{T}^{j}(0)\overset{p}{\rightarrow}0$. Since this holds
for every $j$, we have $T^{-1/2}\sum_{t=1}^{T}g_{t}(\mathbf{\hat{\theta}%
})=o_{p}(1)$.\newline
\end{proof}

\bigskip

The following three lemmas show each of the three parts of Assumption N3 of
Weiss (1991) holds. In the proofs below we make repeated use of mean-value
expansions, and we use $\mathbf{\theta}^{\ast}$ to denote a point on the line
connecting $\mathbf{\hat{\theta}}$ and $\mathbf{\theta}^{0},$ and
$\mathbf{\theta}^{\ast\ast}$ to denote a point on the line connecting
$\mathbf{\theta}^{\ast}$ and $\mathbf{\theta}^{0}.$ The particular point on
the line can vary from expansion to expansion.

\bigskip

\begin{lemma}
\label{lemmaN3i}Under assumptions 1-2, Assumption N3(i) of Weiss (1991)
holds:
\[
\Vert\lambda_{T}(\mathbf{\theta})\Vert\geq a\Vert\mathbf{\theta}%
-\mathbf{\theta}^{0}\Vert,\;\text{for}\;\Vert\mathbf{\theta}-\mathbf{\theta
}^{0}\Vert\leq d_{0}\text{.}%
\]
for $T$ sufficiently large, where $a\ $and $d_{0}$ are strictly positive numbers.
\end{lemma}

\bigskip

\begin{proof}
[Proof of Lemma \ref{lemmaN3i}]A mean-value expansion yields:%
\begin{equation}
\lambda_{T}(\mathbf{\hat{\theta}})=\lambda_{T}(\mathbf{\theta}^{0}%
)+\Lambda_{T}(\mathbf{\theta}^{\ast})(\mathbf{\hat{\theta}}-\mathbf{\theta
}^{0})=\Lambda_{T}(\mathbf{\theta}^{\ast})(\mathbf{\hat{\theta}}%
-\mathbf{\theta}^{0})
\end{equation}
since $\lambda_{T}(\mathbf{\theta}^{0})=0,$ where $\Lambda_{T}(\mathbf{\theta
})=T^{-1}\sum_{t=1}^{T}\partial\mathbb{E}[g_{t}(\mathbf{\theta})]/\partial
\mathbf{\theta}.$ Using the fact that
\begin{equation}
\frac{\partial\mathbb{E}[Y_{t}\mathbf{1}\{Y_{t}\leq v_{t}(\mathbf{\theta
})\}|\mathcal{F}_{t-1}]}{\partial\mathbf{\theta}}=\frac{\partial}%
{\partial\mathbf{\theta}}\left\{  \int_{-\infty}^{v_{t}(\mathbf{\theta}%
)}yf_{t}(y)dy\right\}  =v_{t}(\mathbf{\theta})f_{t}(v_{t}(\mathbf{\theta
}))\nabla v_{t}(\mathbf{\theta})
\end{equation}
we can write:
\begin{align}
\Lambda_{T}(\mathbf{\theta})=T^{-1}\sum_{t=1}^{T} &  \mathbb{E}\left[  \left(
\frac{\nabla^{2}v_{t}(\mathbf{\theta})}{-e_{t}(\mathbf{\theta})}+\frac{\nabla
v_{t}(\mathbf{\theta})^{\prime}\nabla e_{t}(\mathbf{\theta})}{e_{t}%
(\mathbf{\theta})^{2}}+\frac{\nabla e_{t}(\mathbf{\theta})^{\prime}\nabla
v_{t}(\mathbf{\theta})}{e_{t}(\mathbf{\theta})^{2}}\right)  \left(
\frac{F_{t}(v_{t}(\mathbf{\theta}))}{\alpha}-1\right)  \right. \\
&  ~+\left(  \nabla^{2}e_{t}(\mathbf{\theta})\frac{1}{e_{t}(\mathbf{\theta
})^{2}}+\frac{-2}{e_{t}(\mathbf{\theta})^{3}}\nabla e_{t}(\mathbf{\theta
})^{\prime}\nabla e_{t}(\mathbf{\theta})\right) \nonumber\\
&  \text{\ \ \ }\cdot\left(  \left(  \frac{F_{t}(v_{t}(\mathbf{\theta}%
))}{\alpha}-1\right)  v_{t}(\mathbf{\theta})-\frac{1}{\alpha}\mathbb{E}\left[
Y_{t}\mathbf{1}\{Y_{t}\leq v_{t}(\mathbf{\theta})\}|\mathcal{F}_{t-1}\right]
+e_{t}(\mathbf{\theta})\right) \nonumber\\
&  ~+\frac{f_{t}(v_{t}(\mathbf{\theta}))}{-\alpha e_{t}(\mathbf{\theta}%
)}\nabla^{\prime}v_{t}(\mathbf{\theta})\nabla v_{t}(\mathbf{\theta
})\nonumber\\
&  ~\left.  \left.  +\frac{1}{e_{t}(\mathbf{\theta})^{2}}\nabla^{\prime}%
e_{t}(\mathbf{\theta})\nabla e_{t}(\mathbf{\theta})]\}\right\vert
\mathcal{F}_{t-1}\right] \nonumber
\end{align}
Evaluated at $\mathbf{\theta}^{0}\text{,}$ the first two terms of $\Lambda
_{T}$ drop out because $F_{t}\left(  v_{t}(\mathbf{\theta}^{0})\right)
=\alpha$ and $\frac{1}{\alpha}\mathbb{E}[Y_{t}\mathbf{1}\{Y_{t}\leq
v_{t}(\mathbf{\theta}^{0})\}|\mathcal{F}_{t-1}]=e_{t}\left(  \mathbf{\theta
}^{0}\right)  .$ Define $D_{T}$ as
\begin{equation}
D_{T}\equiv\Lambda_{T}(\mathbf{\theta}^{0})=T^{-1}\sum_{t=1}^{T}%
\mathbb{E}\left[  \frac{f_{t}(v_{t}(\mathbf{\theta}^{0}))}{-\alpha
e_{t}(\mathbf{\theta}^{0})}\nabla v_{t}(\mathbf{\theta}^{0})^{\prime}\nabla
v_{t}(\mathbf{\theta}^{0})+\frac{1}{e_{t}(\mathbf{\theta}^{0})^{2}}\nabla
e_{t}(\mathbf{\theta}^{0})^{\prime}\nabla e_{t}(\mathbf{\theta}^{0})\right]
\end{equation}
Below we show that $\Lambda_{T}(\mathbf{\theta}^{\ast})=D_{T}+O(\Vert
\mathbf{\hat{\theta}}-\mathbf{\theta}^{0}\Vert)$ by decomposing $\Vert
\Lambda_{T}(\mathbf{\theta}^{\ast})-D_{T}\Vert$ into four terms and showing
that each is bounded by a $O(\Vert\mathbf{\hat{\theta}}-\mathbf{\theta}%
^{0}\Vert)$ term.

\textbf{First term: }Using a mean-value expansion around $\mathbf{\theta}^{0}$
and Assumptions 2(C)-(D) we obtain:%
\begin{align}
& \left\Vert T^{-1}\sum_{t=1}^{T}\mathbb{E}\left[  \left(  \frac{\nabla
^{2}v_{t}(\mathbf{\theta}^{\ast})}{-e_{t}(\mathbf{\theta}^{\ast})}%
+\frac{\nabla v_{t}(\mathbf{\theta}^{\ast})^{\prime}\nabla e_{t}%
(\mathbf{\theta}^{\ast})}{e_{t}(\mathbf{\theta}^{\ast})^{2}}+\frac{\nabla
e_{t}(\mathbf{\theta}^{\ast})^{\prime}\nabla v_{t}(\mathbf{\theta}^{\ast}%
)}{e_{t}(\mathbf{\theta}^{\ast})^{2}}\right)  \left(  \frac{F_{t}%
(v_{t}(\mathbf{\theta}^{\ast}))}{{\alpha}}-1\right)  \right]  \right\Vert
\nonumber\\
& \leq T^{-1}\sum_{t=1}^{T}\mathbb{E}\left[  \left\Vert \left(  HV_{2}%
(\mathcal{F}_{t-1})+2H^{2}V_{1}(\mathcal{F}_{t-1})H_{1}(\mathcal{F}%
_{t-1})\right)  \left(  \frac{f_{t}(v_{t}(\mathbf{\theta}^{\ast\ast}%
))}{{\alpha}}\nabla v_{t}(\mathbf{\theta}^{\ast\ast})(\mathbf{\theta}^{\ast
}-\mathbf{\theta}^{0})\right)  \right\Vert \right] \\
& \leq T^{-1}\sum_{t=1}^{T}\frac{K}{\alpha}\left\{  H\mathbb{E}[V_{1}%
(\mathcal{F}_{t-1})^{3}]^{1/3}\mathbb{E}[V_{2}(\mathcal{F}_{t-1})^{3/2}%
]^{2/3}+2H^{2}\mathbb{E}[V_{1}(\mathcal{F}_{t-1})^{3}]^{2/3}\mathbb{E}%
[H_{1}(\mathcal{F}_{t-1})^{3}]^{1/3}\right\}  \Vert\mathbf{\theta}^{\ast
}-\mathbf{\theta}^{0}\Vert\nonumber
\end{align}

\textbf{Second term: }Again using a mean-value expansion around
$\mathbf{\theta}^{0}$ and Assumptions 2(C)-(D):%
\begin{align}
&  \left\Vert T^{-1}\sum_{t=1}^{T}\mathbb{E}\left[  \left(  \frac{1}%
{e_{t}(\mathbf{\theta}^{\ast})^{2}}\nabla^{2}e_{t}(\mathbf{\theta}^{\ast
})-\frac{2}{e_{t}(\mathbf{\theta}^{\ast})^{3}}\nabla e_{t}(\mathbf{\theta
}^{\ast})^{\prime}\nabla e_{t}(\mathbf{\theta}^{\ast})\right)  \right.
\right. \nonumber\\
&  ~~~~~~~~\left.  ~\left.  \cdot\left(  \left(  \frac{F_{t}(v_{t}%
(\mathbf{\theta}^{\ast}))}{{\alpha}}-1\right)  v_{t}(\mathbf{\theta}^{\ast
})-\frac{1}{\alpha}\mathbb{E}[Y_{t}\mathbf{1}\{Y_{t}\leq v_{t}(\mathbf{\theta
}^{\ast})\}|\mathcal{F}_{t-1}]+e_{t}(\mathbf{\theta}^{\ast})\right)  \right]
\right\Vert \nonumber\\
\leq &  T^{-1}\sum_{t=1}^{T}\mathbb{E[}\Vert\left(  H_{2}(\mathcal{F}%
_{t-1})H^{2}+H_{1}(\mathcal{F}_{t-1})\cdot2H^{3}\cdot H_{1}(\mathcal{F}%
_{t-1})\right) \\
&  ~~~~~~~~~~~~~~\cdot\left(  (F_{t}\left(  v_{t}(\mathbf{\theta}^{\ast\ast
})\right)  /{\alpha}-1)\nabla v_{t}(\mathbf{\theta}^{\ast\ast})+\nabla
e_{t}(\mathbf{\theta}^{\ast\ast})\right)  (\mathbf{\theta}^{\ast
}-\mathbf{\theta}^{0})\Vert]\nonumber\\
\leq &  T^{-1}\sum_{t=1}^{T}\{(1/{\alpha}+1)(H^{2}\mathbb{E}[V_{1}%
(\mathcal{F}_{t-1})H_{2}(\mathcal{F}_{t-1})]+2H^{3}\mathbb{E}[V_{1}%
(\mathcal{F}_{t-1})H_{1}(\mathcal{F}_{t-1})^{2}])\nonumber\\
&  ~~~~~~~~~~~+(H^{2}\cdot\mathbb{E}[H_{1}(\mathcal{F}_{t-1})H_{2}%
(\mathcal{F}_{t-1})]+2H^{3}\mathbb{E}[H_{1}(\mathcal{F}_{t-1})^{3}%
])\}\Vert\mathbf{\theta}^{\ast}-\mathbf{\theta}^{0}\Vert\nonumber
\end{align}

\textbf{Third term:}%
\begin{align}
&  \left\Vert T^{-1}\sum_{t=1}^{T}\mathbb{E}\left[  \frac{f_{t}(v_{t}%
(\mathbf{\theta}^{\ast}))}{-e_{t}(\theta^{\ast})\alpha}\nabla v_{t}%
(\mathbf{\theta}^{\ast})^{\prime}\nabla v_{t}(\mathbf{\theta}^{\ast}%
)-\frac{f_{t}(v_{t}(\mathbf{\theta}^{0}))}{-e_{t}(\mathbf{\theta}^{0})\alpha
}\nabla v_{t}(\theta^{0})^{\prime}\nabla v_{t}(\mathbf{\theta}^{0})\right]
\right\Vert \nonumber\\
&  =\frac{1}{\alpha}\left\Vert T^{-1}\sum_{t=1}^{T}\mathbb{E}\{\frac
{f_{t}(v_{t}(\mathbf{\theta}^{\ast}))}{-e_{t}(\mathbf{\theta}^{\ast})}\nabla
v_{t}(\mathbf{\theta}^{\ast})^{\prime}\nabla v_{t}(\mathbf{\theta}^{\ast
})-\frac{f_{t}(v_{t}(\mathbf{\theta}^{\ast}))}{-e_{t}(\mathbf{\theta}^{\ast}%
)}\nabla v_{t}(\mathbf{\theta}^{0})^{\prime}\nabla v_{t}(\mathbf{\theta}%
^{\ast})\right. \\
&  +\frac{f_{t}(v_{t}(\mathbf{\theta}^{\ast}))}{-e_{t}(\mathbf{\theta}^{\ast
})}\nabla v_{t}(\mathbf{\theta}^{0})^{\prime}\nabla v_{t}(\mathbf{\theta
}^{\ast})-\frac{f_{t}(v_{t}(\mathbf{\theta}^{0}))}{-e_{t}(\mathbf{\theta
}^{\ast})}\nabla v_{t}(\mathbf{\theta}^{0})^{\prime}\nabla v_{t}%
(\mathbf{\theta}^{\ast})\nonumber\\
&  +\frac{f_{t}(v_{t}(\mathbf{\theta}^{0}))}{-e_{t}(\mathbf{\theta}^{\ast}%
)}\nabla v_{t}(\mathbf{\theta}^{0})^{\prime}\nabla v_{t}(\mathbf{\theta}%
^{\ast})-\frac{f_{t}(v_{t}(\mathbf{\theta}^{0}))}{-e_{t}(\mathbf{\theta}^{0}%
)}\nabla v_{t}(\mathbf{\theta}^{0})^{\prime}\nabla v_{t}(\mathbf{\theta}%
^{\ast})\nonumber\\
&  \left.  +\frac{f_{t}(v_{t}(\mathbf{\theta}^{0}))}{-e_{t}(\mathbf{\theta
}^{0})}\nabla v_{t}(\mathbf{\theta}^{0})^{\prime}\nabla v_{t}(\mathbf{\theta
}^{\ast})-\frac{f_{t}(v_{t}(\mathbf{\theta}^{0}))}{-e_{t}(\mathbf{\theta}%
^{0})}\nabla v_{t}(\mathbf{\theta}^{0})^{\prime}\nabla v_{t}(\mathbf{\theta
}^{0})\}\right\Vert \nonumber\\
&  =\frac{1}{\alpha}\left\Vert T^{-1}\sum_{t=1}^{T}\mathbb{E}\{\frac
{f_{t}(v_{t}(\mathbf{\theta}^{\ast}))}{-e_{t}(\mathbf{\theta}^{\ast})}%
[\nabla^{2}v_{t}(\mathbf{\theta}^{\ast\ast})(\mathbf{\theta}^{\ast
}-\mathbf{\theta}^{0})]\nabla v_{t}(\theta^{\ast})\right. \nonumber\\
&  +\frac{f_{t}(v_{t}(\mathbf{\theta}^{\ast}))-f_{t}(v_{t}(\mathbf{\theta}%
^{0}))}{-e_{t}(\mathbf{\theta}^{\ast})}\nabla v_{t}(\mathbf{\theta}%
^{0})^{\prime}\nabla v_{t}(\mathbf{\theta}^{\ast})\nonumber\\
&  +\frac{f_{t}(v_{t}(\mathbf{\theta}^{0}))}{e_{t}(\mathbf{\theta}^{\ast\ast
})^{2}}(\mathbf{\theta}^{\ast}-\mathbf{\theta}^{0})\nabla v_{t}(\mathbf{\theta
}^{0})^{\prime}\nabla v_{t}(\mathbf{\theta}^{\ast})\nonumber\\
&  \left.  +\frac{f_{t}(v_{t}(\mathbf{\theta}^{0}))}{-e_{t}(\mathbf{\theta
}^{0})}\nabla v_{t}(\mathbf{\theta}^{0})^{\prime}(\mathbf{\theta}^{\ast
}-\mathbf{\theta}^{0})^{2}v_{t}(\mathbf{\theta}^{\ast\ast})\}\right\Vert
\nonumber\\
&  \leq\frac{1}{\alpha}T^{-1}\sum_{t=1}^{T}\mathbb{E}\{V_{2}(\mathcal{F}%
_{t-1})\left(  KH\cdot V_{1}(\mathcal{F}_{t-1})\right)  +KH\cdot
V_{1}(\mathcal{F}_{t-1})^{3}\nonumber\\
&  +KH^{2}H_{1}(\mathcal{F}_{t-1})V_{1}(\mathcal{F}_{t-1})^{2}+KHV_{1}%
(\mathcal{F}_{t-1})V_{2}(\mathcal{F}_{t-1})\}\cdot\Vert\mathbf{\theta}^{\ast
}-\mathbf{\theta}^{0}\Vert\nonumber
\end{align}

\textbf{Fourth term: }The bound on this term follows similar steps to that of
the third term:%
\begin{align}
& \left\Vert T^{-1}\sum_{t=1}^{T}\mathbb{E}\{\frac{{1}}{e_{t}(\mathbf{\theta
}^{\ast})^{2}}\nabla e_{t}(\mathbf{\theta}^{\ast})^{\prime}\nabla
e_{t}(\mathbf{\theta}^{\ast})-\frac{{1}}{e_{t}(\mathbf{\theta}^{0})^{2}}\nabla
e_{t}(\mathbf{\theta}^{0})^{\prime}\nabla e_{t}(\mathbf{\theta}^{0}%
)\}\right\Vert \nonumber\\
& =\left\Vert T^{-1}\sum_{t=1}^{T}\mathbb{E}\{\frac{{1}}{e_{t}(\mathbf{\theta
}^{\ast})^{2}}\nabla e_{t}(\mathbf{\theta}^{\ast})^{\prime}\nabla
e_{t}(\mathbf{\theta}^{\ast})-\frac{{1}}{e_{t}(\mathbf{\theta}^{\ast})^{2}%
}\nabla e_{t}(\mathbf{\theta}^{0})^{\prime}\nabla e_{t}(\mathbf{\theta}^{\ast
})\right. \\
& +\frac{{1}}{e_{t}(\mathbf{\theta}^{\ast})^{2}}\nabla e_{t}(\mathbf{\theta
}^{0})^{\prime}\nabla e_{t}(\mathbf{\theta}^{\ast})-\frac{{1}}{e_{t}%
(\mathbf{\theta}^{0})^{2}}\nabla e_{t}(\mathbf{\theta}^{0})^{\prime}\nabla
e_{t}(\mathbf{\theta}^{\ast})\nonumber\\
& \left.  +\frac{{1}}{e_{t}(\mathbf{\theta}^{0})^{2}}\nabla e_{t}%
(\mathbf{\theta}^{0})^{\prime}\nabla e_{t}(\mathbf{\theta}^{\ast})-\frac{{1}%
}{e_{t}(\mathbf{\theta}^{0})^{2}}\nabla e_{t}(\mathbf{\theta}^{0})^{\prime
}\nabla e_{t}(\mathbf{\theta}^{0})\}\right\Vert |\nonumber\\
& \leq T^{-1}\sum_{t=1}^{T}\{H^{2}\cdot\mathbb{E}[H_{1}(\mathcal{F}%
_{t-1})H_{2}(\mathcal{F}_{t-1})]+2H^{3}\mathbb{E}\left[  H_{1}(\mathcal{F}%
_{t-1})^{3}\right]  +H^{2}\mathbb{E}[H_{1}(\mathcal{F}_{t-1})H_{2}%
(\mathcal{F}_{t-1})]\}\Vert\mathbf{\theta}^{\ast}-\mathbf{\theta}^{0}%
\Vert\nonumber
\end{align}
Therefore, $\Lambda_{T}(\mathbf{\theta}^{\ast})=D_{T}+O(\Vert\mathbf{\hat
{\theta}}-\mathbf{\theta}^{0}\Vert)\Rightarrow\Vert\Lambda_{T}(\mathbf{\theta
}^{\ast})-D_{T}\Vert\leq K\Vert\mathbf{\hat{\theta}}-\mathbf{\theta}^{0}%
\Vert,$where $K$ is some constant $<\infty$ , for T sufficiently large. By
Assumption 2(E), $D_{T}$ has eigenvalues bounded below by a positive constant,
denoted as $a,$ for $T$ sufficiently large. Thus,
\begin{align}
\Vert\lambda_{T}(\mathbf{\hat{\theta}})\Vert &  =\Vert\Lambda_{T}%
(\mathbf{\theta}^{\ast})\left(  \mathbf{\hat{\theta}}-\mathbf{\theta}%
^{0}\right)  \Vert\nonumber\\
&  =\Vert D_{T}(\mathbf{\hat{\theta}}-\mathbf{\theta}^{0})-(D_{T}-\Lambda
_{T}(\mathbf{\theta}^{\ast}))(\mathbf{\hat{\theta}}-\mathbf{\theta}^{0}%
)\Vert\\
&  \geq\Vert D_{T}(\mathbf{\hat{\theta}}-\mathbf{\theta}^{0})\Vert-\Vert
(D_{T}-\Lambda_{T}(\mathbf{\theta}^{\ast}))(\mathbf{\hat{\theta}%
}-\mathbf{\theta}^{0})\Vert\nonumber\\
&  \geq(a-K\Vert\mathbf{\hat{\theta}}-\mathbf{\theta}^{0}\Vert)\cdot
\Vert\mathbf{\hat{\theta}}-\mathbf{\theta}^{0}\Vert\nonumber
\end{align}
The penultimate inequality holds by the triangle inequality, and the final
inequality follows from Assumption 2(E) on the minimum eigenvalue of $D_{T}.$
Thus, for $T$ sufficiently large so that \linebreak$a-K\Vert\mathbf{\hat
{\theta}}-\mathbf{\theta}^{0}\Vert>0,$ the result follows.
\end{proof}

\bigskip

\begin{lemma}
\label{lemmaN3ii}Define
\begin{equation}
\mu_{t}(\mathbf{\theta},d)=\underset{\Vert\mathbf{\tau}-\mathbf{\theta}%
\Vert\leq d}{\sup}\Vert g_{t}(\mathbf{\tau})-g_{t}(\mathbf{\theta})\Vert
\end{equation}
Then under assumptions 1-2, Assumption N3(ii) of Weiss (1991) holds
\begin{equation}
\mathbb{E}[\mu_{t}(\mathbf{\theta},d)]\leq bd,\text{ for }\Vert\mathbf{\theta
}-\mathbf{\theta}^{0}\Vert+d\leq d_{0},d\geq0
\end{equation}
for $T$ sufficiently large, where $b,~d,$and $d_{0}$ are strictly positive numbers.
\end{lemma}

\bigskip

\begin{proof}
[Proof of Lemma \ref{lemmaN3ii}]In this proof, the strictly positive constant
$c$ and the mean-value expansion term, $\mathbf{\tau}^{\ast}$, can change from
line to line. Pick $d_{0}$ such that for any $\mathbf{\theta}$ that satisfies
$\Vert\mathbf{\theta}-\mathbf{\theta}^{0}\Vert\leq d_{0}$, all the conditions
in Assumption 2(C) and 2(D) hold as well as $e_{t}(\mathbf{\theta})\leq
v_{t}(\mathbf{\theta})\leq0$. Let us expand $g_{t}(\mathbf{\theta})$ into six
terms:%
\begin{align}
g_{t}(\mathbf{\theta})=\frac{1}{\alpha} &  \frac{\nabla^{\prime}%
v_{t}(\mathbf{\theta})}{-e_{t}(\mathbf{\theta})}\mathbf{1}\{Y_{t}\leq
v_{t}(\mathbf{\theta})\}-\frac{\nabla^{\prime}v_{t}(\mathbf{\theta})}%
{-e_{t}(\mathbf{\theta})}+\frac{1}{\alpha}\frac{v_{t}(\mathbf{\theta}%
)\nabla^{\prime}e_{t}(\mathbf{\theta})}{e_{t}(\mathbf{\theta})^{2}}%
\mathbf{1}\{Y_{t}\leq v_{t}(\mathbf{\theta})\}\\
&  -\frac{v_{t}(\mathbf{\theta})\nabla^{\prime}e_{t}(\mathbf{\theta})}%
{e_{t}(\mathbf{\theta})^{2}}-\frac{1}{\alpha}\frac{\nabla^{\prime}%
e_{t}(\mathbf{\theta})}{e_{t}(\mathbf{\theta})^{2}}\mathbf{1}\{Y_{t}\leq
v_{t}(\mathbf{\theta})\}Y_{t}+\frac{\nabla^{\prime}e_{t}(\mathbf{\theta}%
)}{e_{t}(\mathbf{\theta})}\nonumber
\end{align}
We will bound $\mu_{t}(\mathbf{\theta},d)$ by considering six terms, $\mu
_{t}(\mathbf{\theta},d)^{(i)},i=1,2,\cdots,6$, defined below. Each term is
shown to be bounded by a constant times $d$.

\textbf{First term:}%
\begin{equation}
\mu_{t}(\mathbf{\theta},d)^{(1)}=\frac{1}{\alpha}\sup\limits_{\Vert
\mathbf{\tau}-\mathbf{\theta}\Vert\leq d}\left\Vert \frac{\nabla^{\prime}%
v_{t}(\mathbf{\tau})}{-e_{t}(\mathbf{\tau})}\mathbf{1}\{Y_{t}\leq
v_{t}(\mathbf{\tau})\}-\frac{\nabla^{\prime}v_{t}(\mathbf{\theta})}%
{-e_{t}(\mathbf{\theta})}\mathbf{1}\{Y_{t}\leq v_{t}(\mathbf{\theta
})\}\right\Vert
\end{equation}
Set $\mathbf{\tau}_{1}=\arg\min_{\Vert\mathbf{\tau}-\mathbf{\theta}\Vert\leq
d}v_{t}(\mathbf{\tau})$ and $\mathbf{\tau}_{2}=\arg\max_{\Vert\mathbf{\tau
}-\mathbf{\theta}\Vert\leq d}v_{t}(\mathbf{\tau})$. Since $v_{t}%
(\mathbf{\theta})$ and $e_{t}(\mathbf{\theta})$ are assumed to be twice
continously differentiable, $\mathbf{\tau}_{1}$ and $\mathbf{\tau}_{2}$ exist.
We want to take the indicator function out from the `sup' operator. To this
end, let us discuss what $\alpha\cdot\mu_{t}(\mathbf{\theta},d)^{(1)}$ equals
in two cases.

Case 1: $Y_{t}\leq v_{t}(\mathbf{\theta})$. (a)\ If $Y_{t}>v_{t}(\mathbf{\tau
}_{2})$, $\alpha\cdot\mu_{t}(\mathbf{\theta},d)^{(1)}=\left\Vert \frac
{\nabla^{\prime}v_{t}(\mathbf{\theta})}{-e_{t}(\mathbf{\theta})}\right\Vert $.
(b) If $Y_{t}<v_{t}(\mathbf{\tau}_{1})$, $\alpha\cdot\mu_{t}(\mathbf{\theta
},d)^{(1)}=\sup\limits_{\Vert\mathbf{\tau}-\mathbf{\theta}\Vert\leq
d}\left\Vert \frac{\nabla^{\prime}v_{t}(\mathbf{\tau})}{-e_{t}(\mathbf{\tau}%
)}-\frac{\nabla^{\prime}v_{t}(\mathbf{\theta})}{-e_{t}(\mathbf{\theta}%
)}\right\Vert $. \ (c) If $v_{t}(\mathbf{\tau}_{1})\leq Y_{t}\leq
v_{t}(\mathbf{\tau}_{2})$,
\begin{align}
\alpha\cdot\mu_{t}(\mathbf{\theta},d)^{(1)}= &  \max\left\{  \sup
\limits_{\Vert\mathbf{\tau}-\mathbf{\theta}\Vert\leq d,Y_{t}\leq
v(\mathbf{\tau})}\left\Vert \frac{\nabla^{\prime}v_{t}(\mathbf{\tau})}%
{-e_{t}(\mathbf{\tau})}-\frac{\nabla^{\prime}v_{t}(\mathbf{\theta})}%
{-e_{t}(\mathbf{\theta})}\right\Vert ,\left\Vert \frac{\nabla^{\prime}%
v_{t}(\mathbf{\theta})}{-e_{t}(\mathbf{\theta})}\right\Vert \right\} \\
\leq &  \sup\limits_{\Vert\mathbf{\tau}-\mathbf{\theta}\Vert\leq d}\left\Vert
\frac{\nabla^{\prime}v_{t}(\mathbf{\tau})}{-e_{t}(\mathbf{\tau})}-\frac
{\nabla^{\prime}v_{t}(\mathbf{\theta})}{-e_{t}(\mathbf{\theta})}\right\Vert
+\left\Vert \frac{\nabla^{\prime}v_{t}(\mathbf{\theta})}{-e_{t}(\mathbf{\theta
})}\right\Vert \nonumber
\end{align}

Case 2: $Y_{t}>v_{t}(\mathbf{\theta})$,
\begin{align}
\alpha\cdot\mu_{t}(\mathbf{\theta},d)^{(1)}  & =\mathbf{1}\{Y_{t}\leq
v(\mathbf{\tau}_{2})\}\cdot\sup\limits_{\Vert\mathbf{\tau}-\mathbf{\theta
}\Vert\leq d,Y_{t}\leq v(\mathbf{\tau})}\left\Vert \frac{\nabla^{\prime}%
v_{t}(\mathbf{\tau})}{-e_{t}(\mathbf{\tau})}\right\Vert \\
& \leq\mathbf{1}\{Y_{t}\leq v(\mathbf{\tau}_{2})\}\cdot\sup\limits_{\Vert
\mathbf{\tau}-\mathbf{\theta}\Vert\leq d}\left\Vert \frac{\nabla^{\prime}%
v_{t}(\mathbf{\tau})}{-e_{t}(\mathbf{\tau})}\right\Vert \nonumber
\end{align}
$\Vert\mathbf{\theta}-\mathbf{\theta}^{0}\Vert+d\leq d_{0}$ implies that both
$\mathbf{\theta}$ and $\mathbf{\tau}$ (which are in a $d$-neighborhood of
$\mathbf{\theta}$) are in a $d_{0}$-neighborhood of $\mathbf{\theta}_{0}$ ,
and so
\begin{equation}
\left\Vert \frac{\nabla^{\prime}v_{t}(\mathbf{\theta})}{-e_{t}(\mathbf{\theta
})}\right\Vert \leq\sup\limits_{\Vert\mathbf{\tau}-\mathbf{\theta}\Vert\leq
d}\left\Vert \frac{\nabla^{\prime}v_{t}(\mathbf{\tau})}{-e_{t}(\mathbf{\tau}%
)}\right\Vert \leq\sup\limits_{\Vert\mathbf{\theta}-\mathbf{\theta}^{0}%
\Vert\leq d_{0}}\left\Vert \frac{\nabla^{\prime}v_{t}(\mathbf{\theta})}%
{-e_{t}(\mathbf{\theta})}\right\Vert
\end{equation}
Thus,
\begin{align}
&  \alpha\cdot\mu_{t}(\mathbf{\theta},d)^{(1)}\nonumber\\
\leq &  \left(  \mathbf{1}\{v_{t}(\mathbf{\tau}_{2})<Y_{t}\leq v_{t}%
(\mathbf{\theta})\}+\mathbf{1}\{v_{t}(\mathbf{\tau}_{1})\leq Y_{t}\leq
v_{t}(\mathbf{\theta})\}+\mathbf{1}\{v_{t}(\mathbf{\theta})<Y_{t}\leq
v_{t}(\mathbf{\tau}_{2})\}\right) \\
&  \cdot\sup\limits_{\Vert\mathbf{\theta}-\mathbf{\theta}^{0}\Vert\leq d_{0}%
}\left\Vert \frac{\nabla^{\prime}v_{t}(\mathbf{\theta})}{-e_{t}(\mathbf{\theta
})}\right\Vert +\sup\limits_{\Vert\mathbf{\tau}-\mathbf{\theta}\Vert\leq
d}\left\Vert \frac{\nabla^{\prime}v_{t}(\mathbf{\tau})}{-e_{t}(\mathbf{\tau}%
)}-\frac{\nabla^{\prime}v_{t}(\mathbf{\theta})}{-e_{t}(\mathbf{\theta}%
)}\right\Vert ,\nonumber
\end{align}
where
\begin{align}
\mathbb{E}_{t-1}[\mathbf{1}\{v_{t}(\mathbf{\tau}_{2})<Y_{t}\leq v_{t}%
(\mathbf{\theta})\}]= &  \int_{v_{t}(\mathbf{\tau}_{2})}^{v_{t}(\mathbf{\theta
})}f_{t}(y)dy\\
\leq &  K|v_{t}(\mathbf{\tau}_{2})-v_{t}(\mathbf{\theta})|\leq KV_{1}%
(\mathcal{F}_{t-1})\Vert\mathbf{\tau}_{2}-\mathbf{\theta}\Vert\leq
KV_{1}(\mathcal{F}_{t-1})d\nonumber
\end{align}
and similarly,%
\begin{align}
\mathbb{E}\left[  \mathbf{1}\{v_{t}(\mathbf{\theta})<Y_{t}\leq v_{t}%
(\mathbf{\tau}_{2})\}|\mathcal{F}_{t-1}\right]   & \leq KV_{1}(\mathcal{F}%
_{t-1})d\\
\text{and \ }\mathbb{E}\left[  \mathbf{1}\{v_{t}(\mathbf{\tau}_{1})<Y_{t}\leq
v_{t}(\mathbf{\theta})\}|\mathcal{F}_{t-1}\right]   & \leq KV_{1}%
(\mathcal{F}_{t-1})d\nonumber
\end{align}
Further
\begin{equation}
\sup\limits_{\Vert\mathbf{\theta}-\mathbf{\theta}^{0}\Vert\leq d}\left\Vert
\frac{\nabla^{\prime}v_{t}(\mathbf{\theta})}{-e_{t}(\mathbf{\theta}%
)}\right\Vert \leq HV_{1}(\mathcal{F}_{t-1})
\end{equation}
and by the mean-value theorem,
\begin{equation}
\frac{\nabla^{\prime}v_{t}(\mathbf{\tau})}{-e_{t}(\mathbf{\tau})}-\frac
{\nabla^{\prime}v_{t}(\mathbf{\theta})}{-e_{t}(\mathbf{\theta})}=\left\Vert
\frac{\nabla^{2}v_{t}(\mathbf{\tau}^{\ast})}{-e_{t}(\mathbf{\tau}^{\ast}%
)}+\frac{\nabla^{\prime}v_{t}(\mathbf{\tau}^{\ast})\nabla e_{t}(\mathbf{\tau
}^{\ast})}{e_{t}(\mathbf{\tau}^{\ast})^{2}}\right\Vert \cdot(\mathbf{\tau
}-\mathbf{\theta})
\end{equation}%
\begin{equation}
\Rightarrow\sup\limits_{\Vert\mathbf{\tau}-\mathbf{\theta}\Vert\leq
d}\left\Vert \frac{\nabla^{\prime}v_{t}(\mathbf{\tau})}{-e_{t}(\mathbf{\tau}%
)}-\frac{\nabla^{\prime}v_{t}(\mathbf{\theta})}{-e_{t}(\mathbf{\theta}%
)}\right\Vert \leq\left(  HV_{2}(\mathcal{F}_{t-1})+H^{2}V_{1}(\mathcal{F}%
_{t-1})H_{1}(\mathcal{F}_{t-1})\right)  \cdot d.
\end{equation}
By Assumption 2(D), $\mathbb{E}[V_{2}(\mathcal{F}_{t-1})]$ and $\mathbb{E}%
[V_{1}(\mathcal{F}_{t-1})H_{1}(\mathcal{F}_{t-1})]$ are finite, so
$\mathbb{E}[\mu_{t}(\mathbf{\theta},d)^{(1)}]\leq cd$, where c is a strictly
positive constant.\newline

\textbf{Second term: }$\mu_{t}(\mathbf{\theta},d)^{(2)}=\sup\limits_{\Vert
\mathbf{\tau}-\mathbf{\theta}\Vert\leq d}\left\Vert \frac{\nabla^{\prime}%
v_{t}(\mathbf{\tau})}{-e_{t}(\mathbf{\tau})}-\frac{\nabla^{\prime}%
v_{t}(\mathbf{\theta})}{-e_{t}(\mathbf{\theta})}\right\Vert .$ It was shown in
the derivations for the first term that $\mathbb{E}[\mu_{t}(\mathbf{\theta
},d)^{(2)}]\leq cd$, where c is a strictly positive constant.\newline

\textbf{Third term:}
\begin{equation}
\mu_{t}(\mathbf{\theta},d)^{(3)}=\frac{1}{\alpha}\sup\limits_{\Vert
\mathbf{\tau}-\mathbf{\theta}\Vert\leq d}\left\Vert \frac{v_{t}(\mathbf{\tau
})\nabla^{\prime}e_{t}(\mathbf{\tau})}{e_{t}(\mathbf{\tau})^{2}}%
\mathbf{1}\{Y_{t}\leq v_{t}(\mathbf{\tau})\}-\frac{v_{t}(\mathbf{\theta
})\nabla^{\prime}e_{t}(\mathbf{\theta})}{e_{t}(\mathbf{\theta})^{2}}%
\mathbf{1}\{Y_{t}\leq v_{t}(\mathbf{\theta})\}\right\Vert
\end{equation}
Similar to the first term, $\alpha\cdot\mu_{t}(\mathbf{\theta},d)^{(3)}$ can
be bounded by
\begin{align}
&  \left(  \mathbf{1}\{v_{t}(\mathbf{\tau}_{2})<Y_{t}\leq v_{t}(\mathbf{\theta
})\}+\mathbf{1}\{v_{t}(\mathbf{\tau}_{1})\leq Y_{t}\leq v_{t}(\mathbf{\theta
})\}+\mathbf{1}\{v_{t}(\mathbf{\theta})<Y_{t}\leq v_{t}(\mathbf{\tau}%
_{2})\}\right) \\
&  \cdot\sup\limits_{\Vert\mathbf{\theta}-\mathbf{\theta}^{0}\Vert\leq d_{0}%
}\left\Vert \frac{v_{t}(\mathbf{\theta})\nabla^{\prime}e_{t}(\mathbf{\theta}%
)}{e_{t}(\mathbf{\theta})^{2}}\right\Vert +\sup\limits_{\Vert\mathbf{\tau
}-\mathbf{\theta}\Vert\leq d}\left\Vert \frac{v_{t}(\mathbf{\tau}%
)\nabla^{\prime}e_{t}(\mathbf{\tau})}{e_{t}(\mathbf{\tau})^{2}}-\frac
{v_{t}(\mathbf{\theta})\nabla^{\prime}e_{t}(\mathbf{\theta})}{e_{t}%
(\mathbf{\theta})^{2}}\right\Vert \nonumber
\end{align}
where
\begin{equation}
\mathbb{E}\left[  \mathbf{1}\{v_{t}(\mathbf{\tau}_{2})<Y_{t}\leq
v_{t}(\mathbf{\theta})\}+\mathbf{1}\{v_{t}(\mathbf{\tau}_{1})\leq Y_{t}\leq
v_{t}(\mathbf{\theta})\}+\mathbf{1}\{v_{t}(\mathbf{\theta})<Y_{t}\leq
v_{t}(\mathbf{\tau}_{2})\}|\mathcal{F}_{t-1}\right]  \leq3KV_{1}%
(\mathcal{F}_{t-1})d
\end{equation}
and
\begin{equation}
\sup\limits_{\Vert\mathbf{\theta}-\mathbf{\theta}^{0}\Vert\leq d}\left\Vert
\frac{v_{t}(\mathbf{\theta})\nabla^{\prime}e_{t}(\mathbf{\theta})}%
{e_{t}(\mathbf{\theta})^{2}}\right\Vert \leq H\cdot H_{1}(\mathcal{F}_{t-1})
\end{equation}
where $e_{t}(\mathbf{\theta})\leq v_{t}(\mathbf{\theta})\leq0$ is used, and by
the mean-value theorem,
\begin{align}
&  \frac{v_{t}(\mathbf{\tau})\nabla^{\prime}e_{t}(\mathbf{\tau})}%
{e_{t}(\mathbf{\tau})^{2}}-\frac{v_{t}(\mathbf{\theta})\nabla^{\prime}%
e_{t}(\mathbf{\theta})}{e_{t}(\mathbf{\theta})^{2}}\\
= &  \left\Vert \frac{\nabla^{\prime}e_{t}(\mathbf{\tau}^{\ast})\nabla
v_{t}(\mathbf{\tau}^{\ast})}{e_{t}(\mathbf{\tau}^{\ast})^{2}}-\frac
{2v_{t}(\mathbf{\tau}^{\ast})\nabla^{\prime}e_{t}(\mathbf{\tau}^{\ast})\nabla
e_{t}(\mathbf{\tau}^{\ast})}{e_{t}(\mathbf{\tau}^{\ast})^{3}}+\frac
{v_{t}(\mathbf{\tau}^{\ast})\nabla^{2}e_{t}(\mathbf{\tau}^{\ast})}%
{e_{t}(\mathbf{\tau}^{\ast})^{2}}\right\Vert \cdot(\mathbf{\tau}%
-\mathbf{\theta})\nonumber
\end{align}%
\begin{align}
\Rightarrow &  \sup\limits_{\Vert\mathbf{\tau}-\mathbf{\theta}\Vert\leq
d}\left\Vert \frac{v_{t}(\mathbf{\tau})\nabla^{\prime}e_{t}(\mathbf{\tau}%
)}{e_{t}(\mathbf{\tau})^{2}}-\frac{v_{t}(\mathbf{\theta})\nabla^{\prime}%
e_{t}(\mathbf{\theta})}{e_{t}(\mathbf{\theta})^{2}}\right\Vert \\
\leq &  \left(  H^{2}V_{1}(\mathcal{F}_{t-1})H_{1}(\mathcal{F}_{t-1}%
)+2H^{2}H_{1}(\mathcal{F}_{t-1})^{2}+H\cdot H_{2}(\mathcal{F}_{t-1})\right)
\cdot d\nonumber
\end{align}
By Assumption 2(D), $\mathbb{E}[V_{1}(\mathcal{F}_{t-1})H_{1}(\mathcal{F}%
_{t-1})]$, $\mathbb{E}[H_{1}(\mathcal{F}_{t-1})^{2}]$, $\mathbb{E}%
[H_{2}(\mathcal{F}_{t-1})]$ $<\infty$. Therefore, $\mathbb{E}[\mu
_{t}(\mathbf{\theta},d)^{(3)}]\leq cd$, where c is a strictly positive
constant.\newline

\textbf{Fourth term: }$\mu_{t}(\mathbf{\theta},d)^{(4)}=\sup\limits_{\Vert
\mathbf{\tau}-\mathbf{\theta}\Vert\leq d}\left\Vert \frac{v_{t}(\mathbf{\tau
})\nabla^{\prime}e_{t}(\mathbf{\tau})}{e_{t}(\mathbf{\tau})^{2}}-\frac
{v_{t}(\mathbf{\theta})\nabla^{\prime}e_{t}(\mathbf{\theta})}{e_{t}%
(\mathbf{\theta})^{2}}\right\Vert .$ In the derivations for the third term we
showed that $\mathbb{E}[\mu_{t}(\mathbf{\theta},d)^{(4)}]\leq cd$, where c is
a strictly positive constant.\newline

\textbf{Fifth term:}
\begin{equation}
\mu_{t}(\mathbf{\theta},d)^{(5)}=\frac{1}{\alpha}\sup\limits_{\Vert
\mathbf{\tau}-\mathbf{\theta}\Vert\leq d}\left\Vert \frac{\nabla^{\prime}%
e_{t}(\mathbf{\tau})}{e_{t}(\mathbf{\tau})^{2}}\mathbf{1}\{Y_{t}\leq
v_{t}(\mathbf{\tau})\}Y_{t}-\frac{\nabla^{\prime}e_{t}(\mathbf{\theta})}%
{e_{t}(\mathbf{\theta})^{2}}\mathbf{1}\{Y_{t}\leq v_{t}(\mathbf{\theta
})\}Y_{t}\right\Vert
\end{equation}
Similar to the first term, $\alpha\cdot\mu_{t}(\mathbf{\theta},d)^{(5)}$ can
be bounded by
\begin{align}
&  \left(  \mathbf{1}\{v_{t}(\mathbf{\tau}_{2})<Y_{t}\leq v_{t}(\mathbf{\theta
})\}+\mathbf{1}\{v_{t}(\mathbf{\tau}_{1})\leq Y_{t}\leq v_{t}(\mathbf{\theta
})\}+\mathbf{1}\{v_{t}(\mathbf{\theta})<Y_{t}\leq v_{t}(\mathbf{\tau}%
_{2})\}\right) \\
&  \cdot|Y_{t}|\sup\limits_{\Vert\mathbf{\theta}-\mathbf{\theta}^{0}\Vert\leq
d_{0}}\left\Vert \frac{\nabla^{\prime}e_{t}(\mathbf{\theta})}{e_{t}%
(\mathbf{\theta})^{2}}\right\Vert +|Y_{t}|\sup\limits_{\Vert\mathbf{\tau
}-\mathbf{\theta}\Vert\leq d}\left\Vert \frac{\nabla^{\prime}e_{t}%
(\mathbf{\tau})}{e_{t}(\mathbf{\tau})^{2}}-\frac{\nabla^{\prime}%
e_{t}(\mathbf{\theta})}{e_{t}(\mathbf{\theta})^{2}}\right\Vert \nonumber
\end{align}
where
\begin{align}
\mathbb{E}\left[  \mathbf{1}\{v_{t}(\mathbf{\tau}_{2})<Y_{t}\leq
v_{t}(\mathbf{\theta})\}|Y_{t}|~|\mathcal{F}_{t-1}\right]  = &  \int
_{v_{t}(\mathbf{\tau}_{2})}^{v_{t}(\mathbf{\theta})}|y|f_{t}(y)dy\leq
K|v_{t}(\mathbf{\tau}_{2})|\cdot|v_{t}(\mathbf{\tau}_{2})-v_{t}(\mathbf{\theta
})|\\
&  \leq KV(\mathcal{F}_{t-1})V_{1}(\mathcal{F}_{t-1})\Vert\mathbf{\tau}%
_{2}-\mathbf{\theta}\Vert\leq KV(\mathcal{F}_{t-1})V_{1}(\mathcal{F}%
_{t-1})d\nonumber
\end{align}
and similarly,
\begin{align}
\mathbb{E}\left[  \mathbf{1}\{v_{t}(\mathbf{\tau}_{1})<Y_{t}\leq
v_{t}(\mathbf{\theta})\}|Y_{t}|~|\mathcal{F}_{t-1}\right]   & \leq
KV(\mathcal{F}_{t-1})V_{1}(\mathcal{F}_{t-1})d\\
\text{and \ \ }\mathbb{E}\left[  \mathbf{1}\{v_{t}(\mathbf{\theta})<Y_{t}\leq
v_{t}(\mathbf{\tau}_{2})\}|Y_{t}|~|\mathcal{F}_{t-1}\right]   & \leq
KV(\mathcal{F}_{t-1})V_{1}(\mathcal{F}_{t-1})d\nonumber
\end{align}
Further
\begin{equation}
\sup\limits_{\Vert\mathbf{\theta}-\mathbf{\theta}^{0}\Vert\leq d}\left\Vert
\frac{\nabla^{\prime}e_{t}(\mathbf{\theta})}{e_{t}(\mathbf{\theta})^{2}%
}\right\Vert \leq H^{2}H_{1}(\mathcal{F}_{t-1})
\end{equation}
and by the mean-value theorem,
\begin{gather}
\frac{\nabla^{\prime}e_{t}(\mathbf{\tau})}{e_{t}(\mathbf{\tau})^{2}}%
-\frac{\nabla^{\prime}e_{t}(\mathbf{\theta})}{e_{t}(\mathbf{\theta})^{2}%
}=\left\Vert -\frac{2\nabla^{\prime}e_{t}(\mathbf{\tau}^{\ast})\nabla
e_{t}(\mathbf{\tau}^{\ast})}{e_{t}(\mathbf{\tau}^{\ast})^{3}}+\frac{\nabla
^{2}e_{t}(\mathbf{\tau}^{\ast})}{e_{t}(\mathbf{\tau}^{\ast})^{2}}\right\Vert
\cdot(\mathbf{\tau}-\mathbf{\theta})\\
\Rightarrow\sup\limits_{\Vert\mathbf{\tau}-\mathbf{\theta}\Vert\leq
d}\left\Vert \frac{\nabla^{\prime}e_{t}(\mathbf{\tau})}{e_{t}(\mathbf{\tau
})^{2}}-\frac{\nabla^{\prime}e_{t}(\mathbf{\theta})}{e_{t}(\mathbf{\theta
})^{2}}\right\Vert \leq\left(  2H^{3}H_{1}(\mathcal{F}_{t-1})^{2}+H^{2}%
H_{2}(\mathcal{F}_{t-1})\right)  \cdot d\nonumber
\end{gather}
By Assumption 2(D), $\mathbb{E}[V(\mathcal{F}_{t-1})V_{1}(\mathcal{F}%
_{t-1})H_{1}(\mathcal{F}_{t-1})]$, $\mathbb{E}[H_{1}(\mathcal{F}_{t-1}%
)^{2}|Y_{t}|]$, $\mathbb{E}[H_{2}(\mathcal{F}_{t-1})|Y_{t}|]$ $<\infty$.
Therefore, $\mathbb{E}[\mu_{t}(\mathbf{\theta},d)^{(5)}]\leq cd$, where c is a
strictly positive constant.\newline

\textbf{Sixth term:}
\begin{equation}
\mu_{t}^{(6)}(\mathbf{\theta},d)=\sup\limits_{\Vert\mathbf{\tau}%
-\mathbf{\theta}\Vert\leq d}\left\Vert \frac{\nabla^{\prime}e_{t}%
(\mathbf{\tau})}{-e_{t}(\mathbf{\tau})}-\frac{\nabla^{\prime}e_{t}%
(\mathbf{\theta})}{-e_{t}(\mathbf{\theta})}\right\Vert
\end{equation}
By the mean-value theorem,
\begin{equation}
\frac{\nabla^{\prime}e_{t}(\mathbf{\tau})}{-e_{t}(\mathbf{\tau})}-\frac
{\nabla^{\prime}e_{t}(\mathbf{\theta})}{-e_{t}(\mathbf{\theta})}=\left\Vert
\frac{\nabla^{\prime}e_{t}(\mathbf{\tau}^{\ast})\nabla e_{t}(\mathbf{\tau
}^{\ast})}{e_{t}(\mathbf{\tau}^{\ast})^{2}}+\frac{\nabla^{2}e_{t}%
(\mathbf{\tau}^{\ast})}{-e_{t}(\mathbf{\tau}^{\ast})}\right\Vert
\cdot(\mathbf{\tau}-\mathbf{\theta})
\end{equation}%
\begin{equation}
\Rightarrow\sup\limits_{\Vert\mathbf{\tau}-\mathbf{\theta}\Vert\leq
d}\left\Vert \frac{\nabla^{\prime}e_{t}(\mathbf{\tau})}{-e_{t}(\mathbf{\tau}%
)}-\frac{\nabla^{\prime}e_{t}(\mathbf{\theta})}{-e_{t}(\mathbf{\theta}%
)}\right\Vert \leq\left(  H^{2}H_{1}(\mathcal{F}_{t-1})^{2}+H\cdot
H_{2}(\mathcal{F}_{t-1})\right)  \cdot d.
\end{equation}
By Assumption 2(D), $\mathbb{E}[H_{1}(\mathcal{F}_{t-1})^{2}]$, $\mathbb{E}%
[H_{2}(\mathcal{F}_{t-1})]$ $<\infty$. Therefore, $\mathbb{E}[\mu
_{t}(\mathbf{\theta},d)^{(6)}]\leq cd$, where c is a strictly positive constant.

Thus we have shown that $\mu_{t}(\mathbf{\theta},d)\leq\sum_{i=1}^{6}\mu
_{t}(\mathbf{\theta},d)^{(i)}$ with $\mathbb{E}[\mu_{t}(\mathbf{\theta
},d)^{(i)}]\leq cd$, $\forall i=1,2,\cdots,6$, where c is a strictly positive
constant, proving the lemma.
\end{proof}

\bigskip

\begin{lemma}
\label{lemmaN3iii}Under Assumptions 1-2, Assumption N3(iii) of Weiss (1991)
holds:%
\[
\mathbb{E}[\mu_{t}(\mathbf{\theta},d)^{q}]\leq cd,\text{ for }\Vert
\mathbf{\theta}-\mathbf{\theta}^{0}\Vert+d\leq d_{0},\text{ and some }q>2
\]
where $c,~d$ and $d_{0}$ are strictly positive numbers.
\end{lemma}

\bigskip

\begin{proof}
[Proof of Lemma \ref{lemmaN3iii}]In this proof, the strictly positive constant
$c$ and the mean-value expansion term, $\mathbf{\tau}^{\ast}$, can change from
line to line. Pick $d_{0}$ such that for any $\mathbf{\theta}$ that satisfies
$\Vert\mathbf{\theta}-\mathbf{\theta}^{0}\Vert\leq d_{0}$, all the conditions
in Assumption 2(C) and 2(D) hold as well as $e_{t}(\mathbf{\theta})\leq
v_{t}(\mathbf{\theta})\leq0$. Similar to Lemma \ref{lemmaN3ii}, we will
decompose $\mu_{t}(\mathbf{\theta},d)$ into six terms, $\mu_{t}(\mathbf{\theta
},d)^{(i)},$ for $i=1,2,...,6$. By Jensen's inequality, $\mathbb{E}\left[
\mu_{t}(\mathbf{\theta},d)^{q}\right]  \leq6^{q-1}\sum_{i=1}^{6}%
\mathbb{E}[\left(  \mu_{t}(\mathbf{\theta},d)^{\left(  i\right)  }\right)
^{q}],~q>2$. We will show that for some $0<\delta<1$, $\mathbb{E}[\left(
\mu_{t}(\mathbf{\theta},d)^{(i)}\right)  ^{2+\delta}]\leq cd$, $\forall
~i=1,2,\cdots,6$, where c is a strictly positive constant.

\textbf{First term:}%
\begin{equation}
\mu_{t}(\mathbf{\theta},d)^{(1)}=\frac{1}{\alpha}\sup\limits_{\Vert
\mathbf{\tau}-\mathbf{\theta}\Vert\leq d}\left\Vert \frac{\nabla^{\prime}%
v_{t}(\mathbf{\tau})}{-e_{t}(\mathbf{\tau})}\mathbf{1}\{Y_{t}\leq
v_{t}(\mathbf{\tau})\}-\frac{\nabla^{\prime}v_{t}(\mathbf{\theta})}%
{-e_{t}(\mathbf{\theta})}\mathbf{1}\{Y_{t}\leq v_{t}(\mathbf{\theta
})\}\right\Vert
\end{equation}
Set $\mathbf{\tau}_{1}=\arg\min_{\Vert\mathbf{\tau}-\mathbf{\theta}\Vert\leq
d}v_{t}(\mathbf{\tau})$ and $\mathbf{\tau}_{2}={\arg\max}_{\Vert\mathbf{\tau
}-\mathbf{\theta}\Vert\leq d}v_{t}(\mathbf{\tau})$. Following the same
argument as in the proof of Lemma 4, we obtain
\begin{align}
& [\alpha\cdot\mu_{t}(\mathbf{\theta},d)^{(1)}]^{2+\delta}\leq c\cdot\left(
\mathbf{1}\{v_{t}(\mathbf{\tau}_{2})<Y_{t}\leq v_{t}(\mathbf{\theta
})\}+\mathbf{1}\{v_{t}(\mathbf{\tau}_{1})\leq Y_{t}\leq v_{t}(\mathbf{\theta
})\}+\mathbf{1}\{v_{t}(\mathbf{\theta})<Y_{t}\leq v_{t}(\mathbf{\tau}%
_{2})\}\right) \\
& \cdot\left(  \sup\limits_{\Vert\mathbf{\theta}-\mathbf{\theta}^{0}\Vert\leq
d_{0}}\left\Vert \frac{\nabla^{\prime}v_{t}(\mathbf{\theta})}{-e_{t}%
(\mathbf{\theta})}\right\Vert \right)  ^{2+\delta}+\left(  \sup\limits_{\Vert
\mathbf{\tau}-\mathbf{\theta}\Vert\leq d}\left\Vert \frac{\nabla^{\prime}%
v_{t}(\mathbf{\tau})}{-e_{t}(\mathbf{\tau})}-\frac{\nabla^{\prime}%
v_{t}(\mathbf{\theta})}{-e_{t}(\mathbf{\theta})}\right\Vert \right)
^{2+\delta}\nonumber
\end{align}
where
\begin{equation}
\mathbb{E}_{t-1}\left[  \mathbf{1}\{v_{t}(\mathbf{\tau}_{2})<Y_{t}\leq
v_{t}(\mathbf{\theta})\}+\mathbf{1}\{v_{t}(\mathbf{\tau}_{1})\leq Y_{t}\leq
v_{t}(\mathbf{\theta})\}+\mathbf{1}\{v_{t}(\mathbf{\theta})<Y_{t}\leq
v_{t}(\mathbf{\tau}_{2})\}\right]  \leq3KV_{1}(\mathcal{F}_{t-1})d
\end{equation}
and
\begin{equation}
\left(  \sup\limits_{\Vert\mathbf{\theta}-\mathbf{\theta}^{0}\Vert\leq
d}\left\Vert \frac{\nabla^{\prime}v_{t}(\mathbf{\theta})}{-e_{t}%
(\mathbf{\theta})}\right\Vert \right)  ^{2+\delta}\leq\left(  HV_{1}%
(\mathcal{F}_{t-1})\right)  ^{2+\delta}%
\end{equation}
For $\left(  \sup\limits_{\Vert\mathbf{\tau}-\mathbf{\theta}\Vert\leq
d}\left\Vert \frac{\nabla^{\prime}v_{t}(\mathbf{\tau})}{-e_{t}(\mathbf{\tau}%
)}-\frac{\nabla^{\prime}v_{t}(\mathbf{\theta})}{-e_{t}(\mathbf{\theta}%
)}\right\Vert \right)  ^{2+\delta}$, we need to combine the two following two
results:
\begin{align}
\sup\limits_{\Vert\mathbf{\tau}-\mathbf{\theta}\Vert\leq d}\left\Vert
\frac{\nabla^{\prime}v_{t}(\mathbf{\tau})}{-e_{t}(\mathbf{\tau})}-\frac
{\nabla^{\prime}v_{t}(\mathbf{\theta})}{-e_{t}(\mathbf{\theta})}\right\Vert
&  \leq\left(  HV_{2}(\mathcal{F}_{t-1})+H^{2}V_{1}(\mathcal{F}_{t-1}%
)H_{1}(\mathcal{F}_{t-1})\right)  d\\
\left(  \sup\limits_{\Vert\mathbf{\tau}-\mathbf{\theta}\Vert\leq d}\left\Vert
\frac{\nabla^{\prime}v_{t}(\mathbf{\tau})}{-e_{t}(\mathbf{\tau})}-\frac
{\nabla^{\prime}v_{t}(\mathbf{\theta})}{-e_{t}(\mathbf{\theta})}\right\Vert
\right)  ^{1+\delta} &  \leq\left(  2HV_{1}(\mathcal{F}_{t-1})\right)
^{1+\delta}\nonumber
\end{align}
Combining with Assumption 2(D), we thus have $\mathbb{E}[\left(  \mu
_{t}(\mathbf{\theta},d)^{(1)}\right)  ^{2+\delta}]\leq cd$, where c is a
strictly positive constant.\newline

\textbf{Second term: }$\mu_{t}(\mathbf{\theta},d)^{(2)}=\sup\limits_{\Vert
\mathbf{\tau}-\mathbf{\theta}\Vert\leq d}\left\Vert \frac{\nabla^{\prime}%
v_{t}(\mathbf{\tau})}{-e_{t}(\mathbf{\tau})}-\frac{\nabla^{\prime}%
v_{t}(\mathbf{\theta})}{-e_{t}(\mathbf{\theta})}\right\Vert .$ It was shown in
the derivations for the first term that $\mathbb{E}[\left(  \mu_{t}%
(\mathbf{\theta},d)^{(2)}\right)  ^{2+\delta}]\leq cd$, where c is a strictly
positive constant.\newline

\textbf{Third term:}
\begin{equation}
\mu_{t}(\mathbf{\theta},d)^{(3)}=\frac{1}{\alpha}\sup\limits_{\Vert
\mathbf{\tau}-\mathbf{\theta}\Vert\leq d}\left\Vert \frac{v_{t}(\mathbf{\tau
})\nabla^{\prime}e_{t}(\mathbf{\tau})}{e_{t}(\mathbf{\tau})^{2}}%
\mathbf{1}\{Y_{t}\leq v_{t}(\mathbf{\tau})\}-\frac{v_{t}(\mathbf{\theta
})\nabla^{\prime}e_{t}(\mathbf{\theta})}{e_{t}(\mathbf{\theta})^{2}}%
\mathbf{1}\{Y_{t}\leq v_{t}(\mathbf{\theta})\}\right\Vert
\end{equation}
Similar to the first term, $\left(  \alpha\cdot\mu_{t}(\mathbf{\theta
},d)^{(3)}\right)  ^{2+\delta}$ can be bounded by
\begin{align}
c\cdot &  \left(  \mathbf{1}\{v_{t}(\mathbf{\tau}_{2})<Y_{t}\leq
v_{t}(\mathbf{\theta})\}+\mathbf{1}\{v_{t}(\mathbf{\tau}_{1})\leq Y_{t}\leq
v_{t}(\mathbf{\theta})\}+\mathbf{1}\{v_{t}(\mathbf{\theta})<Y_{t}\leq
v_{t}(\mathbf{\tau}_{2})\}\right) \\
&  \cdot\left(  \sup\limits_{\Vert\mathbf{\theta}-\mathbf{\theta}^{0}\Vert\leq
d_{0}}\left\Vert \frac{v_{t}(\mathbf{\theta})\nabla^{\prime}e_{t}%
(\mathbf{\theta})}{e_{t}(\mathbf{\theta})^{2}}\right\Vert \right)  ^{2+\delta
}+\left(  \sup\limits_{\Vert\mathbf{\tau}-\mathbf{\theta}\Vert\leq
d}\left\Vert \frac{v_{t}(\mathbf{\tau})\nabla^{\prime}e_{t}(\mathbf{\tau}%
)}{e_{t}(\mathbf{\tau})^{2}}-\frac{v_{t}(\mathbf{\theta})\nabla^{\prime}%
e_{t}(\mathbf{\theta})}{e_{t}(\mathbf{\theta})^{2}}\right\Vert \right)
^{2+\delta}\nonumber
\end{align}
where
\begin{equation}
\mathbb{E}_{t-1}\left(  \mathbf{1}\{v_{t}(\mathbf{\tau}_{2})<Y_{t}\leq
v_{t}(\mathbf{\theta})\}+\mathbf{1}\{v_{t}(\mathbf{\tau}_{1})\leq Y_{t}\leq
v_{t}(\mathbf{\theta})\}+\mathbf{1}\{v_{t}(\mathbf{\theta})<Y_{t}\leq
v_{t}(\mathbf{\tau}_{2})\}\right)  \leq3KV_{1}(\mathcal{F}_{t-1})d
\end{equation}
and
\begin{equation}
\left(  \sup\limits_{\Vert\mathbf{\theta}-\mathbf{\theta}^{0}\Vert\leq
d}\left\Vert \frac{v_{t}(\mathbf{\theta})\nabla^{\prime}e_{t}(\mathbf{\theta
})}{e_{t}(\mathbf{\theta})^{2}}\right\Vert \right)  ^{2+\delta}\leq\left(
H\cdot H_{1}(\mathcal{F}_{t-1})\right)  ^{2+\delta}%
\end{equation}
As for $\left(  \sup\limits_{\Vert\mathbf{\tau}-\mathbf{\theta}\Vert\leq
d}\left\Vert \frac{v_{t}(\mathbf{\tau})\nabla^{\prime}e_{t}(\mathbf{\tau}%
)}{e_{t}(\mathbf{\tau})^{2}}-\frac{v_{t}(\mathbf{\theta})\nabla^{\prime}%
e_{t}(\mathbf{\theta})}{e_{t}(\mathbf{\theta})^{2}}\right\Vert \right)
^{2+\delta}$, we need to combine the following two results:
\begin{align}
&  \sup\limits_{\Vert\mathbf{\tau}-\mathbf{\theta}\Vert\leq d}\left\Vert
\frac{v_{t}(\mathbf{\tau})\nabla^{\prime}e_{t}(\mathbf{\tau})}{e_{t}%
(\mathbf{\tau})^{2}}-\frac{v_{t}(\mathbf{\theta})\nabla^{\prime}%
e_{t}(\mathbf{\theta})}{e_{t}(\mathbf{\theta})^{2}}\right\Vert \leq\left(
H^{2}V_{1}(\mathcal{F}_{t-1})H_{1}(\mathcal{F}_{t-1})+2H^{2}H_{1}%
(\mathcal{F}_{t-1})^{2}+H\cdot H_{2}(\mathcal{F}_{t-1})\right)  d\\
&  \left(  \sup\limits_{\Vert\mathbf{\tau}-\mathbf{\theta}\Vert\leq
d}\left\Vert \frac{v_{t}(\mathbf{\tau})\nabla^{\prime}e_{t}(\mathbf{\tau}%
)}{e_{t}(\mathbf{\tau})^{2}}-\frac{v_{t}(\mathbf{\theta})\nabla^{\prime}%
e_{t}(\mathbf{\theta})}{e_{t}(\mathbf{\theta})^{2}}\right\Vert \right)
^{1+\delta}\leq\left(  2H\cdot H_{1}(\mathcal{F}_{t-1})\right)  ^{1+\delta}%
\end{align}
Combining with Assumption 2(D), we thus have $\mathbb{E}[\left(  \mu
_{t}(\mathbf{\theta},d)^{(3)}\right)  ^{2+\delta}]\leq cd$, where c is a
strictly positive constant.\newline

\textbf{Fourth term: }$\mu_{t}(\mathbf{\theta},d)^{(4)}=\sup\limits_{\Vert
\mathbf{\tau}-\mathbf{\theta}\Vert\leq d}\left\Vert \frac{v_{t}(\mathbf{\tau
})\nabla^{\prime}e_{t}(\mathbf{\tau})}{e_{t}(\mathbf{\tau})^{2}}-\frac
{v_{t}(\mathbf{\theta})\nabla^{\prime}e_{t}(\mathbf{\theta})}{e_{t}%
(\mathbf{\theta})^{2}}\right\Vert .$ It was shown in the derivations for the
third term that $\mathbb{E}[\left(  \mu_{t}(\mathbf{\theta},d)^{(4)}\right)
^{2+\delta}]\leq cd$, where c is a strictly positive constant.\newline

\textbf{Fifth term:}
\begin{equation}
\mu_{t}(\mathbf{\theta},d)^{(5)}=\frac{1}{\alpha}\sup\limits_{\Vert
\mathbf{\tau}-\mathbf{\theta}\Vert\leq d}\left\Vert \frac{\nabla^{\prime}%
e_{t}(\mathbf{\tau})}{e_{t}(\mathbf{\tau})^{2}}\mathbf{1}\{Y_{t}\leq
v_{t}(\mathbf{\tau})\}Y_{t}-\frac{\nabla^{\prime}e_{t}(\mathbf{\theta})}%
{e_{t}(\mathbf{\theta})^{2}}\mathbf{1}\{Y_{t}\leq v_{t}(\mathbf{\theta
})\}Y_{t}\right\Vert
\end{equation}
Similar to the first and third terms, $\left(  \alpha\cdot\mu_{t}%
(\mathbf{\theta},d)^{(5)}\right)  ^{2+\delta}$ can be bounded by
\begin{align}
c\cdot &  \left(  \mathbf{1}\{v_{t}(\mathbf{\tau}_{2})<Y_{t}\leq
v_{t}(\mathbf{\theta})\}+\mathbf{1}\{v_{t}(\mathbf{\tau}_{1})\leq Y_{t}\leq
v_{t}(\mathbf{\theta})\}+\mathbf{1}\{v_{t}(\mathbf{\theta})<Y_{t}\leq
v_{t}(\mathbf{\tau}_{2})\}\right) \\
&  \cdot|Y_{t}|^{2+\delta}\left(  \sup\limits_{\Vert\mathbf{\theta
}-\mathbf{\theta}^{0}\Vert\leq d_{0}}\left\Vert \frac{\nabla^{\prime}%
e_{t}(\mathbf{\theta})}{e_{t}(\mathbf{\theta})^{2}}\right\Vert \right)
^{2+\delta}+|Y_{t}|^{2+\delta}\left(  \sup\limits_{\Vert\mathbf{\tau
}-\mathbf{\theta}\Vert\leq d}\left\Vert \frac{\nabla^{\prime}e_{t}%
(\mathbf{\tau})}{e_{t}(\mathbf{\tau})^{2}}-\frac{\nabla^{\prime}%
e_{t}(\mathbf{\theta})}{e_{t}(\mathbf{\theta})^{2}}\right\Vert \right)
^{2+\delta}\nonumber
\end{align}
where
\begin{align}
\mathbb{E}_{t-1}[\mathbf{1}\{v_{t}(\mathbf{\tau}_{2})<Y_{t}\leq v_{t}%
(\mathbf{\theta})\}|Y_{t}|^{2+\delta}]= &  \int_{v_{t}(\mathbf{\tau}_{2}%
)}^{v_{t}(\mathbf{\theta})}|y|^{2+\delta}f_{t}(y)dy\leq K|v_{t}(\mathbf{\tau
}_{2})|^{2+\delta}\cdot|v_{t}(\mathbf{\tau}_{2})-v_{t}(\mathbf{\theta})|\\
&  \leq KV(\mathcal{F}_{t-1})^{2+\delta}V_{1}(\mathcal{F}_{t-1})\Vert
\mathbf{\tau}_{2}-\mathbf{\theta}\Vert\leq KV(\mathcal{F}_{t-1})^{2+\delta
}V_{1}(\mathcal{F}_{t-1})d\nonumber
\end{align}
and similarly,
\begin{align}
\mathbb{E}\left[  \mathbf{1}\{v_{t}(\mathbf{\tau}_{1})<Y_{t}\leq
v_{t}(\mathbf{\theta})\}|Y_{t}|^{2+\delta}~|\mathcal{F}_{t-1}\right]   & \leq
KV(\mathcal{F}_{t-1})^{2+\delta}V_{1}(\mathcal{F}_{t-1})d\\
~\text{and \ }\mathbb{E}\left[  \mathbf{1}\{v_{t}(\mathbf{\theta})<Y_{t}\leq
v_{t}(\mathbf{\tau}_{2})\}|Y_{t}|^{2+\delta}~|\mathcal{F}_{t-1}\right]   &
\leq KV(\mathcal{F}_{t-1})^{2+\delta}V_{1}(\mathcal{F}_{t-1})d\nonumber
\end{align}
Further
\begin{equation}
\sup\limits_{\Vert\mathbf{\theta}-\mathbf{\theta}^{0}\Vert\leq d}\left\Vert
\frac{\nabla^{\prime}e_{t}(\mathbf{\theta})}{e_{t}(\mathbf{\theta})^{2}%
}\right\Vert \leq H^{2}H_{1}(\mathcal{F}_{t-1})
\end{equation}
As for $\left(  \sup\limits_{\Vert\mathbf{\tau}-\mathbf{\theta}\Vert\leq
d}\left\Vert \frac{\nabla^{\prime}e_{t}(\mathbf{\tau})}{e_{t}(\mathbf{\tau
})^{2}}-\frac{\nabla^{\prime}e_{t}(\mathbf{\theta})}{e_{t}(\mathbf{\theta
})^{2}}\right\Vert \right)  ^{2+\delta}$, we also need to combine the
following two results:
\begin{align}
\sup\limits_{\Vert\mathbf{\tau}-\mathbf{\theta}\Vert\leq d}\left\Vert
\frac{\nabla^{\prime}e_{t}(\mathbf{\tau})}{e_{t}(\mathbf{\tau})^{2}}%
-\frac{\nabla^{\prime}e_{t}(\mathbf{\theta})}{e_{t}(\mathbf{\theta})^{2}%
}\right\Vert  &  \leq\left(  2H^{3}H_{1}(\mathcal{F}_{t-1})^{2}+H^{2}%
H_{2}(\mathcal{F}_{t-1})\right)  d\\
\left(  \sup\limits_{\Vert\mathbf{\theta}-\mathbf{\theta}^{0}\Vert\leq
d}\left\Vert \frac{\nabla^{\prime}e_{t}(\mathbf{\theta})}{e_{t}(\mathbf{\theta
})^{2}}\right\Vert \right)  ^{1+\delta} &  \leq\left(  2H^{2}H_{1}%
(\mathcal{F}_{t-1})\right)  ^{1+\delta}\nonumber
\end{align}
Combining with Assumption 2(D), we thus have $\mathbb{E}[\left(  \mu
_{t}(\mathbf{\theta}d)^{(5)}\right)  ^{2+\delta}]\leq cd$, where c is a
strictly positive constant.\newline

\textbf{Sixth term:}
\begin{equation}
\mu_{t}^{(6)}(\mathbf{\theta},d)=\sup\limits_{\Vert\mathbf{\tau}%
-\mathbf{\theta}\Vert\leq d}\left\Vert \frac{\nabla^{\prime}e_{t}%
(\mathbf{\tau})}{-e_{t}(\mathbf{\tau})}-\frac{\nabla^{\prime}e_{t}%
(\mathbf{\theta})}{-e_{t}(\mathbf{\theta})}\right\Vert
\end{equation}
We have
\begin{align}
\sup\limits_{\Vert\mathbf{\tau}-\mathbf{\theta}\Vert\leq d}\left\Vert
\frac{\nabla^{\prime}e_{t}(\mathbf{\tau})}{-e_{t}(\mathbf{\tau})}-\frac
{\nabla^{\prime}e_{t}(\mathbf{\theta})}{-e_{t}(\mathbf{\theta})}\right\Vert
&  \leq\left(  H^{2}H_{1}(\mathcal{F}_{t-1})^{2}+HH_{2}(\mathcal{F}%
_{t-1})\right)  d\\
\left(  \sup\limits_{\Vert\mathbf{\tau}-\mathbf{\theta}\Vert\leq d}\left\Vert
\frac{\nabla^{\prime}e_{t}(\mathbf{\tau})}{-e_{t}(\mathbf{\tau})}-\frac
{\nabla^{\prime}e_{t}(\mathbf{\theta})}{-e_{t}(\mathbf{\theta})}\right\Vert
\right)  ^{1+\delta} &  \leq\left(  2HH_{1}(\mathcal{F}_{t-1})\right)
^{1+\delta}\nonumber
\end{align}
Combining with Assumption 2(D), we thus have $\mathbb{E}[\left(  \mu
_{t}(\mathbf{\theta},d)^{(6)}\right)  ^{2+\delta}]\leq cd$, where c is a
strictly positive constant. Thus $\mathbb{E}[\mu_{t}(\mathbf{\theta}%
,d)^{(i)}]^{2+\delta}\leq cd$, $\forall i=1,2,\cdots,6,$ proving the lemma.
\end{proof}

\bigskip

\begin{lemma}
\label{lemmaN4} Under Assumptions 1-2, Assumption N4 of Weiss (1991) holds:
$E\Vert g_{t}(\mathbf{\theta}^{0})\Vert^{2}\leq M$, for all t and some $M>0$.
\end{lemma}

\bigskip

\begin{proof}
[Proof of Lemma \ref{lemmaN4}]%
\begin{align}
\mathbb{E}\Vert g_{t}(\mathbf{\theta}^{0})\Vert^{2}\leq &  4\left\{
\mathbb{E}\left\Vert \frac{\nabla^{\prime}v_{t}(\mathbf{\theta}^{0})}%
{-e_{t}(\mathbf{\theta}^{0})}\left(  \frac{1}{\alpha}\mathbf{1}\left\{
Y_{t}\leq v_{t}(\mathbf{\theta}^{0})\right\}  -1\right)  \right\Vert
^{2}\right. \\
&  +\mathbb{E}\left\Vert \frac{v_{t}(\mathbf{\theta}^{0})\nabla^{\prime}%
e_{t}(\mathbf{\theta}^{0})}{e_{t}(\mathbf{\theta}^{0})^{2}}\left(  \frac
{1}{\alpha}\mathbf{1}\left\{  Y_{t}\leq v_{t}(\mathbf{\theta}^{0})\right\}
-1\right)  \right\Vert ^{2}+\mathbb{E}\left\Vert \frac{\nabla^{\prime}%
e_{t}(\mathbf{\theta}^{0})}{e_{t}(\mathbf{\theta}^{0})}\right\Vert
^{2}\nonumber\\
&  \left.  +\mathbb{E}\left\Vert \frac{\nabla^{\prime}e_{t}(\mathbf{\theta
}^{0})}{e_{t}(\mathbf{\theta}^{0})^{2}}\frac{1}{\alpha}\mathbf{1}\left\{
Y_{t}\leq v_{t}(\mathbf{\theta}^{0})\right\}  Y_{t}\right\Vert ^{2}\right\}
\nonumber\\
\leq &  4\left\{  \mathbb{E}\left[  \left(  \frac{1}{\alpha}+1\right)
^{2}H^{2}V_{1}(\mathcal{F}_{t-1})^{2}\right]  +\mathbb{E}\left[  \left(
\frac{1}{\alpha}+1\right)  ^{2}H^{2}H_{1}(\mathcal{F}_{t-1})^{2}\right]
\right. \nonumber\\
&  \left.  +\frac{1}{\alpha^{2}}H^{4}\mathbb{E}[H_{1}(\mathcal{F}_{t-1}%
)^{2}Y_{t}^{2}]+\mathbb{E}\left[  H^{2}H_{1}(\mathcal{F}_{t-1})^{2}\right]
\right\} \nonumber\\
\leq &  M\nonumber
\end{align}
where $M$ is some finite constant, and the second inequality follows using
Assumptions 2(C) and 2(D).
\end{proof}

\bigskip

\bigskip

\pagebreak%

\def\baselinestretch{1.0}\small\normalsize

{\LARGE Appendix SA.2:\ Additional tables}

\bigskip

\begin{center}
\bigskip

\textbf{Table S1: Finite-sample performance of (Q)MLE}

\bigskip%

\begin{tabular}
[c]{cccccccc}\hline
& \multicolumn{3}{c}{$T=2500$} &  & \multicolumn{3}{c}{$T=5000$}%
\\\cline{2-4}\cline{6-8}
& $\omega$ & $\beta$ & $\gamma$ &  & $\omega$ & $\beta$ & $\gamma
$\\\cline{2-4}\cline{6-8}
&  &  &  &  &  &  & \\
\multicolumn{4}{l}{\textbf{Panel A: N(0,1) innovations}} &
\multicolumn{1}{l}{} &  &  & \\
\multicolumn{1}{l}{True} & 0.050 & 0.950 & 0.050 &  & 0.050 & 0.950 & 0.050\\
\multicolumn{1}{l}{Median} & 0.053 & 0.897 & 0.050 &  & 0.051 & 0.899 &
0.050\\
\multicolumn{1}{l}{Avg bias} & 0.011 & (0.011) & 0.000 &  & 0.005 & (0.005) &
0.000\\
\multicolumn{1}{l}{St dev} & 0.056 & 0.064 & 0.013 &  & 0.023 & 0.029 &
0.009\\
\multicolumn{1}{l}{Coverage} & 0.936 & 0.930 & 0.928 &  & 0.936 & 0.933 &
0.937\\\hline
&  &  &  &  &  &  & \\
\multicolumn{4}{l}{\textbf{Panel B: Skew t (5,-0.5) innovations}} &
\multicolumn{1}{l}{} &  &  & \\
\multicolumn{1}{l}{True} & 0.050 & 0.950 & 0.050 &  & 0.050 & 0.950 & 0.050\\
\multicolumn{1}{l}{Median} & 0.052 & 0.895 & 0.049 &  & 0.052 & 0.897 &
0.050\\
\multicolumn{1}{l}{Avg bias} & 0.017 & (0.023) & 0.005 &  & 0.006 & (0.008) &
0.002\\
\multicolumn{1}{l}{St dev} & 0.077 & 0.095 & 0.028 &  & 0.026 & 0.037 &
0.017\\
\multicolumn{1}{l}{Coverage} & 0.899 & 0.907 & 0.897 &  & 0.913 & 0.907 &
0.903\\\hline
\end{tabular}

\end{center}

\bigskip

\textit{Notes:} This table presents results from 1000 replications of the
estimation of the parameters of a GARCH(1,1)\ model, using the Normal
likelihood. In Panel A the innovations are standard Normal, and so estimation
is then ML. In Panel B the innovations are standardized skew $t,$ and so
estimation is QML. Details are described in\ Section \ref{sSIMULATION} of the
main paper. The top row of each panel presents the true values of the
parameters. The second, third, and fourth rows present the median estimated
parameters, the average bias, and the standard deviation (across simulations)
of the estimated parameters. The last row of each panel presents the coverage
rates for 95\% confidence intervals constructed using estimated standard errors.

\pagebreak

\begin{center}
\textbf{Table S2: Simulation results for Normal innovations, }

\textbf{estimation by CAViaR}

\bigskip%

\begin{tabular}
[c]{rccccccc}\hline
& \multicolumn{3}{c}{$T=2500$} &  & \multicolumn{3}{c}{$T=5000$}%
\\\cline{2-4}\cline{6-8}
& $\beta$ & $\gamma$ & $a_{\alpha}$ &  & $\beta$ & $\gamma$ & $a_{\alpha}%
$\\\cline{2-4}\cline{6-8}
&  &  &  &  &  &  & \\
& \multicolumn{7}{c}{$\alpha=0.01$}\\
\multicolumn{1}{l}{True} & 0.900 & 0.050 & -2.326 &  & 0.900 & 0.050 &
-2.326\\
\multicolumn{1}{l}{Median} & 0.901 & 0.048 & -2.275 &  & 0.899 & 0.048 &
-2.347\\
\multicolumn{1}{l}{Avg bias} & -0.017 & 0.012 & -0.120 &  & -0.011 & 0.006 &
-0.095\\
\multicolumn{1}{l}{St dev} & 0.079 & 0.066 & 0.957 &  & 0.051 & 0.034 &
0.718\\
\multicolumn{1}{l}{Coverage} & 0.881 & 0.874 & 0.907 &  & 0.892 & 0.886 &
0.905\\\hline
\multicolumn{1}{l}{} &  &  &  &  &  &  & \\
\multicolumn{1}{l}{} & \multicolumn{7}{c}{$\alpha=0.025$}\\
\multicolumn{1}{l}{True} & 0.900 & 0.050 & -1.960 &  & 0.900 & 0.050 &
-1.960\\
\multicolumn{1}{l}{Median} & 0.898 & 0.047 & -1.953 &  & 0.896 & 0.047 &
-2.009\\
\multicolumn{1}{l}{Avg bias} & -0.018 & 0.005 & -0.136 &  & -0.012 & 0.002 &
-0.110\\
\multicolumn{1}{l}{St dev} & 0.068 & 0.038 & 0.728 &  & 0.052 & 0.023 &
0.566\\
\multicolumn{1}{l}{Coverage} & 0.906 & 0.879 & 0.934 &  & 0.913 & 0.892 &
0.918\\\hline
\multicolumn{1}{l}{} &  &  &  &  &  &  & \\
\multicolumn{1}{l}{} & \multicolumn{7}{c}{$\alpha=0.05$}\\
\multicolumn{1}{l}{True} & 0.900 & 0.050 & -1.645 &  & 0.900 & 0.050 &
-1.645\\
\multicolumn{1}{l}{Median} & 0.901 & 0.047 & -1.639 &  & 0.899 & 0.049 &
-1.667\\
\multicolumn{1}{l}{Avg bias} & -0.014 & 0.005 & -0.085 &  & -0.009 & 0.002 &
-0.070\\
\multicolumn{1}{l}{St dev} & 0.068 & 0.037 & 0.597 &  & 0.045 & 0.023 &
0.436\\
\multicolumn{1}{l}{Coverage} & 0.909 & 0.884 & 0.930 &  & 0.918 & 0.900 &
0.935\\\hline
\multicolumn{1}{l}{} &  &  &  &  &  &  & \\
\multicolumn{1}{l}{} & \multicolumn{7}{c}{$\alpha=0.10$}\\
\multicolumn{1}{l}{True} & 0.900 & 0.050 & -1.282 &  & 0.900 & 0.050 &
-1.282\\
\multicolumn{1}{l}{Median} & 0.898 & 0.047 & -1.291 &  & 0.898 & 0.048 &
-1.289\\
\multicolumn{1}{l}{Avg bias} & -0.016 & 0.006 & -0.076 &  & -0.010 & 0.003 &
-0.055\\
\multicolumn{1}{l}{St dev} & 0.069 & 0.041 & 0.482 &  & 0.047 & 0.025 &
0.364\\
\multicolumn{1}{l}{Coverage} & 0.916 & 0.883 & 0.933 &  & 0.921 & 0.896 &
0.937\\\hline
\multicolumn{1}{l}{} &  &  &  &  &  &  & \\
\multicolumn{1}{l}{} & \multicolumn{7}{c}{$\alpha=0.20$}\\
\multicolumn{1}{l}{True} & 0.900 & 0.050 & -0.842 &  & 0.900 & 0.050 &
-0.842\\
\multicolumn{1}{l}{Median} & 0.898 & 0.048 & -0.848 &  & 0.899 & 0.048 &
-0.840\\
\multicolumn{1}{l}{Avg bias} & -0.023 & 0.022 & -0.058 &  & -0.016 & 0.007 &
-0.049\\
\multicolumn{1}{l}{St dev} & 0.091 & 0.107 & 0.391 &  & 0.063 & 0.044 &
0.304\\
\multicolumn{1}{l}{Coverage} & 0.914 & 0.876 & 0.931 &  & 0.929 & 0.901 &
0.940\\\hline
\end{tabular}

\end{center}

\textit{Notes:} This table presents results from 1000 replications of the
estimation of VaR from a GARCH(1,1)\ DGP with standard Normal innovations.
Details are described in\ Section \ref{sSIMULATION} of the main paper. The top
row of each panel presents the true values of the parameters. The second,
third, and fourth rows present the median estimated parameters, the average
bias, and the standard deviation (across simulations) of the estimated
parameters. The last row of each panel presents the coverage rates for 95\%
confidence intervals constructed using estimated standard errors.

\bigskip

\bigskip\pagebreak

\begin{center}
\textbf{Table S3: Simulation results for skew t innovations, }

\textbf{estimation by CAViaR}

\bigskip%

\begin{tabular}
[c]{rccccccc}\hline
& \multicolumn{3}{c}{$T=2500$} &  & \multicolumn{3}{c}{$T=5000$}%
\\\cline{2-4}\cline{6-8}
& $\beta$ & $\gamma$ & $a_{\alpha}$ &  & $\beta$ & $\gamma$ & $a_{\alpha}%
$\\\cline{2-4}\cline{6-8}
&  &  &  &  &  &  & \\
& \multicolumn{7}{c}{$\alpha=0.01$}\\
\multicolumn{1}{l}{True} & 0.900 & 0.050 & -3.290 &  & 0.900 & 0.050 &
\multicolumn{1}{r}{-3.290}\\
\multicolumn{1}{l}{Median} & 0.898 & 0.045 & -3.272 &  & 0.899 & 0.045 &
\multicolumn{1}{r}{-3.306}\\
\multicolumn{1}{l}{Avg bias} & -0.041 & 0.022 & -0.355 &  & -0.027 & 0.008 &
\multicolumn{1}{r}{-0.306}\\
\multicolumn{1}{l}{St dev} & 0.142 & 0.097 & 1.928 &  & 0.103 & 0.044 &
\multicolumn{1}{r}{1.546}\\
\multicolumn{1}{l}{Coverage} & 0.771 & 0.805 & 0.827 &  & 0.785 & 0.808 &
\multicolumn{1}{r}{0.823}\\\hline
\multicolumn{1}{l}{} &  &  &  &  &  &  & \\
\multicolumn{1}{l}{} & \multicolumn{7}{c}{$\alpha=0.025$}\\
\multicolumn{1}{l}{True} & 0.900 & 0.050 & -2.408 &  & 0.900 & 0.050 &
\multicolumn{1}{r}{-2.408}\\
\multicolumn{1}{l}{Median} & 0.899 & 0.047 & -2.371 &  & 0.898 & 0.049 &
\multicolumn{1}{r}{-2.414}\\
\multicolumn{1}{l}{Avg bias} & -0.026 & 0.012 & -0.190 &  & -0.016 & 0.004 &
\multicolumn{1}{r}{-0.144}\\
\multicolumn{1}{l}{St dev} & 0.103 & 0.067 & 1.135 &  & 0.070 & 0.033 &
\multicolumn{1}{r}{0.862}\\
\multicolumn{1}{l}{Coverage} & 0.832 & 0.841 & 0.877 &  & 0.830 & 0.862 &
\multicolumn{1}{r}{0.859}\\\hline
\multicolumn{1}{l}{} &  &  &  &  &  &  & \\
\multicolumn{1}{l}{} & \multicolumn{7}{c}{$\alpha=0.05$}\\
\multicolumn{1}{l}{True} & 0.900 & 0.050 & -1.800 &  & 0.900 & 0.050 &
\multicolumn{1}{r}{-1.800}\\
\multicolumn{1}{l}{Median} & 0.899 & 0.047 & -1.780 &  & 0.899 & 0.049 &
\multicolumn{1}{r}{-1.792}\\
\multicolumn{1}{l}{Avg bias} & -0.023 & 0.008 & -0.146 &  & -0.013 & 0.004 &
\multicolumn{1}{r}{-0.087}\\
\multicolumn{1}{l}{St dev} & 0.092 & 0.060 & 0.782 &  & 0.057 & 0.028 &
\multicolumn{1}{r}{0.563}\\
\multicolumn{1}{l}{Coverage} & 0.863 & 0.861 & 0.892 &  & 0.883 & 0.871 &
\multicolumn{1}{r}{0.890}\\\hline
\multicolumn{1}{l}{} &  &  &  &  &  &  & \\
\multicolumn{1}{l}{} & \multicolumn{7}{c}{$\alpha=0.10$}\\
\multicolumn{1}{l}{True} & 0.900 & 0.050 & -1.223 &  & 0.900 & 0.050 &
\multicolumn{1}{r}{-1.223}\\
\multicolumn{1}{l}{Median} & 0.900 & 0.049 & -1.205 &  & 0.900 & 0.049 &
\multicolumn{1}{r}{-1.217}\\
\multicolumn{1}{l}{Avg bias} & -0.019 & 0.008 & -0.074 &  & -0.010 & 0.004 &
\multicolumn{1}{r}{-0.043}\\
\multicolumn{1}{l}{St dev} & 0.080 & 0.050 & 0.495 &  & 0.050 & 0.027 &
\multicolumn{1}{r}{0.356}\\
\multicolumn{1}{l}{Coverage} & 0.895 & 0.892 & 0.919 &  & 0.892 & 0.905 &
\multicolumn{1}{r}{0.910}\\\hline
\multicolumn{1}{l}{} &  &  &  &  &  &  & \\
\multicolumn{1}{l}{} & \multicolumn{7}{c}{$\alpha=0.20$}\\
\multicolumn{1}{l}{True} & 0.900 & 0.050 & -0.652 &  & 0.900 & 0.050 &
\multicolumn{1}{r}{-0.652}\\
\multicolumn{1}{l}{Median} & 0.903 & 0.051 & -0.619 &  & 0.902 & 0.051 &
\multicolumn{1}{r}{-0.636}\\
\multicolumn{1}{l}{Avg bias} & -0.027 & 0.026 & -0.035 &  & -0.016 & 0.009 &
\multicolumn{1}{r}{-0.028}\\
\multicolumn{1}{l}{St dev} & 0.122 & 0.109 & 0.353 &  & 0.084 & 0.042 &
\multicolumn{1}{r}{0.271}\\
\multicolumn{1}{l}{Coverage} & 0.867 & 0.887 & 0.897 &  & 0.890 & 0.889 &
\multicolumn{1}{r}{0.916}\\\hline
\end{tabular}

\end{center}

\textit{Notes:} This table presents results from 1000 replications of the
estimation of VaR from a GARCH(1,1)\ DGP with skew $t$ innovations. Details
are described in\ Section \ref{sSIMULATION} of the main paper. The top row of
each panel presents the true values of the parameters. The second, third, and
fourth rows present the median estimated parameters, the average bias, and the
standard deviation (across simulations) of the estimated parameters. The last
row of each panel presents the coverage rates for 95\% confidence intervals
constructed using estimated standard errors.\pagebreak

\begin{center}
\textbf{Table S4: Diebold-Mariano t-statistics on average out-of-sample loss
differences}

\textbf{for the DJIA, NIKKEI and\ FTSE100 (alpha=0.05)}

\bigskip%

\begin{tabular}
[c]{rcccccccccc}\hline
& RW125 & RW250 & RW500 & G-N & G-Skt & G-EDF & FZ-2F & FZ-1F & G-FZ &
Hybrid\\\cline{2-11}
&  &  &  &  &  &  &  &  &  & \\
\multicolumn{1}{l}{} & \multicolumn{10}{c}{\textbf{Panel A: DJIA}}\\
\multicolumn{1}{l}{RW125} &  & -2.547 & -4.234 & 3.189 & 3.812 & 3.793 &
3.305 & 4.368 & 3.475 & 3.853\\
\multicolumn{1}{l}{RW250} & 2.547 &  & -4.145 & 4.028 & 4.579 & 4.595 &
4.601 & 5.358 & 4.529 & 4.598\\
\multicolumn{1}{l}{RW500} & 4.234 & 4.145 &  & 5.328 & 5.802 & 5.825 & 5.903 &
6.553 & 5.868 & 5.901\\\hline
\multicolumn{1}{l}{G-N} & -3.189 & -4.028 & -5.328 &  & 3.312 & 2.773 &
0.818 & 2.171 & 1.811 & 1.769\\
\multicolumn{1}{l}{G-Skt} & -3.812 & -4.579 & -5.802 & -3.312 &  & 0.391 &
-0.143 & 1.430 & -0.160 & -0.022\\
\multicolumn{1}{l}{G-EDF} & -3.793 & -4.595 & -5.825 & -2.773 & -0.391 &  &
-0.187 & 1.434 & -0.367 & -0.174\\\hline
\multicolumn{1}{l}{FZ-2F} & -3.305 & -4.601 & -5.903 & -0.818 & 0.143 &
0.187 &  & 0.142 & 0.028 & 1.179\\
\multicolumn{1}{l}{FZ-1F} & -4.022 & -4.738 & -5.750 & -0.965 & 0.004 &
0.038 & -0.142 &  & -1.597 & -1.402\\
\multicolumn{1}{l}{G-FZ} & -3.475 & -4.529 & -5.868 & -1.811 & 0.275 & 0.367 &
-0.028 & 1.597 &  & 0.086\\
\multicolumn{1}{l}{Hybrid} & -3.826 & -4.506 & -5.710 & -2.426 & -1.425 &
-1.430 & -1.179 & 1.402 & -0.086 & \\\hline
\multicolumn{1}{l}{} &  &  &  &  &  &  &  &  &  & \\
\multicolumn{1}{l}{} & \multicolumn{10}{c}{\textbf{Panel B: NIKKEI}}\\
\multicolumn{1}{l}{RW125} &  & -0.245 & -1.181 & 4.015 & 3.993 & 4.030 &
3.804 & 3.464 & 3.933 & 4.166\\
\multicolumn{1}{l}{RW250} & 0.245 &  & -1.418 & 4.460 & 4.473 & 4.519 &
4.075 & 3.887 & 4.437 & 4.661\\
\multicolumn{1}{l}{RW500} & 1.181 & 1.418 &  & 4.412 & 4.433 & 4.476 & 4.348 &
3.965 & 4.431 & 4.582\\\hline
\multicolumn{1}{l}{G-N} & -4.015 & -4.460 & -4.412 &  & 1.180 & 2.177 &
-1.877 & -1.271 & 1.251 & 0.419\\
\multicolumn{1}{l}{G-Skt} & -3.993 & -4.473 & -4.433 & -1.180 &  & 1.831 &
-1.931 & -1.389 & 0.613 & 0.255\\
\multicolumn{1}{l}{G-EDF} & -4.030 & -4.519 & -4.476 & -2.177 & -1.831 &  &
-2.031 & -1.520 & -0.901 & 0.075\\\hline
\multicolumn{1}{l}{FZ-2F} & -3.804 & -4.075 & -4.348 & 1.877 & 1.931 & 2.031 &
& 1.135 & 1.950 & 2.495\\
\multicolumn{1}{l}{FZ-1F} & -3.250 & -3.629 & -3.659 & 1.195 & 1.319 & 1.463 &
-1.135 &  & 1.426 & 2.741\\
\multicolumn{1}{l}{G-FZ} & -3.933 & -4.437 & -4.431 & -1.251 & -0.640 &
0.901 & -1.950 & -1.426 &  & 0.171\\
\multicolumn{1}{l}{Hybrid} & -3.998 & -4.500 & -4.364 & -0.565 & -0.410 &
-0.226 & -2.495 & -2.741 & -0.171 & \\\hline
\end{tabular}

\bigskip
\end{center}

\textit{Table continued on next page.}\pagebreak

\begin{center}
\textbf{Table S4: Diebold-Mariano t-statistics on average out-of-sample loss
differences}

\textbf{for the DJIA, NIKKEI and\ FTSE100 (alpha=0.05)}

\bigskip%

\begin{tabular}
[c]{rcccccccccc}\hline
& RW125 & RW250 & RW500 & G-N & G-Skt & G-EDF & FZ-2F & FZ-1F & G-FZ &
Hybrid\\\cline{2-11}
&  &  &  &  &  &  &  &  &  & \\
\multicolumn{1}{l}{} & \multicolumn{10}{c}{\textbf{Panel C: FTSE}}\\
\multicolumn{1}{l}{RW125} &  & -2.707 & -3.955 & 3.723 & 3.988 & 3.846 &
-3.329 & 3.623 & 3.651 & 3.398\\
\multicolumn{1}{l}{RW250} & 2.707 &  & -3.245 & 4.784 & 5.036 & 4.898 &
-2.188 & 4.724 & 4.764 & 4.486\\
\multicolumn{1}{l}{RW500} & 3.955 & 3.245 &  & 5.470 & 5.685 & 5.570 &
-0.834 & 5.479 & 5.513 & 5.321\\\hline
\multicolumn{1}{l}{G-N} & -3.723 & -4.784 & -5.470 &  & 4.494 & 3.434 &
-6.805 & 0.406 & 1.526 & 0.796\\
\multicolumn{1}{l}{G-Skt} & -3.988 & -5.036 & -5.685 & -4.494 &  & -4.167 &
-6.898 & -0.347 & -0.671 & 0.172\\
\multicolumn{1}{l}{G-EDF} & -3.846 & -4.898 & -5.570 & -3.434 & 4.167 &  &
-6.847 & 0.065 & 0.569 & 0.519\\\hline
\multicolumn{1}{l}{FZ-2F} & 3.329 & 2.188 & 0.834 & 6.805 & 6.898 & 6.847 &  &
6.187 & 6.920 & 7.263\\
\multicolumn{1}{l}{FZ-1F} & -3.831 & -4.853 & -5.382 & -0.247 & 0.355 &
0.020 & -6.187 &  & 0.125 & 0.760\\
\multicolumn{1}{l}{G-FZ} & -3.651 & -4.764 & -5.513 & -1.526 & 0.710 &
-0.569 & -6.920 & -0.125 &  & 0.417\\
\multicolumn{1}{l}{Hybrid} & -3.208 & -4.242 & -5.027 & -0.643 & 0.008 &
-0.355 & -7.263 & -0.760 & -0.417 & \\\hline
\end{tabular}

\end{center}

\bigskip

\textit{Notes:} This table presents $t$-statistics from Diebold-Mariano tests
comparing the average losses, using the FZ0 loss function, over the
out-of-sample period from January 2000 to December 2016, for ten different
forecasting models. A positive value indicates that the row model has higher
average loss than the column model. Values greater than 1.96 in absolute value
indicate that the average loss difference is significantly different from zero
at the 95\% confidence level. Values along the main diagonal are all
identically zero and are omitted for interpretability. The first three rows
correspond to rolling window forecasts, the next three rows correspond
to\ GARCH forecasts based on different models for the standardized residuals,
and the last four rows correspond to models introduced in Section
\ref{sMODELS} of the main paper.\pagebreak\bigskip%

\begin{landscape}%

\begin{center}
\bigskip

\textbf{Table S5: Out-of-sample average losses and goodness-of-fit tests
(alpha=0.025)}

\bigskip%

\begin{tabular}
[c]{rccccllllllllll}\hline
&  &  &  &  &  &  &  &  &  &  &  &  &  & \\
& \multicolumn{4}{c}{\textit{Average loss}} &  &
\multicolumn{4}{c}{\textit{GoF p-values: VaR}} &  &
\multicolumn{4}{c}{\textit{GoF p-values:\ ES}}\\\cline{2-5}\cline{7-10}%
\cline{12-15}
& \textbf{S\&P} & \textbf{DJIA} & \textbf{NIK} & \textbf{FTSE} &  &
\multicolumn{1}{c}{\textbf{S\&P}} & \multicolumn{1}{c}{\textbf{DJIA}} &
\multicolumn{1}{c}{\textbf{NIK}} & \multicolumn{1}{c}{\textbf{FTSE}} &  &
\multicolumn{1}{c}{\textbf{S\&P}} & \multicolumn{1}{c}{\textbf{DJIA}} &
\multicolumn{1}{c}{\textbf{NIK}} & \multicolumn{1}{c}{\textbf{FTSE}%
}\\\cline{2-15}
&  &  &  &  &  &  &  &  &  &  &  &  &  & \\
\multicolumn{1}{l}{RW-125} & 1.119 & 1.088 & 1.525 & 1.166 &  &
\multicolumn{1}{c}{0.022} & \multicolumn{1}{c}{0.003} &
\multicolumn{1}{c}{0.000} & \multicolumn{1}{c}{0.000} &  &
\multicolumn{1}{c}{0.009} & \multicolumn{1}{c}{0.004} &
\multicolumn{1}{c}{0.001} & \multicolumn{1}{c}{0.001}\\
\multicolumn{1}{l}{RW-250} & 1.164 & 1.117 & 1.525 & 1.209 &  &
\multicolumn{1}{c}{0.005} & \multicolumn{1}{c}{0.007} &
\multicolumn{1}{c}{0.002} & \multicolumn{1}{c}{0.000} &  &
\multicolumn{1}{c}{0.023} & \multicolumn{1}{c}{0.039} &
\multicolumn{1}{c}{0.010} & \multicolumn{1}{c}{0.005}\\
\multicolumn{1}{l}{RW-500} & 1.245 & 1.187 & 1.561 & 1.294 &  &
\multicolumn{1}{c}{0.001} & \multicolumn{1}{c}{0.000} &
\multicolumn{1}{c}{0.004} & \multicolumn{1}{c}{0.000} &  &
\multicolumn{1}{c}{0.019} & \multicolumn{1}{c}{0.011} &
\multicolumn{1}{c}{0.007} & \multicolumn{1}{c}{0.000}\\
\multicolumn{1}{l}{GCH-N} & 1.089 & 1.016 & 1.341 & 1.053 &  &
\multicolumn{1}{c}{0.000} & \multicolumn{1}{c}{0.002} &
\multicolumn{1}{c}{\textbf{0.172}} & \multicolumn{1}{c}{0.000} &  &
\multicolumn{1}{c}{0.000} & \multicolumn{1}{c}{0.000} &
\multicolumn{1}{c}{0.048} & \multicolumn{1}{c}{0.000}\\
\multicolumn{1}{l}{GCH-Skt} & 1.043 & 0.975 & \textbf{1.328} & \textbf{1.025}
&  & \multicolumn{1}{c}{0.005} & \multicolumn{1}{c}{\textit{0.057}} &
\multicolumn{1}{c}{\textbf{0.789}} & \multicolumn{1}{c}{0.000} &  &
\multicolumn{1}{c}{0.010} & \multicolumn{1}{c}{\textit{0.076}} &
\multicolumn{1}{c}{\textbf{0.736}} & \multicolumn{1}{c}{0.001}\\
\multicolumn{1}{l}{GCH-EDF} & \textit{1.028} & \textit{0.970} & 1.329 &
1.040 &  & \multicolumn{1}{c}{\textbf{0.164}} &
\multicolumn{1}{c}{\textbf{0.149}} & \multicolumn{1}{c}{\textbf{0.789}} &
\multicolumn{1}{c}{0.000} &  & \multicolumn{1}{c}{\textbf{0.237}} &
\multicolumn{1}{c}{\textbf{0.379}} & \multicolumn{1}{c}{\textbf{0.588}} &
\multicolumn{1}{c}{0.000}\\
\multicolumn{1}{l}{FZ-2F} & 1.041 & 0.998 & 4.037 & 2.445 &  &
\multicolumn{1}{c}{0.000} & \multicolumn{1}{c}{\textbf{0.117}} &
\multicolumn{1}{c}{0.000} & \multicolumn{1}{c}{0.000} &  &
\multicolumn{1}{c}{0.001} & \multicolumn{1}{c}{\textbf{0.341}} &
\multicolumn{1}{c}{0.000} & \multicolumn{1}{c}{0.000}\\
\multicolumn{1}{l}{FZ-1F} & 1.032 & 1.004 & 1.415 & \textit{1.039} &  &
\multicolumn{1}{c}{\textbf{0.343}} & \multicolumn{1}{c}{\textbf{0.314}} &
\multicolumn{1}{c}{0.043} & \multicolumn{1}{c}{0.028} &  &
\multicolumn{1}{c}{\textbf{0.393}} & \multicolumn{1}{c}{\textbf{0.334}} &
\multicolumn{1}{c}{0.047} & \multicolumn{1}{c}{0.045}\\
\multicolumn{1}{l}{GCH-FZ} & \textbf{1.020} & \textbf{0.951} & \textit{1.328}
& 1.059 &  & \multicolumn{1}{c}{\textit{0.095}} &
\multicolumn{1}{c}{\textbf{0.358}} & \multicolumn{1}{c}{\textbf{0.608}} &
\multicolumn{1}{c}{0.000} &  & \multicolumn{1}{c}{\textbf{0.188}} &
\multicolumn{1}{c}{\textbf{0.419}} & \multicolumn{1}{c}{\textbf{0.473}} &
\multicolumn{1}{c}{0.000}\\
\multicolumn{1}{l}{Hybrid} & 1.034 & 1.018 & 1.341 & 1.056 &  &
\multicolumn{1}{c}{0.002} & \multicolumn{1}{c}{\textit{0.082}} &
\multicolumn{1}{c}{\textbf{0.700}} & \multicolumn{1}{c}{0.000} &  &
\multicolumn{1}{c}{0.007} & \multicolumn{1}{c}{\textit{0.064}} &
\multicolumn{1}{c}{\textbf{0.629}} & \multicolumn{1}{c}{0.000}\\\hline
\end{tabular}

\end{center}

\bigskip

\textit{Notes:} The left panel of this table presents the average losses,
using the FZ0 loss function, for four daily equity return series, over the
out-of-sample period from January 2000 to December 2016, for ten different
forecasting models. The lowest average loss in each column is highlighted in
bold, the second-lowest is highlighted in italics. The first three rows
correspond to rolling window forecasts, the next three rows correspond
to\ GARCH forecasts based on different models for the standardized residuals,
and the last four rows correspond to models introduced in Section
\ref{sMODELS} of the main paper. The middle and right panels of this table
present $p$-values from goodness-of-fit tests of the VaR and\ ES forecasts
respectively. Values that are greater than 0.10 (indicating no evidence
against optimality at the 0.10 level) are in bold, and values between 0.05 and
0.10 are in italics.

\bigskip%

\end{landscape}%
\pagebreak

\begin{center}
\bigskip

\textbf{Table S6: Diebold-Mariano t-statistics on average out-of-sample loss
differences}

\textbf{for the \textbf{S\&P 500, }DJIA, NIKKEI and\ FTSE100 (alpha=0.025)}

\bigskip%

\begin{tabular}
[c]{rcccccccccc}\hline
& RW125 & RW250 & RW500 & G-N & G-Skt & G-EDF & FZ-2F & FZ-1F & G-FZ &
Hybrid\\\cline{2-11}
&  &  &  &  &  &  &  &  &  & \\
\multicolumn{1}{l}{} & \multicolumn{10}{c}{\textbf{Panel A: S\&P 500}}\\
\multicolumn{1}{l}{RW125} &  & -2.035 & -3.587 & 1.100 & 2.728 & 3.125 &
1.972 & 3.599 & 3.212 & 2.642\\
\multicolumn{1}{l}{RW250} & 2.035 &  & -3.454 & 1.901 & 3.112 & 3.472 &
2.637 & 4.240 & 3.613 & 3.447\\
\multicolumn{1}{l}{RW500} & 3.587 & 3.454 &  & 3.283 & 4.388 & 4.731 & 3.966 &
5.605 & 4.879 & 4.968\\\hline
\multicolumn{1}{l}{G-N} & -1.100 & -1.901 & -3.283 &  & 4.241 & 3.522 &
1.645 & 2.346 & 3.835 & 1.963\\
\multicolumn{1}{l}{G-Skt} & -2.728 & -3.112 & -4.388 & -4.241 &  & 2.393 &
0.093 & 0.738 & 2.850 & -0.447\\
\multicolumn{1}{l}{G-EDF} & -3.125 & -3.472 & -4.731 & -3.522 & -2.393 &  &
-0.595 & -0.198 & 1.482 & -1.500\\\hline
\multicolumn{1}{l}{FZ-2F} & -1.972 & -2.637 & -3.966 & -1.645 & -0.093 &
0.595 &  & 0.348 & 1.111 & 0.368\\
\multicolumn{1}{l}{FZ-1F} & -3.599 & -4.240 & -5.605 & -2.346 & -0.738 &
0.198 & -0.348 &  & 0.739 & -1.406\\
\multicolumn{1}{l}{G-FZ} & -3.212 & -3.613 & -4.879 & -3.835 & -2.850 &
-1.482 & -1.111 & -0.739 &  & -2.300\\
\multicolumn{1}{l}{Hybrid} & -2.642 & -3.447 & -4.968 & -1.963 & 0.447 &
1.500 & -0.368 & 1.406 & 2.300 & \\\hline
\multicolumn{1}{l}{} &  &  &  &  &  &  &  &  &  & \\
\multicolumn{1}{l}{} & \multicolumn{10}{c}{\textbf{Panel B: DJIA}}\\
\multicolumn{1}{l}{RW125} &  & -1.066 & -2.722 & 2.676 & 3.902 & 3.879 &
3.194 & 3.906 & 3.637 & 1.945\\
\multicolumn{1}{l}{RW250} & 1.066 &  & -3.065 & 2.754 & 3.852 & 3.900 &
4.102 & 4.343 & 3.744 & 2.249\\
\multicolumn{1}{l}{RW500} & 2.722 & 3.065 &  & 3.968 & 5.053 & 5.131 & 5.529 &
5.764 & 5.026 & 3.661\\\hline
\multicolumn{1}{l}{G-N} & -2.676 & -2.754 & -3.968 &  & 3.430 & 3.009 &
0.703 & 1.313 & 2.775 & -0.970\\
\multicolumn{1}{l}{G-Skt} & -3.902 & -3.852 & -5.053 & -3.430 &  & 1.390 &
-1.211 & -0.958 & 1.722 & -3.640\\
\multicolumn{1}{l}{G-EDF} & -3.879 & -3.900 & -5.131 & -3.009 & -1.390 &  &
-1.553 & -1.265 & 1.620 & -3.563\\\hline
\multicolumn{1}{l}{FZ-2F} & -3.194 & -4.102 & -5.529 & -0.703 & 1.211 &
1.553 &  & -0.310 & 1.962 & -0.744\\
\multicolumn{1}{l}{FZ-1F} & -3.906 & -4.343 & -5.764 & -1.313 & 0.958 &
1.265 & 0.310 &  & 1.736 & -1.835\\
\multicolumn{1}{l}{G-FZ} & -3.637 & -3.744 & -5.026 & -2.775 & -1.722 &
-1.620 & -1.962 & -1.736 &  & -3.364\\
\multicolumn{1}{l}{Hybrid} & -1.945 & -2.249 & -3.661 & 0.970 & 3.640 &
3.563 & 0.744 & 1.835 & 3.364 & \\\hline
\end{tabular}

\bigskip
\end{center}

\textit{Table continued on next page.}\pagebreak

\begin{center}
\textbf{Table S6: Diebold-Mariano t-statistics on average out-of-sample loss
differences}

\textbf{for the \textbf{S\&P 500, }DJIA, NIKKEI and\ FTSE100 (alpha=0.025),
continued}

\bigskip%

\begin{tabular}
[c]{rcccccccccc}\hline
& RW125 & RW250 & RW500 & G-N & G-Skt & G-EDF & FZ-2F & FZ-1F & G-FZ &
Hybrid\\\cline{2-11}
&  &  &  &  &  &  &  &  &  & \\
\multicolumn{1}{l}{} & \multicolumn{10}{c}{\textbf{Panel C: NIKKEI}}\\
\multicolumn{1}{l}{RW125} &  & 0.011 & -0.977 & 4.223 & 4.166 & 4.211 &
-16.674 & 2.677 & 4.148 & 4.052\\
\multicolumn{1}{l}{RW250} & -0.011 &  & -1.773 & 4.499 & 4.568 & 4.592 &
-16.612 & 2.767 & 4.542 & 4.466\\
\multicolumn{1}{l}{RW500} & 0.977 & 1.773 &  & 4.536 & 4.628 & 4.638 &
-17.116 & 3.019 & 4.602 & 4.620\\\hline
\multicolumn{1}{l}{G-N} & -4.223 & -4.499 & -4.536 &  & 1.896 & 2.089 &
-16.040 & -2.765 & 2.042 & -0.126\\
\multicolumn{1}{l}{G-Skt} & -4.166 & -4.568 & -4.628 & -1.896 &  & -0.864 &
-15.803 & -3.078 & -0.283 & -0.828\\
\multicolumn{1}{l}{G-EDF} & -4.211 & -4.592 & -4.638 & -2.089 & 0.864 &  &
-15.847 & -3.072 & 0.415 & -0.764\\\hline
\multicolumn{1}{l}{FZ-2F} & 16.674 & 16.612 & 17.116 & 16.040 & 15.803 &
15.847 &  & 15.323 & 15.834 & 15.784\\
\multicolumn{1}{l}{FZ-1F} & -2.677 & -2.767 & -3.019 & 2.765 & 3.078 & 3.072 &
-15.323 &  & 3.035 & 3.650\\
\multicolumn{1}{l}{G-FZ} & -4.148 & -4.542 & -4.602 & -2.042 & 0.283 &
-0.415 & -15.834 & -3.035 &  & -0.785\\
\multicolumn{1}{l}{Hybrid} & -4.052 & -4.466 & -4.620 & 0.126 & 0.828 &
0.764 & -15.784 & -3.650 & 0.785 & \\\hline
\multicolumn{1}{l}{} &  &  &  &  &  &  &  &  &  & \\
\multicolumn{1}{l}{} & \multicolumn{10}{c}{\textbf{Panel D: FTSE}}\\
\multicolumn{1}{l}{RW125} &  & -1.754 & -3.623 & 3.329 & 3.989 & 3.639 &
-4.888 & 3.253 & 2.818 & 2.375\\
\multicolumn{1}{l}{RW250} & 1.754 &  & -3.406 & 4.122 & 4.786 & 4.435 &
-4.800 & 4.139 & 3.716 & 3.257\\
\multicolumn{1}{l}{RW500} & 3.623 & 3.406 &  & 5.066 & 5.638 & 5.339 &
-4.613 & 5.355 & 4.809 & 4.533\\\hline
\multicolumn{1}{l}{G-N} & -3.329 & -4.122 & -5.066 &  & 4.696 & 3.860 &
-5.167 & -0.306 & -0.827 & -2.199\\
\multicolumn{1}{l}{G-Skt} & -3.989 & -4.786 & -5.638 & -4.696 &  & -4.658 &
-5.230 & -2.170 & -3.470 & -3.828\\
\multicolumn{1}{l}{G-EDF} & -3.639 & -4.435 & -5.339 & -3.860 & 4.658 &  &
-5.191 & -1.163 & -2.332 & -3.130\\\hline
\multicolumn{1}{l}{FZ-2F} & 4.888 & 4.800 & 4.613 & 5.167 & 5.230 & 5.191 &  &
5.173 & 5.154 & 5.110\\
\multicolumn{1}{l}{FZ-1F} & -3.253 & -4.139 & -5.355 & 0.306 & 2.170 & 1.163 &
-5.173 &  & -0.147 & -1.526\\
\multicolumn{1}{l}{G-FZ} & -2.818 & -3.716 & -4.809 & 0.827 & 3.470 & 2.332 &
-5.154 & 0.147 &  & -2.015\\
\multicolumn{1}{l}{Hybrid} & -2.375 & -3.257 & -4.533 & 2.199 & 3.828 &
3.130 & -5.110 & 1.526 & 2.015 & \\\hline
\end{tabular}

\end{center}

\bigskip

\textit{Notes:} This table presents $t$-statistics from Diebold-Mariano tests
comparing the average losses, using the FZ0 loss function, over the
out-of-sample period from January 2000 to December 2016, for ten different
forecasting models. A positive value indicates that the row model has higher
average loss than the column model. Values greater than 1.96 in absolute value
indicate that the average loss difference is significantly different from zero
at the 95\% confidence level. Values along the main diagonal are all
identically zero and are omitted for interpretability. The first three rows
correspond to rolling window forecasts, the next three rows correspond
to\ GARCH forecasts based on different models for the standardized residuals,
and the last four rows correspond to models introduced in Section
\ref{sMODELS} of the main paper.


\begin{thebibliography}{99}                                                                                               %
\bibitem {}Andersen, T.G., Bollerslev, T., Christoffersen, P., Diebold, F.X.,
2006. Volatility and Correlation Forecasting, in (ed.s) G. Elliott, C.W.J.
Granger, and A. Timmermann, \textit{Handbook of Economic Forecasting}, Vol. 1.
Elsevier, Oxford.

\bibitem {}Andrews, D.W.K., 1987, Consistency in nonlinear econometric models:
ageneric uniform law of large numbers, \textit{Econometrica}, 55, 1465--1471.

\bibitem {}Artzner, P., F. Delbaen, J.M. Eber and D. Heath, 1999, Coherent
measures of risk, \textit{Mathematical Finance}, 9, 203-228.

\bibitem {}Barendse, S., 2017, Interquantile Expectation Regression, Tinbergen
Institute Discussion Paper, TI 2017-034/III.

\bibitem {}Basel Committee on\ Banking Supervision, 2010, Basel III:\ A Global
Regulatory Framework for More Resiliant Banks and Banking Systems, Bank for
International Settlements. \texttt{http://www.bis.org/publ/bcbs189.pdf}

\bibitem {}Bollerslev, T., 1986, Generalized Autoregressive Conditional
Heteroskedasticity, \textit{Journal of Econometrics}, 31, 307-327.

\bibitem {}Bollerslev, T. and\ J.M. Wooldridge, 1992, Quasi-Maximum Likelihood
Estimation and Inference in Dynamic Models with Time Varying Covariances,
\textit{Econometric Reviews}, 11(2), 143-172.

\bibitem {}Cai, Z. and X. Wang, 2008, Nonparametric estimation of conditional
VaR and expected shortfall, \textit{Journal of Econometrics}, 147, 120-130.

\bibitem {}Creal, D.D., S.J. Koopman, and A. Lucas, 2013, Generalized
Autoregressive Score Models with Applications, \textit{Journal of Applied
Econometrics}, 28(5), 777-795.

\bibitem {}Creal, D.D., S.J.Koopman, A. Lucas, and M. Zamojski, 2015,
Generalized Autoregressive Method of Moments, Tinbergen Institute Discussion
Paper, TI 2015-138/III.

\bibitem {}Diebold, F.X. and R.S. Mariano, 1995. Comparing predictive
accuracy, \textit{Journal of Business \& Economic Statistics,}13(3), 253--263.

\bibitem {}Dimitriadis, T. and S. Bayer, 2017, A Joint Quantile and\ Expected
Shortfall Regression\ Framework, working paper, available at
\texttt{arXiv:1704.02213v1}.

\bibitem {}Du, Z. and J.C. Escanciano, 2017, Backtesting Expected
Shortfall:\ Accounting for Tail Risk,\textit{\ Management Science}, 63(4), 940-958.

\bibitem {}Engle, R.F. and S. Manganelli, 2004a, CAViaR: Conditional
Autoregressive Value at Risk by Regression Quantiles, \textit{Journal of
Business \& Economic Statistics}, 22, 367-381.

\bibitem {}Engle, R.F. and S. Manganelli, 2004b, A Comparison of Value-at-Risk
Models in Finance, in Giorgio Szego (ed.) \textit{Risk Measures for the 21st
Century}, Wiley.

\bibitem {}Engle, R.F. and J.R. Russell, 1998, Autoregressive Conditional
Duration: A New Model for Irregularly Spaced Transaction Data,
\textit{Econometrica}, 66, 1127-1162.

\bibitem {}Fissler, T., 2017, \textit{On Higher Order Elicitability and Some
Limit\ Theorems on the Poisson and\ Weiner Space}, PhD thesis, University of Bern.

\bibitem {}Fissler, T., and\ J. F. Ziegel, 2016, Higher order elicitability
and Osband's principle, \textit{Annals of Statistics}, 44(4), 1680-1707.

\bibitem {}Francq,\ C. and J.-M. Zako\"{\i}an, 2015, Risk-parameter estimation
in volatility models, \textit{Journal of Econometrics}, 184, 158-173.

\bibitem {}Gerlach, R. and C.W.S. Chen, 2015, Bayesian Expected Shortfall
Forecasting Incorporating the Intraday Range, \textit{Journal of Financial
Econometrics}, 14(1), 128-158.

\bibitem {}Gneiting, T., 2011, Making and Evaluating Point Forecasts,
\textit{Journal of the American Statistical Association}, 106(494), 746-762.

\bibitem {}Gsch\"{o}pf, P., W.K. H\"{a}rdle, and A. Mihoci, Tail\ Event Risk
Expectile based Shortfall, SFB 649 Discussion Paper 2015-047.

\bibitem {}Hansen, B.E., 1994, Autoregressive Conditional Density Estimation,
\textit{International Economic Review}, 35(3), 705-730.

\bibitem {}Harvey, A.C., 2013, \textit{Dynamic Models for Volatility and Heavy
Tails}, Econometric Society Monograph 52, Cambridge University Press, Cambridge.

\bibitem {}Huber, P.J., 1967, The behavior of maximum likelihood estimates
under nonstandard conditions, in (ed.s) L.M. Le Cam and\ J. Neyman
\textit{Proceedings of the Fifth Berkeley Symposium on Mathematical Statistics
and Probability}, Vol. 1, University of California Press, Berkeley.

\bibitem {}Komunjer, I., 2005, Quasi-Maximum Likelihood Estimation for
Conditional Quantiles, \textit{Journal of Econometrics}, 128(1), 137-164.

\bibitem {}Komunjer, I., 2013, Quantile Prediction, in (ed.s) G. Elliott, and
A. Timmermann, \textit{Handbook of Economic Forecasting}, Vol. 2. Elsevier, Oxford.

\bibitem {}Koopman, S.J., A. Lucas and M. Scharth, Predicting Time-Varying
Parameters with Parameter Driven and Observation-Driven Models, \textit{Review
of Economics and Statistics}, 98(1), 97-110.

\bibitem {}Newey, W.K. and D. McFadden, 1994, Large Sample Estimation and
Hypothesis Testing, in R.F. Engle and D.L. McFadden (eds.) \textit{Handbook of
Econometrics},\textit{\ }Vol. IV, Elsevier.

\bibitem {}Newey, W.K. and J.L. Powell, 1987, Asymmetric least squares
estimation and testing, \textit{Econometrica}, 55(4), 819-847.

\bibitem {}Nolde, N. and J. F. Ziegel, 2017, Elicitability and backtesting:
Perspectives for banking regulation, \textit{Annals of Applied Statistics}, forthcoming.

\bibitem {}Patton,\ A.J., 2011, Volatility Forecast Comparison using Imperfect
Volatility Proxies, \textit{Journal of Econometrics}, 160(1), 246-256.

\bibitem {}Patton, A.J., 2016, Comparing Possibly Misspecified Forecasts,
working paper, Duke University.

\bibitem {}Patton, A.J. and K. Sheppard, 2009, Evaluating Volatility and
Correlation Forecasts, in T.G. Andersen, R.A. Davis, J.-P. Kreiss and T.
Mikosch (eds.) \textit{Handbook of Financial Time Series}, Springer Verlag.

\bibitem {}P\"{o}tscher, B.M. and I.R. Prucha, 1989, A uniform law of large
numbers for dependent and heterogeneous data processes, \textit{Econometrica},
57, 675--683.

\bibitem {}Taylor, J.W., 2008, Estimating Value-at-Risk and\ Expected
Shortfall using Expectiles, \textit{Journal of Financial Econometrics}, 231-252.

\bibitem {}Taylor, J.W., 2017, Forecasting Value at Risk and\ Expected
Shortfall using a Semiparametric Approach Based on the Asymmetric Laplace
Distribution, \textit{Journal of Business \&\ Economic Statistics}, forthcoming.

\bibitem {}Weiss, A.A., 1991, Estimating Nonlinear Dynamic Models Using Least
Absolute Error Estimation, \textit{Econometric Theory}, 7(1), 46-68.

\bibitem {}White, H. 1994, \textit{Estimation, Inference and Specification
Analysis}, Econometric Society Monographs No. 22, Cambridge University\ Press.

\bibitem {}Zhu, D. and J.W. Galbraith, 2011, Modeling and forecasting expected
shortfall with the generalized asymmetric Student-$t$ and asymmetric
exponential power distributions, \textit{Journal of Empirical Finance}, 18, 765-778.

\bibitem {}Zwingmann T. and H. Holzmann, 2016, Asymptotics for Expected
Shortfall, working paper, available at \texttt{arXiv:1611.07222}.
\end{thebibliography}
\end{document}